\documentclass[a4paper,onecolumn,10pt,UKenglish,accepted=2025-01-17]{quantumarticle}
\pdfoutput=1
\usepackage{nameref}
\usepackage[UKenglish]{babel}
\usepackage[vmargin=40mm,inner=25mm,outer=30mm,headheight=14pt]{geometry}
\usepackage{amsthm,amsmath,amssymb}
\usepackage{physics}
\usepackage{relsize}

\usepackage{xcolor}
\newtheorem{theorem}{Theorem}
\newtheorem{corollary}{Corollary}

\newtheorem{lemma}{Lemma}
\newtheorem{conjecture}{Conjecture}

\newtheorem{remark}{Remark}

\newtheorem{proposition}{Proposition}

\DeclareMathOperator{\Var}{Var}

\newcommand{\R}{{\mathbb{R}}}

\renewcommand{\P}{\mathbb{P}}

\newcommand{\mc}{\mathcal}

\newcommand{\bs}{\boldsymbol}

\newcommand{\p}{\varphi}
\newcommand{\Pd}{\mathbb{P}}
\newcommand{\Qd}{\mathbb{Q}}
\newcommand{\pp}{\mathbf{p}}
\newcommand{\de}{d}
\newcommand{\vertiii}[1]{{\left\vert\kern-0.25ex\left\vert\kern-0.25ex\left\vert
#1 
    \right\vert\kern-0.25ex\right\vert\kern-0.25ex\right\vert}}

\newcommand{\blue}[1]{{\color{black}{#1}}}  
 
\newcommand{\eps}{\varepsilon}

\newcommand{\bigO}{\mc{O}}
\newcommand{\bigOt}{\tilde{\mc{O}}}
\newcommand{\bigOlower}{\Omega}

\newcommand{\fid}[1]{\sqrt{F}\big( #1\big)}
\newcommand{\id}{\mathrm{Id}}

\newcommand{\bigqm}[1][1]{\text{\larger[#1]{\textbf{?}}}}
\DeclareMathOperator*{\argmax}{arg\,max}
\DeclareMathOperator*{\argmin}{arg\,min}

\newcommand\extrafootertext[1]{%
    \bgroup
    \renewcommand\thefootnote{\fnsymbol{footnote}}%
    \renewcommand\thempfootnote{\fnsymbol{mpfootnote}}%
    \footnotetext[0]{#1}%
    \egroup
}

\title{Multi-Armed Bandits and Quantum Channel Oracles} 
\author{Simon Buchholz, Jonas M. K\"ubler, Bernhard Sch\"olkopf}
\affiliation{Max Planck Institute for Intelligent Systems, T\"ubingen, Germany}
\email{\{sbuchholz, jmkuebler, bs\}@tue.mpg.de}
\thanks{JMK is now with Amazon.}

\begin{document}

\maketitle

\begin{abstract}
Multi-armed bandits are one of the theoretical pillars of reinforcement learning. 
Recently, the investigation of quantum algorithms for multi-armed bandit problems was started, and it was found that a quadratic speedup (in query complexity) is possible when 
the arms and the randomness of the rewards of the arms can be queried in superposition. 
Here we introduce further bandit models where we only have limited access to the randomness of the rewards, but we can still query the arms in superposition. We show that then the query complexity is the same as for classical algorithms.
This  generalizes the prior result that no speedup is possible 
for unstructured search when the oracle has positive failure probability. 
\end{abstract}

\section{Introduction}
Quantum computing is a model of computation that is based on quantum properties of matter.
By using superposition and entanglement, it offers potentially large speedups when compared to classical algorithms. 
For some problems exponential speedups have been shown under the assumption of widely believed hardness results for classical computing, the two most prominent examples being Shor's algorithm for factoring integers \cite{shor1997factoring} and the HHL algorithm for sampling from the solution of sparse linear equations \cite{harrow2009quantum}. On the other hand, it was shown that for many problems only polynomial speedups are possible, in particular Grover's algorithm \cite{grover} for the unstructured search problem only offers a quadratic speedup and no greater improvement is possible.

Recently, quantum machine learning has emerged as one potential area of application for quantum computers \cite{quantumML}. It was suggested to use quantum computers for linear algebra subroutines 
but also complete implementations of well-known classical algorithms for supervised learning were 
designed, e.g.,  quantum support vector machines \cite{QSVM}, quantum principal component analysis \cite{Lloyd:2014ud},
and recommender systems \cite{quantum_recommendation}.
There has also been some work on unsupervised learning and reinforcement learning
\cite{clustering, clustering_neurips, quantumRL}.
For a recent review, we refer to \cite{review_ciliberto}.
 
\blue{
 Multi-armed bandit problems are a class  of  problems from machine learning 
 where a decision maker in each time step selects an action from a fixed set of options and then receives a corresponding reward.  The goal of the decision maker is to identify the best possible action. 
The two main problems considered in this field are regret minimization
where the goal is to maximize the total reward received, and best arm identification where we want to identify the best arm in the fewest rounds possible. Bandit problems have a long history, starting among others from the works \cite{thompson,robins52,robins85} and a long list of works established different algorithms and (almost) optimal results for a wide range of settings, and we refer to the next section for additional details and the literature for a complete overview \cite{bubeck_lecture,lattimore2020bandit}.

Recently, the investigation of the best arm identification problem  with quantum computers was initialized in \cite{multi_armed_quantum}. They show that a quadratic speedup compared to the classical algorithm is possible and optimal in their setting. 
Their implementation of the bandit arms is, however, not directly comparable to the classical setting because they assume that the algorithm has control over the internal 
randomness of the bandit arms. In particular, the rewards are obtained by the action of a unitary oracle.

In this work, we will discuss when this multi-armed bandit model is applicable, and we introduce further models of multi-armed bandits that are more suitable in different settings. 
Those models differ in the degree of control that we have over the internal randomness of the rewards of each arm. When we have no access to the randomness of the rewards, a pull of the arms can no longer be described by a unitary map, but instead they are modelled by a non-unitary quantum channel.
This provides a link to the channel discrimination problem, however, there the focus has been on rather different types of quantum channels that are either generic or more related to the transmission of information  \cite{channel_disc}. 
}
Here we show that in this more restricted setting no polynomial speedup compared to classical algorithms is possible, even though the bandits can be queried in superposition. 
Previously, it was shown in \cite{faulty_grover} that no speedup for unstructured search is possible if 
the oracle has a fixed probability of error. Our 
results are an extension to a substantially more general setting. 
Thus, this work highlights that quantum speedups can be impeded by (small) amounts of classical randomness present in the quantum channel used in the algorithm, underscoring again that having parts of the computation routine without error correction poses challenges.
From a technical side, we connect classical methods from quantum information theory with 
coupling techniques from probability theory.

The rest of this paper is structured as follows: In the next section, we introduce different settings of quantum bandits and give an overview of their query complexities.
Then we discuss the main strategies and ingredients 
used in the proofs of the main results in Section~\ref{sec:overview}. 
Full proofs of our results are delegated to the appendices.

\section{Setting and Main Results}\label{sec:setting}
We now discuss the setting for our main results. 
In this section, we review the relevant background on (classical) bandits and then state our main results for quantum bandits.

\subsection{Classical bandits}\label{sec:classical_bandits_setting}
\blue{
To set the scene, we briefly review the multi-armed bandit problem,
which should also provide sufficient background for readers from the quantum side who are less familiar with the setting.
In the   multi-armed bandit problem, an agent can choose in every round $t$ an arm (an action from a finite set) $i_t\in \{1,\ldots, N\}$  and receives a probabilistic reward $r_t$ depending on the chosen action.  
There are a vast number of variants and generalizations of this problem in the classical setting, but here we focus on the simplest model where
the number of arms is finite, and the rewards upon choosing 
arm $i$ are drawn independently from some distribution $\nu_i$ where $\nu_i\in \mc{P}$
is in some class of probability distributions known to the agent. 
Assume that $\nu_i$ has mean $\mu_i$ and we denote the highest mean reward by $\mu^\ast= \max_i \mu_i$.

The goal in bandit theory is generally to identify the action that leads to the highest mean reward $\mu^\ast$. There are two main problems of interest. The first is regret minimization, where we try to minimize the regret which is defined after $T$ rounds with rewards $(r_t)_{1\leq t\leq T}$ by
\begin{align}
    \rho_T = T\mu^\ast - \sum_{t=1}^T r_t.
\end{align}
Equivalently, we want to maximize the total reward received, and the regret measures the reward difference to  the optimal course of actions.  Often, a key goal is to achieve sublinear regret, i.e., $\rho_T\in o(T)$ which means that asymptotically the average reward received converges to the mean of the best arm.
In the context of multi-armed bandits, a much finer understanding of optimal regret bounds is known, namely it was shown
that the uniform confidence bound algorithm is asymptotically optimal and corresponding 
bounds were derived (see, e.g., \cite{robins85,lai,auer}).
The regret minimization problem captures the trade-off between exploration, i.e., finding good, unexplored options, and exploitation, i.e., choosing favorable options, which is essential for reinforcement learning.

A second problem investigated in the  bandit setting is the best arm identification problem. Here the goal is to identify the best arm with a fixed confidence (i.e., up to some given error probability at most $\delta$) in the fewest rounds possible. Another variant is to 
maximize the probability of finding the best arm for a fixed time horizon.
Note that the setting  only features the exploration component and not the exploitation component of the regret minimization problem.
Even though the setting of best arm identification is simpler than regret minimization, it is often a challenging problem and in some settings optimal results were found later than for the regret minimization problem.
In this paper, we focus on the best arm identification problem because this readily generalizes to the quantum setting, while it is not directly obvious how the regret should be defined and interpreted in the quantum setting.

Let us now first briefly review the main results and provide some intuition for the best arm identification problem in the classical multi-armed bandit setting introduced above.}
We restrict our attention to Bernoulli distributions $\nu_i\sim \mathrm{Ber}(p_i)$
which have mass $p$ on 1 and mass $1-p$ on $0$ because this can be easily generalized to the quantum setting where the reward can then be encoded by a single qubit. In this case, the reward distribution is fully characterized by the means $p_i$ of the arms.
For concreteness, we fix a mean reward vectors $\pp=(p_0,\ldots, p_{N-1})\in\R^{N}$, indicating that
arm $i$ has mean reward  $p_i$. Note that unconventional indexing is used here to ensure consistency with our setup later on.
We usually assume that the rewards are ordered, i.e.,  $p_0> p_1\geq \ldots \geq p_{N-1}$. 
We will use the shorthand $\Delta_i=p_0-p_i$ for the difference in mean rewards between
the best arm and arm $i$.
We define the important quantity 
\begin{align}
H(\pp)= \sum_{i\geq 1} \frac{1}{(p_0-p_i)^2}=\sum_{i\geq 1} \Delta_i^{-2}.
\end{align}
It can be shown that $H(\pp)$
governs the optimal query complexity to identify the best
arm (up to logarithmic terms) in a fixed confidence setting.
Moreover, this is optimal when  $p_i \in [\eta,1-\eta]$ for all $i$ and some $\eta>0$. 
The following well-known theorem provides a formal statement of those results.
\begin{theorem}[Thm.~2 in \cite{dar_best_arm}, Thm. 5 in \cite{mannor_best_arm}]\label{th:classical_bandit_result}
Consider an algorithm that identifies the best arm of a multi-armed bandit with probability at least $1-\delta$
for Bernoulli distributed rewards with reward vector $\pp\in [0,1]^N$
such that $\pp_i\in [\eta,1-\eta]$.
Then this algorithm requires $\bigOlower(H(\pp))$ rounds in the worst case.
On the other hand, there exists such an algorithm 
requiring $\bigOt(H(\pp))$ steps.
\end{theorem}
Let us emphasize that identification of the best arm is not easier if we know the rewards up to a permutation, in fact, the following result
holds.
\begin{theorem}[Thm.~4 in \cite{bubeck_best_arm}]\label{th:classical_bandit_result2}
Let $\pp\in [\eta, 1-\eta]^N$ be a reward vector.
Then any algorithm that identifies the best arm for   
Bernoulli distributed rewards with means
$\pp'$ where $\pp'$ is any permutation of $\pp$ requires at least $\bigOlower(H(\pp))$ rounds and such an algorithm exists (up to logarithmic terms).
\end{theorem}
Let us add several remarks to these results. 
There is a long list of results improving upon the two results above by deriving bounds on the logarithmic correction, considering more general reward distributions, and analyzing various algorithms, see e.g.,
\cite{gabillon,garivier,jamieson,chen}.

Let us explain for the readers not so familiar with the literature on 
bandit problems the intuition underlying 
the results mentioned above. The key statistical problem is essentially to decide which of two arms has a higher reward.
Assume that we know the mean reward $p_0$ of the best arm exactly.
Let us first investigate how many pulls $t_i$ of the arm with true mean reward $p_i$ are sufficient to verify with high probability that $p_i<p_1$.
We can bound the difference between the empirical mean $\hat{p}_i$
and $p_i$ using Hoeffding's inequality (see Lemma~\ref{le:hoeffding}) by
\begin{align}
    \P(\hat{p}_i-p_i\geq (p_0-p_i))\leq e^{-2t_i(p_0-p_i)^2}.
\end{align}
This shows that after
\begin{align}\label{eq_condition_t_i}
    t_i \gtrsim (p_0-p_i)^{-2} =\Delta_i^{-2}
\end{align}
pulls we can rule out with high probability that an arm with true reward $p_i$ has
a return higher than $p_0$. On the other hand, $t_i$ of order $\Delta_i^{-2}$ is necessary  when 
$p_i\in [\eta, 1-\eta]$ for some constant $\eta>0$.  
Indeed, by the central limit theorem
$\hat{p}_i-p_i$ will be typically of order $\sqrt{\Var(\mathrm{Ber}(p_i)/t_i}$ and the variance is lower bounded by $\eta(1-\eta)$ if $p_i\in [\eta, 1-\eta]$.
We conclude
that \eqref{eq_condition_t_i} is necessary to conclude that $p_i<p_0$ with high probability.
Applying this reasoning to all arms $i$ suggests that $H(\pp)$  governs the query complexity of the best arm identification problem. Note that here we ignored the fact that the highest reward $p^\ast$ is unknown, and algorithms overcome this problem by using the optimism principle.

Let us emphasize that the argument for the lower bound  uses crucially that all rewards 
are away from $1$ and $0$. Otherwise, we cannot consider the variance term $p_i(1-p_i)$
 to be a fixed constant in the analysis.
To illustrate this, we consider the particular case that  $p_1=p>0$ and $p_i=0$ for $i>1$.
Then $H(\pp)=N/p^2$ but
only $\bigO(N/p)$ pulls are required to identify the best arm.
To see this, note that it takes with high probability $\bigO(p^{-1})$ pulls to get a single success on an arm with
$\mathrm{Ber}(p)$ distributed rewards.
This is the setting considered in \cite{faulty_grover} and we will come back to it in Section~\ref{sec:channel_oracles}.

\subsection{Quantum bandits}\label{sec:quantum_bandits}
\blue{
In this section, we  discuss how the classical bandit problem can be generalized to a quantum setting. This builds upon the recent works  \cite{casale, multi_armed_quantum} that studied the best arm identification problem
for quantum bandits.
Here, we want to review and extend prior results and definitions in the literature and in particular highlight that different assumptions for the available oracles 
are reasonable in different settings. 

Let us first explain the general implementation of decision problems in the quantum setting,
while we refer to standard textbooks (e.g., \cite{NC00}) for a general introduction to quantum computing. 
Typically, we assume that we are given black box access to an oracle $O$ that is a unitary map on some Hilbert space from a finite set $\mc{O}$ of unitary maps
 and our goal is to identify which oracle we are given using the fewest possible number of invocations of $O$.
To give a concrete example, we consider the arguably most famous example of unstructured search where we are given one of the oracles $O^i$ that mark an element $i\in \{1,\ldots, N\}$, i.e., they act by 
\begin{align}
    O^i \ket{j}\ket{c}= \ket{j}\ket{c \oplus \delta_{ij}}
\end{align}
where $\oplus$ denotes addition modulo 2 and $\ket{j}$ is a basis
of an $N$ dimensional Hilbert space and $\ket{c}$ is a qubit basis state. 
Our goal is now to identify the marked element $i$ by a quantum algorithm. This means that starting from some initial state $\ket{\Omega}$, we 
can apply a sequence of arbitrary unitary maps $U_t$ interleaved by invocations of the oracle $O$ to obtain the state
\begin{align}
  \ket{\p_t}= U_T O U_{T-1}O\ldots O U_1 O  U_0\ket{\Omega}.
\end{align} 
Finally, we perform a measurement on the state $\ket{\p_t}$ and  measurement outcomes are then mapped to an oracle $O^i$. We 
say that the algorithm succeeds if it outputs $i$ when applied with $O=O^i$
with a fixed lower bounded probability for all $i$.
This can then be lifted to a high probability guarantee by repetition.
In the example of unstructured search introduced above, Grover's algorithm can be used to identify $i$ with $\mc{O}(\sqrt{N})$ calls to $O^i$ \cite{grover}. Note that this setting can be interpreted as a bandit problem where a single arm always returns reward 1 and all other arms always return reward 0.

We now introduce our definition of  general bandit problems.
In this case, we need to consider the broader class of decision problems where we can partition the set of oracles $\mc{O}=\mc{O}_1\sqcup \mc{O}_2\sqcup\ldots\sqcup \mc{O}_N$ and we want to identify $n$ such that $O\in \mc{O}_n$ with the fewest number of invocations of the given oracle $O$.
Let us now describe how the oracle $O$ that models a pull of the arms is implemented.
As before, we assume that there are $N$ arms.  To query the arms, we consider a Hilbert space $\mc{H}_A$ for the arms with dimension $|\mc{H}_A|=N$ and we identify the arms with states $\ket{i}$ from a fixed basis.
 Moreover, we assume that the internal randomness of the reward is captured through an additional Hilbert space $\mc{H}_P$ and we fix a basis $\ket{\omega}$.  The reward for a certain arm and state of internal randomness is collected in a two-dimensional Hilbert space $\mc{H}_R$ with basis states $\ket{0}$ and $\ket{1}$. We then assume that in each round we can query the arms  through an oracle $O$ acting
 on $\in\mc{H}_A\otimes \mc{H}_P\otimes \mc{H}_R$.
 It acts on a state a basis state  $\ket{i}\ket{\omega}\ket{c}$
 i.e., arm $i$, internal randomness $\omega$, and initial reward state $c$ 
 by
\begin{align}\label{eq:def_oracle_general}
O\ket{i}\ket{\omega}\ket{c}=\ket{i}\ket{\omega}\ket{c+r_i(\omega)}
\end{align}
where $r_i(\omega)\in \{0,1\}$ denotes the reward for arm $i$ with internal randomness $\omega$
and addition is again modulo 2
(this is very similar to the definition in \cite{regret_quantum_bandits}).
Note that $r_i(\omega)$ are not random for fixed $i$ and $\omega$.
Thus, for each arm we receive for every $\omega$ and any arm $i$ a reward that is either $0$ or $1$ and  averaging over $\omega$ gives
\begin{align}
p_i=|\mc{H}_P|^{-1}\sum_\omega r_i(\omega)
\end{align}
the mean reward for arm $i$.
As before, we collect the mean rewards $p_i$ in a vector
$\pp\in [0,1]^N$ with $\pp_i=p_i$. 
Quantum algorithms then consist of a sequence of unitary operations on $\mc{H}_A\otimes \mc{H}_P\otimes \mc{H}_R\otimes \mc{H}_W$ where $\mc{H}_W$ is a work space interleaved by 
calls to the oracle $O$. We emphasize that the indices are named to reflect the Hilbert spaces for the arms, the internal randomness (probability), the rewards, and the work space.

Let us now compare this setting to  the classical setting and outline a crucial difference.  
In the classical setting the rewards are random variables, e.g., functions from some probability space $\Omega$ to $\R$ but all we observe is the received reward and explicit reference to the probability space is not necessary. In contrast, here the probability space is made explicit through the space $\mc{H}_P$ and a crucial ingredient of the model. Moreover, the rewards for all rounds are given by $r_i(\omega)$
so there is no independence of the reward distributions for different times. This is different from the classical setting. But note that if we consider a setting where we  can in each round query an action $(i,\omega)$ and receive the reward $r_i(\omega)$ 
  the best-arm identification problem is not simpler
  than in the standard formulation (if the dimension of $\mc{H}_P$
is sufficiently large). On the other hand, the regret minimization problem is not meaningful in this setting because it is sufficient to identify a single good realization such that $r_i(\omega)=1$
which can then be selected in all future rounds.
Note that this setup might be a reasonable model for certain classical problems, e.g., when the different 
arms correspond to different exercises and different $\omega$ to different students and the goal is to identify the hardest exercise.
}

We will discuss potential applications in the quantum setting below, but first we consider different variants of the problem above. 
Motivated by the differences from the classical setting, we investigate three different settings that are characterized by the degree of control that we have over the 
space $\mc{H}_P$ determining the internal randomness of the bandits. 
\begin{enumerate}
\item We have full control over the space $\mc{H}_P$, i.e., we can apply arbitrary unitary operators to the system $\mc{H}_A\otimes \mc{H}_P\otimes \mc{H}_R$ and a potential work
space.
\item We can prepare an arbitrary  computational basis state $\omega\in\mc{H}_P$ but have no further access to $\mc{H}_P$ except through the oracle $O$.
\item For each pull of the arm, i.e., for each invocation of $O$
the state of $\mc{H}_P$ is a uniformly random basis state $\ket{\omega}$ which is resampled in each step. Equivalently, each oracle invocation uses a different copy of the maximally mixed state on $\mc{H}_P$.
\end{enumerate}
We now discuss how the restricted control over $\mc{H}_P$
can be related to settings without an explicit space $\mc{H}_P$
which are therefore closer to the standard setup of classical bandits.
We define
for $X\in \{0,1\}^N$ (corresponding to a vector of binary rewards)
the oracle $O_X$ acting  on $\mc{H}_A\otimes \mc{H}_R$ by 
\begin{align}\label{eq:def_O_X}
O_X\ket{i}\ket{c}\to \ket{i}\ket{c+X^i}.
\end{align}
Let us assume that $X_t\in \{0,1\}^N$ is a sequence of i.i.d.\ random variables where the coordinates $X_t^i$
are independent and $\mathrm{Ber}(p_i)$ distributed.
Then having access to the sequence of oracles
\begin{align}
\label{eq:oracle_col}
O_{X_t}\ket{i}\ket{c} &= \ket{i}\ket{c+X_t^i}, \quad \text{for $t= 1,2,\ldots$}.
\end{align} 
is similar to 
the second  scenario above (by identifying 
$r_i(\omega_t)=X_t^i$ where $\omega_t$ enumerates the basis of $\mc{H}_P$). Note that here we assume for simplicity that the coordinates $X^i_t$ are
independent, while the general oracle definition in \eqref{eq:def_oracle_general} 
can also model correlated rewards for different arms.
In the classical setting, this  is not important because 
we can never observe the rewards of two different arms in the same round.
In the quantum setting, it might matter because we can query the arms in superposition.

We now similarly reformulate the third setting above.
This setting is similar to having access to an oracle acting as a quantum channel on $\mc{H}_A\otimes \mc{H}_R$  by
\begin{align}
\begin{split}\label{eq:oracle_channel}
\mc{E}(\rho) &= \sum_{X \in \{0, 1\}^n} \Pd(X) \; O_X \rho (O_X)^\dagger
\\
&\text{where $\Pd(X)= \prod_i p_i^{X^i}(1-p_i)^{1-X^i}$}.
\end{split}
\end{align} 
Again, the difference to the setting above is the assumed independence of different arms of the variables $X^i$.

Note that the second and the third setting are equivalent 
if we are allowed to use each of the oracles $O_{X_t}$ only once. Multiple invocations
can be useful to uncompute parts of the computation and thereby avoiding decoherence of the system.

\subsection{Models and Main Results}\label{sec:models}
We now motivate the three settings described above in  more detail and discuss our main results. 
Let us emphasize that quantum bandits can only be useful when the reward is given by the observable of a quantum system or the evaluation of a computation on a quantum device, as the acquisition of the data is commonly seen as the expensive part. 
In other words, it appears unlikely that we collect rewards in some trial and then  store this data in, e.g., a QRAM
to query them in superposition to identify the best arm, as this will always be more expensive than classically evaluating the mean of the collected rewards. Thus, we will motivate all three settings from a quantum perspective.
Note that the classical setting roughly corresponds to  the case where only queries of the form $\ket{i}\ket{\omega}$ without superposition 
are allowed, i.e., we can query  a single arm for a single reward. 
\begin{table}[t]
\centering
\begin{tabular}{l  c  c } 
 Oracle & Lower bound & Upper bound \\[1ex] 
\hline \hline
\\ [-1.5ex]
 Classical & $\bigOlower(H(\pp))$ & $\bigOt(H(\pp))$ \\[1ex]  
 \hline 
 \\[-1.5ex]
 ERM (eq. \eqref{eq:def_oracle_general}) & $\bigOlower(\sqrt{H(\pp)})$  (Thm.~\ref{th:lower_emp}) & $\bigOt(\sqrt{H(\pp)})$ (Thm.~1 in \cite{multi_armed_quantum}) \\[1ex]  
 \hline 
 \\[-1.5ex]
 Reusable (eq. \eqref{eq:oracle_col}) & $\bigqm[1]$ & 
$\bigOt(\sqrt{\sum_i \Delta_i^{-4}})$  (Thm.~\ref{th:reusable})
\\[1ex]  
 \hline 
 \\[-1.5ex]
One-time (eq. \eqref{eq:oracle_channel}) & $\bigOlower(H(\pp))$  (Thm.~\ref{th:main}) & $\bigOt(H(\pp))$ \\ [1ex] 
 \hline
\end{tabular}
\caption{Overview of query complexity bounds. The upper bound for the one-time oracle
and the reusable oracle follow from the classical result. We conjecture
(see Conjecture~\ref{con:reusable})	 that the upper bound
for the reusable oracle are  optimal.}
\label{table:1}
\end{table}

\subsubsection{Empirical risk minimisation}\label{sec:erm}
As already explained in \cite{multi_armed_quantum}, one setting 
where we have full access to an oracle as in \eqref{eq:def_oracle_general} is empirical risk minimization. 
To make this concrete, assume we have a dataset $(x_j, y_j)\in \mc{X}\times \{1,\ldots, K\}$ and a finite set of candidate functions $f_i$. We
now want to find the index $i_0$ such that 
\begin{align}
i_0=\argmin_i \sum_j \mathbf{1}_{f_i(x_j)\neq y_j}
= \argmax_i \sum_j \mathbf{1}_{f_i(x_j)= y_j},
\end{align}
i.e., for simplicity we consider 0-1 loss in a classification problem
or equivalently, we maximize the accuracy over the functions $f_i$.
If we can access $x_j$  and evaluate $f_i$ this provides us with an oracle acting by
\begin{align}\label{eq:oracle_erm}
\ket{i}\ket{\omega}\ket{c}\to    \ket{i}\ket{\omega}\ket{c\oplus r_i(\omega)},
\end{align}
where the reward for $\omega=(x,y)$ is given by $r_i(\omega)=\mathbf{1}_{y=f_i(x)}$. Now, the problem of best arm identification with respect to this oracle is equivalent to empirical risk minimization.
Moreover, this oracle is exactly of the form introduced in  \eqref{eq:def_oracle_general}.
When considering this setting, it is a reasonable assumption that the dataset can be accessed in arbitrary superposition, i.e., is stored in our computing device, and we can also 
evaluate functions $f_i$ and thus losses in superposition.
\blue{
Note that this setup can also arise in the investigation of quantum systems. 
Suppose that we have 
a quantum system of interest with corresponding Hilbert space $\mc{H}_P$ and
observables $A_i$ acting on this space. We assume for simplicity that $A_i \ket{\omega}= r_i(\omega)\ket{\omega}$ where $r_i(\omega)\in \{0,1\}$, i.e., the observables have eigenvalues 0 and 1 and can be simultaneously diagonalized. We are interested in finding the observable with the largest expectation, i.e.,  
\begin{align}
    \argmax_i \tr(A_i\rho)=\argmax_i \sum_{\omega}\bra{\omega}A_i\ket{\omega},
\end{align}
where $\rho \propto \sum_{\omega} \ket{\omega}\bra{\omega}$ is the completely mixed state.
Then, this reduces to the best arm identification problem if we can 
construct an oracle acting as in \eqref{eq:def_oracle_general} where $r_i(\omega)$
corresponds to the eigenvalue of $A_i$ for the eigenvector $\omega$.
}

Note that  the problem of empirical risk minimization and the relation to bandit problems 
was considered before in \cite{multi_armed_quantum},
 where they consider 
an oracle that acts as 
\begin{align}\label{eq:oracle_previous}
 \ket{i}\ket{0}\to \ket{i}(\sqrt{p_i}\ket{1}+\sqrt{1-p_i}\ket{0}).
\end{align}
The relation to our setting is that when applying our oracle to a uniform superposition 
over the $\omega$ register we obtain
\begin{align}\label{eq:uniform_omega}
\sum_\omega
\ket{i}\ket{\omega}\ket{0}\to   \sqrt{p_i} \ket{i}\ket{v}\ket{1}
+ \sqrt{1-p_i}\ket{i}\ket{u}\ket{0}
\end{align}
where $u$ and $v$ are suitable junk states which can be neglected as argued in \cite{multi_armed_quantum} (at least when restricting attention to query complexity). 
 Note that this superposition  eliminates the statistical randomness of the bandits. 
 Then the following result holds.
 \begin{theorem}[Theorem~1 in \cite{multi_armed_quantum}]
 There is a  quantum  algorithm that identifies
 the best arm of a quantum oracle as in \eqref{eq:oracle_previous}
 for any reward vector $\pp\in [\eta,1-\eta]^N$ 
 with probability $1-\delta$ with query complexity
 \begin{align}
 T \leq \tilde{\mathcal{O}}(\sqrt{H(\pp)}).
 \end{align}
 Moreover, this bound is optimal up to logarithmic terms.
 \end{theorem}
 Thus,   a quadratic speedup compared to the classical setting is achievable.
 They prove the lower bound only for the oracle \eqref{eq:oracle_previous}.
 For completeness, we prove the same lower bound when the more general oracle \eqref{eq:def_oracle_general} is available, i.e., the ability to query arbitrary superpositions of the data points does not allow any speedups compared to always considering  the uniform superposition. 
 \begin{theorem}[Informal version]\label{th:lower_emp}
 Any algorithm that identifies the best arm for any reward vector $\pp\in [\eta,1-\eta]^N$  with confidence $1-\delta$ given access to an oracle as in \eqref{eq:def_oracle_general}
 requires at least $\Omega(\sqrt{H(\pp)})$ calls to the oracle.
 \end{theorem}
 The proof and a formal statement of this result are in Appendix~\ref{app:lower_emp}.
 Note that while it is intuitively clear that it is optimal to query over the uniform
 mixture of $\omega$ as in \eqref{eq:uniform_omega} a rigorous proof requires a 
 careful tracking of the classical randomness of the oracle and its interaction with the quantum algorithm.

\subsubsection{Reusable oracles}\label{sec:reuse}
We now consider oracles as in \eqref{eq:oracle_col}, i.e., 
we can query arms in superposition, but we can only retrieve the reward for one chosen realization of the randomness.
A similar type of oracle was considered in \cite{quantum_hedging}. They show that hedging algorithms can be implemented using 
these oracles which have runtime $O(\sqrt{N})$ for $N$ arms, thus offering a quadratic speedup compared to the classical algorithm. Their setting is not directly comparable to the  best arm identification
problem for multi-armed bandits considered here. We try to identify the best arm, while in their setting an $\eps$-optimal arm is sufficient. On the other hand, they want to control a suitable variant of the regret. 
\blue{
One motivation is that as in Section~\ref{sec:erm} we want to identify the observable $A_i$ from a collection of observables that has the maximal expectation. But in contrast to the previous setting, we cannot perform arbitrary manipulations on the experimental setup (because  $\mc{H}_P$ corresponds to a quantum system we study and not part of the computing device).
Instead, we can only  apply maps that transition between different $\omega$ values
from the fixed basis of $\mc{H}_P$ and probe the system through the oracle $O$.
As stated above, this is equivalent to having access to oracles of the form $O_{X_t}$
acting on $\mc{H}_A\otimes \mc{H}_R$.
While this setup might be not the physically most relevant model, it is nevertheless helpful as an intermediate setting that helps us understand when different speedups arise.
For this case we can only give partial results. 
}

We have the following result.
\begin{theorem}\label{th:reusable}
For confidence $\delta\in (0,1)$  there is a quantum algorithm that outputs
the best arm with probability $1-\delta$ using 
\begin{align}
T\leq \bigOt\left( \sqrt{\sum_{i\geq 1}{\Delta_i^{-4}}} \right)
\end{align}
queries to an oracle as in \eqref{eq:oracle_col} 
where $\sim$ indicates terms that are polynomial in 
$\ln(N/(\delta \Delta_2))$.
\end{theorem}
A sketch of the proof of this result is in Appendix~\ref{app:reusable}.
It relies on a small modification of the algorithm in \cite{multi_armed_quantum}. Their algorithm is based on a clever application of variable time algorithms 
\cite{var_time} to count the number of arms with reward bigger than a given threshold and to rotate on the corresponding subspace of arms. As this is the only upper bound in this work, its proof has no direct relation to the remaining results.
We conjecture that the  bound in Theorem~\ref{th:reusable} is also optimal. 
\begin{conjecture}\label{con:reusable}
Any quantum algorithm that identifies the best arm
for any reward vector $\pp\in [\eta,1-\eta]^n$
for some $\eta>0$ with probability $1-\delta$ for some $\delta<\frac12$
requires at least
\begin{align}
T \geq c\sqrt{\sum_{i\geq 1}\Delta_i^{-4}}
\end{align}
calls to an oracle as in \eqref{eq:oracle_col}.
\end{conjecture}
The main difficulty to prove a lower bound is  that the applied oracles can be reused so that the fidelity between the quantum states obtained 
for different mean rewards $\pp$ and $\pp'$ is not necessarily monotonically decreasing. This makes it hard to extend our other proofs that rely on the loss of fidelity in a single step to this setting.
In general, standard techniques to obtain lower bounds and the techniques used in this work do not appear to be sufficient to address this problem.

\subsubsection{One-time oracles}
\blue{
Finally, we consider the oracle defined in \eqref{eq:oracle_channel}.
Let us first make the connection to the oracle $O$ in \eqref{eq:def_oracle_general} a bit more precise.
We assume that we cannot act on the space $\mc{H}_P$ and only extract information from the system through the oracle $O$ otherwise the state on $\mc{H}_P$ follows, e.g., 
a time evolution given through some Hamiltonian $H$.
Then it is reasonable to assume that the 
 initial state on $\mc{H}_P$ is the completely mixed state $\rho_P\propto \sum \ket{\omega}\bra{\omega}$ and then the reduced state on the system $\mc{H}_{A}\times \mc{H}_R$ after application of $O$ is given by 
\begin{align}
    \tr_P O (\rho_{AR}\otimes \rho_P) O^\dagger =
    \sum_{X\in \{0,1\}^N} p(X) O_X\rho_{AR}O_X^\dagger
\end{align}
where $p(X)=|\{\omega: r_i(\omega)=X^i\}/|\mc{H}_P|$
and $O_X$ is defined in \eqref{eq:def_O_X}.
In particular, we recover indeed the expression in \eqref{eq:oracle_channel}
if the rewards $\omega\to r_i(\omega)$ for different arms $i$ are independent.
If the dynamics $H$ is sufficiently mixing also future invocations will approximately act
as the quantum channel $\mc{E}$ on the system. 

Assume that the oracle $O$ corresponds to our running example of identifying the mean of a collection of observables $A_i$.
Then this setup is connected to the theory of quantum sensing \cite{sensing} which  refers to the general use of quantum phenomena to measure quantum observables. 
Using quantum devices to learn properties of quantum systems has recently emerged as a promising direction to achieve quantum advantage. In particular, it was shown that
learning with a quantum device interacting with the system of interest can have an  exponential advantage over simply performing measurements on the system \cite{huang_science}. 
We show that, in contrast, in our setting no speedup compared to the classical setting is possible as stated (slightly informally) in the following result.

\begin{theorem}\label{th:main}
Any algorithm that identifies the best arm for any reward vector
$\pp\in [\eta,1-\eta]^N$ for some $\eta>0$ with probability $1-\delta$ for 
some $\delta<1/2$ based on calls to an oracle as in \eqref{eq:oracle_channel}
requires at least 
\begin{align}
T\geq c(\delta, \eta)H(\pp)
\end{align}
calls to the oracle.
\end{theorem}
 In particular, in our setting it is not advantageous (up to constant factors) to query the system in superpositions, but we can also instead use a classical algorithm to decide which arm $i$ to query, evaluate $\mc{E}(\ket{i}\ket{0}\bra{0}\bra{i})$, i.e., query arm $i$ and then measure the reward register containing a $\mathrm{Ber}(p_i)$ sample. 
In the context of our example, where we want to identify the observable $A_i$ with maximum mean, we can just directly measure $\langle A_i\rangle$.

Again, this result is not a consequence of the generality of reward vectors $\pp$ allowed, but even when the set of mean rewards is known no better result is possible.
Note that the assumption that the rewards are Bernoulli distributed and independent might be unrealistic for applications. However, our main result shows that even in this simplified case no improvement over classical algorithms is possible in terms of query complexity. 
A more precise and slightly stronger statement of the result above is given in Theorem~\ref{th:optimal}. 
Theorem~\ref{th:optimal} and its proof can be found in 
Appendix~\ref{app:main}. An overview of the proof techniques and related results will be given in Section~\ref{sec:overview}.

We also remark that on a technical level, our setting is similar to \cite{faulty_grover}, where they essentially consider the case $p_i=0$ for all $i \neq i_0$ and $p_{i_0}>0$. 
Their motivation is to study Grover search where the oracle has a certain failure probability $1-p_{i_0}$ and they also find that this impeded any quantum speedup.
We  will review their results in more detail in  Section~\ref{sec:channel_oracles} which also provides insight into the more general problem.

}

\subsection{Comparison of settings}
Let us further discuss these results to give a better intuition of the result. 
We remark that access to  the oracle \eqref{eq:def_oracle_general}
is strictly more powerful than access to oracles as in \eqref{eq:oracle_col}
which in turn is more powerful than access to the channel oracle \eqref{eq:oracle_channel}.
This is reflected in the following chain of inequalities for the query complexities
\begin{align}
\sqrt{\sum_{i \geq 1} \Delta_i^{-2}}
\leq
\sqrt{\sum_{i \geq 1} \Delta_i^{-4}}
\leq
{\sum_{i \geq 1} \Delta_i^{-2}}.
\end{align}
Here we used $\Delta_i<1$ for the first inequality.
On the other hand, 
we have the following reverse bound for the complexities
\begin{align}
\blue{\frac{1}{\Delta_1}\sqrt{\sum_{i\geq 1} \Delta_i^{-2}}}\geq \sqrt{\sum_{i\geq 1} \Delta_i^{-4}}
\geq
\frac{1}{\sqrt{N}}{\sum_{i\geq 1} \Delta_i^{-2}}.
\end{align}
\blue{
Here we used $\Delta_1 = p_0-p_1\leq p_0-p_i= \Delta_i$ for the first inequality
and the inequality between arithmetic and quadratic mean for the
second inequality.
Thus, the speedup of the empirical risk minimization setting compared to the (conjectured) complexity of the reusable oracles setting is at most
$\Delta_1^{-1}=(p_0-p_1)^{-1}$
while the speedup between the (conjectured) complexity of the reusable oracles setting and the quantum channel oracle is at most
$\sqrt{N}$ but both can be less, depending on $\pp$.}
To illustrate this further we consider as a prototypical example  a reward vector $\pp$ with $p_0>p_1=\ldots=p_{N-1}$
with $p_0-p_1=\varepsilon$ 
	Then the query complexities are 
	\begin{align}
	T_{\mathrm{erm}}\approx \sqrt{\frac{N}{\eps^2}}=\frac{\sqrt{N}}{\eps}<
	T_{\mathrm{reusable}} \approx \frac{\sqrt{N}}{\eps^2}
	<T_{\mathrm{classical}}\approx T_{\mathrm{one-time}}\approx\frac{N}{\eps^2}.
	\end{align}	
For this setting the two difficulties can be well separated: The expected reward of each arm needs to be estimated (statistical complexity) and the correct arm needs to be searched.
Classically the statistical complexity is $\varepsilon^{-2}$ 
but it can be reduced to $\varepsilon^{-1}$ 
in the empirical risk minimization setting (similar to quantum metrology \cite{metrology}).
Not having access to a superposition of basis states $\ket{\omega}$ prevents this speedup.
The complexity of the search of the best arm is $\sqrt{N}$ in a quantum setting compared
to the complexity of $N$ in the classical setting or a noisy quantum setting. This  gives rise to the difference of the query complexities for the one-time oracle and the classical setting or the reusable oracle.
Thus, even under the favorable assumptions we made, e.g., Bernoulli distributed rewards, quantum algorithms offer no improvement in query complexity when $\mc{H}_P$ is not part of the computing device.

 \section{Proof techniques and overview}\label{sec:overview}
 In this section we give an overview of the  techniques 
 used in the proof of our main result, Theorem~\ref{th:main}. The proof is rather involved and technical, so we collect and review the main ingredients here, along with some results that are of independent interest. We first introduce some additional notation.

 To clarify the setting and to cover general oracles we assume that the oracle acts on a Hilbert space given by $\mc{H}=\mc{H}_A\otimes \mc{H}_R$ where $\mc{H}_A =\langle \ket{i}, 1\leq i\leq N\rangle$ 
and $\mc{H}_R$ is the space where the output is written, which will typically be a single qubit space. We assume that we are given oracles $O_i$ acting by 
\begin{align}\label{eq:def_action_oracle}
 O_i\ket{i}\otimes \ket{w}= \ket{i}\otimes U \ket{w}, \quad O_i\ket{j}\otimes \ket{w}= \ket{j}\otimes \ket{w}\quad \text{for $j\neq i$}
\end{align}
where $U$ denotes a unitary map. This covers the case of the phase flip ($\ket{i}\otimes \ket{w}\to -\ket{i}\otimes \ket{w}$) and the bit flip oracle
($\ket{i}\otimes \ket{w}\to \ket{i}\otimes \ket{w\oplus 1}$). 
 A central role is played by the quantum channels $\mc{F}_i^p$ acting by
\begin{align}\label{eq:def_F_i}
\mathcal{F}_i^p(\rho) = (1-p)\rho + p O_i\rho O_i^\dagger. 
\end{align}
The oracle $\mc{F}_i^p$ implements the pull of arm $i$ with a $\mathrm{Ber}(p)$ distributed arm. Note that the channels $\mc{F}_i^{p_i}$ and $\mc{F}_j^{p_j}$ commute for $i\neq j$
because $O_i$ and $O_j$ commute for $i\neq j$.

To relate this to the setting of Theorem~\ref{th:main}
we consider probability vectors $\pp\in [0,1]^N$
and define further  
\begin{align}\label{eq:concat}
\mc{E}^\pp = \mc{F}_1^{\pp_1}\circ\ldots \circ \mc{F}_N^{\pp_N}.
\end{align}
When $O_i$ denotes the bit flip
of the reward register on arm $\ket{i}$
the channel $\mc{E}^\pp$ agrees with the definition in \eqref{eq:oracle_channel}
where the vector $\pp$ indicates the mean rewards, i.e., 
\begin{align}\label{eq:rel_oracles}
\mc{E}^\pp (\rho)
&= \sum_{x\in \{0,1\}^N}\Pd(x) O_x \rho O_x^\dagger
\quad\text{where}\quad
O_x = \prod_{i: x_i=1} O_i
\quad \text{and} \quad
\Pd(x)= \prod \pp_i^{x_i}(1-\pp_i)^{1-x_i}.
\end{align}
Let us now introduce a class of probability vectors $\pp$, which we will use to establish our lower bounds. Specifically, we adopt the setup previously employed in the proof of Theorem~\ref{th:classical_bandit_result} (see \cite{dar_best_arm,mannor_best_arm}), where the lower bound is demonstrated for this particular class of reward vectors. Extending these results to a setting as in Theorem~\ref{th:classical_bandit_result2} is a promising direction for future research.
 We consider a probability vector $\pp=(p_0,\ldots,p_N)\in [0,1]^{N+1}$ with
 $p_0>p_1>p_2\geq \ldots\geq p_N$.
As before, we denote $\Delta_i = p_0-p_i$.  
Given $\pp$ we then define $N+1$ probability vectors
$\pp^i\in [0,1]^N$ given by $\pp^i_j=p_j$ for $i\neq j$ and $\pp^i_i=p_0$
and $\pp^0$ given by $\pp^0_i=p_i$.
In other words, $\pp^0=(p_1,\ldots, p_N)$ and $\pp^i$ is obtained from
$\pp^0$ by replacing the $i$-th reward by $p_0$.
In particular, for every $i\geq1$ the vectors $\pp^0$ and $\pp^i$ differ only in the entry $i$ and arm $i$ has the highest reward for reward vector $\pp^i$.
 We make the additional  assumption that $\Delta_1=p_0-p_1=p_1-p_2$. 
 This ensures (see Lemma~\ref{le:reward})
 \begin{align}
     \frac14 H(\pp^0)\leq H(\pp^i)\leq 2H(\pp^0)
 \end{align}
 for all $i$. We will also use the shorthand
 $\mc{E}^{\pp_j}=\mc{E}_j$. Our goal is then to show that 
 it is hard to distinguish which of the reward vectors $\pp^i$ 
 corresponds to a given oracle $O$. More precisely, we show that there is an index $i$ such that 
 at least $\Omega(H(\pp^0))$ oracle calls are necessary to identify arm $i$ when the true reward vector is $\pp^i$.

Now that we introduced the relevant notation, we can present the three main ingredients and techniques used in the proof.
First, we show general lower bounds for the
 query complexity to distinguish the
quantum channels $\mc{F}_i^p$ from the identity channel.
Then we carefully analyze the change in fidelity when applying
 $\mc{F}^p_i$ and $\mc{F}_i^q$ 
 to two states $\rho$ and $\sigma$ for $|p-q|$ small.
 Finally, we use decompositions of density matrices and coupling arguments that allow us to combine the previous two ingredients.
 The next three subsections cover those aspects in more detail. We also refer to Appendix~\ref{sec:classical} where we give a quantum inspired proof of the fixed confidence best arm identification problem for classical multi-armed bandits that shares many ideas with the proof of the quantum result.

\subsection{Quantum channel oracles}\label{sec:channel_oracles}
In this section we consider the query complexity for
the decision problem whether we are given 
an oracle as in \eqref{eq:def_F_i}
or the trivial quantum channel. This is a simplified setting of the more general bandit problem. In fact, it corresponds to the special case that all mean rewards are $0$, except for one arm with mean reward $p$.
Most of the work related to oracle query complexity has focused on oracles that act as a unitary map.  There is a large body of research work on quantum channels also with a focus on quantum channel discrimination and general 
lower and upper bounds were derived \cite{channel_discrimination_old,channel_discrimination,channel_disriminate_ultimate}. However, the application of these general bounds is mostly targeted towards rather simple  channels, in particular channels that implement error mechanisms present in quantum devices.
Those general results are not directly applicable here, and we will consider bounds targeted
at our specific setting. Phrased as above, the setting in this section essentially agrees with 
\cite{faulty_grover}
where they considered the quantum search problem
given an oracle with a certain failure probability 
$(1-p)$ (note that we exchanged $p$ and $1-p$
compared to their convention, which is more natural in our bandit setting). Their main result is that  $N/p$ queries are required to identify $i$. In particular, no speedup compared to  classical algorithms is possible.

Below, 
we will sketch a different proof of their result because we believe it makes the main ideas slightly clearer and because the proof strategy is an important ingredient used in the proof of Theorem~\ref{th:main}.
The intuition of the proof is that the progress we make is directly related to the decoherence of the state, as measured by its purity. 
As the purity is lower bounded, this gives us tight control of the progress.

We could make the output space explicit by decomposing the work space $\mc{H}_A$ but this is not necessary.
As above, we assume that we have oracle access to one of  the quantum channels $\mc{F}_i(\rho)$
as in \eqref{eq:def_F_i} which we seek to identify. For convenience we define $\mc{F}_0=\id$.
Then we have the following slightly extended version of Theorem~1 in \cite{faulty_grover}.

\begin{theorem}\label{th:faulty_grover}
Any algorithm that can decide whether $\mc{F}=\mc{F}_i^p$ for some $i>0$ or $\mc{F}=\mc{F}_0$ with probability $1-\delta$ requires at least 
\begin{align}\label{eq:lower_noisy_grover}
T\geq \frac{(1-p)(1-4\delta(1-\delta))^2}{p}N
\end{align}
calls to the channel.
\end{theorem}
\begin{remark}
\begin{enumerate}
\item
We emphasize again that the original proof   in \cite{faulty_grover}
immediately generalizes to the slightly different setting considered here.
\item
Note that the bound becomes vacuous for $p\to 1$. Then we recover the setting
of Grover's algorithm, where the well known lower bound scales as $\sqrt{N}$.
For $p=1-\sqrt{N}^{-1}$ the bound \eqref{eq:lower_noisy_grover} agrees
with the Grover lower bound, so it is possible that
a small error probability which depends on the number of arms does not impede
the quadratic speedup that is possible in the noiseless case.
Similar questions were investigated in \cite{random_noisy}.
\end{enumerate}
\end{remark}
Here we discuss the key elements of the proof, while
some calculations are delegated to Appendix~\ref{app:channel_oracles}.
\begin{proof}
Let us denote by $\Phi_U$ the quantum channel acting by the unitary $U$, i.e., $\Phi_U(\rho)=U\rho U^\dagger$. We consider an algorithm 
acting by $\Phi_{U_{T}}\circ \mc{F}\circ \Phi_{U_{T-1}}\circ\ldots \circ \Phi_{U_1}\circ \mc{F}\circ \Phi_{U_0}(\rho_0)$ on some initial state $\rho_0=\ket{\Omega}\bra{\Omega}$ for some pure state $\Omega$.
We define 
\begin{align}
\tilde\rho_t^i = \Phi_{U_t}( \rho_t^i), \quad \rho_t^i = \mc{F}_i( \tilde{\rho}_{t-1}^i),\quad \rho_0^i=\rho_0.
\end{align}
Note that for $i=0$ the state remains pure
during the entire algorithm and we denote it by $\psi_t$ and $\tilde\psi_t$.
We now define 
\begin{align}
R_t^i &=\tr(\rho_t^i)^2,
\\
F_t^i &= F(\rho^i_t, \rho^0_t),
\end{align}
i.e., the purity of the state and the fidelity (defined by $F(\sigma, \rho)=\left( \tr\sqrt{\sqrt{\rho}\sigma\sqrt{\rho}}\right)^2$) of the state with respect to the state corresponding to the trivial oracle. 
Let us remark on the notation, that we always write $\tr(\rho)^2=\tr(\rho \rho)$
for the trace of the square of an operator while $(\tr \rho)^2$ denotes the squared trace.
For a brief overview of properties of the fidelity, 
we refer to Appendix~\ref{app:distance}.
We show that the changes of the two quantities are directly related, i.e., for every increase in distance (loss in fidelity) we have to pay with a loss in purity, i.e., decoherence.

We define the projections $P_i = \ket{i}\bra{i}\otimes \id$.
Then it can be shown (see full proof in Appendix~\ref{app:channel_oracles})
that 
\begin{align}
\label{eq:fid_pur_relation}
F_{t-1}^i-F_t^i 
\leq 
2p \lVert P_i{\psi}_{t}\rVert\left( \tr(\tilde{\rho}^i_{t-1} - O_i  \tilde{\rho}^i_{t-1}O_i^\dagger)^2\right)^{\tfrac12}
\leq 
2 \lVert P_i{\psi}_{t}\rVert \sqrt{\frac{p}{1-p}} \sqrt{R_{t-1}^i - R_{t}^i }
\end{align}

Note that the initial values of $F$ and $R$ are $F_0^i=R_0^i=1$ and $R_t^i\geq 0$. Thus, we conclude using Cauchy-Schwarz
\begin{align}
\begin{split}
\sum_i (F_0^i - F_T^i)
&=
\sum_{i,t } (F_{t-1}^i-F_t^i)
\leq \sum_{i,t} 2 \lVert P_i{\psi}_{t}\rVert \sqrt{\frac{p}{1-p}} \sqrt{R_{t-1}^i - R_{t}^i }
\\
&\leq 2 \sqrt{\frac{p}{1-p}}
\left(\sum_{i,t} \lVert P_i{\psi}_{t}\rVert^2\right)^{\tfrac12}
\left(\sum_{i,t}R_{t-1}^i - R_{t}^i\right)^{\tfrac12}
\\
&\leq
2 \sqrt{\frac{p}{1-p}}
\left(\sum_{t} \lVert {\psi}_{t}\rVert^2\right)^{\tfrac12}
\left(\sum_{i}R_{0}^i - R_{T}^i\right)^{\tfrac12}
\leq
2 \sqrt{\frac{p}{1-p}}\sqrt{T}\sqrt{N}.
\end{split}
\end{align}
Finally we use our assumption that the algorithm can decide whether the oracle is trivial $\mc{E}=\mc{E}_0$ or not with probability $1-\delta$ for some $\delta<1/2$.  From this we can conclude that the output of the algorithm 
for oracles $\mc{F}^0$ and $\mc{F}^i$ must be sufficiently different. 
Formally, success of the algorithm implies (see \eqref{eq:trace_success} and \eqref{eq:trace_fidelity} in
Appendix~\ref{app:distance}) that for each
$i$
\begin{align}
1-2\delta \leq T( \rho_T^i , \rho_T^0)
\leq \sqrt{1-F(\rho_T^i, \rho_T^0)}
\end{align}
where $T(\rho, \sigma) = \tfrac12 \lVert \rho-\sigma\rVert_{\tr}$ denotes the trace distance.
This implies the bound 
$F_T^i \leq 4\delta(1-\delta)$.
We conclude that 
\begin{align}
2 \sqrt{\frac{p}{1-p}}\sqrt{T}\sqrt{N}\geq N (1- 4\delta(1-\delta))
\Rightarrow T\geq \frac{N(1-p)(1-4\delta(1-\delta))^2}{p}.
\end{align}

\end{proof}
A natural question suggested by this result and also our main result is whether speedups
can be obtained for non-unitary oracles. This question was also posed in \cite{faulty_grover}.
We now show that this is not true (without making additional assumptions). The simplest example is a faulty oracle that indicates its own failure. To define this 
we consider  oracles $O_i$ acting by 
\begin{align}\label{eq:oracle_own_failure}
O_i\ket{i}\ket{c}= -\ket{i}\ket{c\oplus 1}, \quad O_i\ket{j}\ket{c}=\ket{j}\ket{c\oplus 1}\quad \text{for $j\neq i$},
\end{align}
i.e., the bit flip indicates the marked element $i$ and the change in the second register $\ket{c}\to\ket{c\oplus 1}$ is used to store  that the oracle worked.
We consider the faulty versions of these oracles given by
\begin{align}
\mc{F}_i(\rho) = p O_i\rho O_i^\dagger + (1-p)\rho.
\end{align}
Then we can obtain the same speedup as with the usual oracle except that we need to correct for the number of times the oracle is not working. 
This is not in contradiction to the previous result, as this oracle is not of the form defined in \eqref{eq:def_action_oracle}. In particular, the action of the oracle in \eqref{eq:oracle_own_failure} 
is not trivial on $\ket{j}\otimes \ket{c}$. 
\begin{theorem}\label{th:speedup}
The channel $i$ can be identified with probability at least $1/4$ using $\lfloor\pi / (2\theta p) \rfloor$ queries to the oracle where $\theta = 2\arcsin(\sqrt{N}^{-1})\approx 2\sqrt{N}^{-1}$.
\end{theorem}
\begin{remark}
Note that up to constant factors we need $\sqrt{N}/p$ queries of which
typically $\sqrt{N}$ queries work, this is the same scaling as 
the usual Grover algorithm. 
As usual, this bound can also  be obtained if $p$ is unknown by iteratively increasing the number of iterations in the algorithm described below.
\end{remark}
The proof can be found in Appendix~\ref{app:speedup}. While this result is simple and not very surprising, it underlines that it will
be difficult to obtain general results showing that no quantum speedup is possible.

\subsection{Optimal Fidelity Bounds and Implications for non-adaptive algorithms}
\label{sec:non_adaptive}
In this section we state optimal fidelity bounds 
when applying the channels $\mc{F}_i^p$ and
$\mc{F}_i^q$ to two states $\rho$ and $\sigma$.
Then we show how this allows us to derive bounds 
for non-adaptive algorithms and suboptimal bounds
for adaptive algorithms. While those results also directly follow from Theorem~\ref{th:main} 
we think it is nevertheless helpful to include those as they motivate the result and clarify the scaling.

We start with  a lemma that controls the fidelity between
applications of the oracle. This result provides a sharp bound that might be of independent interest.
\begin{lemma}\label{le:fidelity}
Assume that $O_i$ is self-adjoint and unitary and acts as in \eqref{eq:def_action_oracle}.
For density matrices $\rho, \sigma$ and $p,q\in [\eta, 1-\eta]$ the following bound holds 
\begin{align}
\sqrt{F}(\mc{F}_i^p(\rho), \mc{F}_i^{q}(\sigma))
\geq \sqrt{F}(\rho, \sigma)  - \frac{(p-q)^2}{\eta(1-\eta)} \sqrt{\tr (P_i\rho) \tr(P_i\sigma)}
\end{align}
where $P_i=\ket{i}\bra{i}\otimes \id$ denotes as before the projection on state $\ket{i}$.
\end{lemma}
The proof of this lemma can be found in Appendix~\ref{app:proof_fid_lemma}.  Note the quadratic scaling in $p-q$. The condition that $p$ and $q$ are away from 0 and 1 is necessary because otherwise the optimal bound only scales with $|p-q|$ (note that the bound in Lemma~\ref{le:fidelity} is vacuous as $\eta\to 0$). These results mirror the classical results about discerning $\mathrm{Ber}(p)$ and $\mathrm{Ber}(q)$ variables.
We state one direct consequence of the previous lemma.
\begin{corollary}\label{co:fid1}
Let $\pp, \pp' \in [\eta, 1-\eta]^N$. Then
\begin{align}
\sqrt{F}(\mc{E}^\pp(\rho),\mc{E}^{\pp'}(\sigma))\geq \sqrt{F}(\rho,\sigma)
-\sum_i \frac{(\pp_i-\pp_i')^2}{2\eta(1-\eta)}\sqrt{\tr(P_i\rho)\tr(P_i\sigma)}.
\end{align}
\end{corollary}
\begin{proof}
We note that $[P_i,O_j]=0$ for $i\neq j$ and as $O_j$ is unitary we get
\begin{align}
\tr ( P_i \mc{F}_j^p(\rho))
=
\tr( P_i ((1-p)\rho+p O_j\rho O_j^\dagger))
=
\tr( P_i ((1-p)\rho+p \rho O_j^\dagger O_j))=
\tr(P_i \rho).
\end{align}
Then  Lemma~\ref{le:fidelity} can be applied inductively to the relation \eqref{eq:concat}
to obtain the claim. 
\end{proof}
From this result we conclude that any non-adaptive algorithm requires the same amount of oracle queries as the best classical algorithm. 
\begin{corollary}\label{co:fid2}
Assume that $\pp^j$ are as introduced at the beginning of
this section with $p_i\in [\eta,1-\eta]$ for
some $\eta>0$.
Fix a density matrix $\rho$.
We have access to $m$ copies of the state $\mc{E}(\rho)$
and it is known that $\mc{E}$ is as in \eqref{eq:oracle_channel}
where the mean reward vector is in 
$\{\pp^0,\ldots,\pp^N\}$. If the best arm of $\mc{E}$ can be identified
 with probability at least $1-\delta$  for some $\delta< 1/2$ 
then
\begin{align}
m\geq \frac{\eta(1-\eta)(1-2\sqrt{\delta(1-\delta)})}{16} H(\pp)
= \Omega(H(\pp)).
\end{align}
\end{corollary}
The short proof of this result can be found in Appendix~\ref{sec:co_fid2proof}.

We can similarly derive a suboptimal bound for adaptive algorithms.
Here we lose a $\sqrt{N}$ factor compared to the tight result stated in Theorem~\ref{th:main}.

The reason that the techniques used in the proof of Theorem~\ref{th:faulty_grover} do not provide optimal lower bounds in the more general setting is that 
the argument uses in an essential way that one of the density matrices is pure.
Indeed, we show that the state of the other oracle  decoheres with respect to this pure reference state.
 In the setting here, both density matrices are highly mixed, so it is more subtle to formalize their decoherence. 
\begin{corollary}\label{co:fid3}
Assume that $\pp^j$ are as introduced at the beginning of Section~\ref{sec:classical} with $p_i\in [\eta,1-\eta]$ for
some $\eta>0$.
Any quantum algorithm that identifies the best arm
when it is known that the reward vector is in 
$\{\pp^0,\ldots,\pp^N\}$ with probability at least $1-\delta$  for some $\delta< \tfrac12$ 
requires at least 
\begin{align}
T\geq \frac{\eta(1-\eta)(1-2\sqrt{\delta(1-\delta)})}{16\sqrt{N}}H(\pp)
= \Omega(H(\pp)/\sqrt{N})
\end{align}
calls to the oracle $\mc{E}_i(\rho)=\mc{E}^{\pp^i}(\rho)$. 
\end{corollary}
\begin{proof}
The proof is close to the proof of Corollary~\ref{co:fid2}.
We assume we are given any algorithm 
$ (\mc{E}_i\otimes \id)\circ \mc{E}_{U_T}
\circ \ldots \circ (\mc{E}_i\otimes \id)\circ \mc{E}_{U_1}$
where $U_i$ are arbitrary unitary maps.
We denote the  state using the oracle $i$ before the 
$t$-th invocation of the oracle by $\rho_t^i$.
Using the invariance of the fidelity under unitary maps and Corollary~\ref{co:fid1}
we bound
\begin{align}
\sqrt{F}(\rho_T^i,\rho_T^0)
\geq 1 - \sum_t 
 \frac{4\Delta_i^2}{\eta(1-\eta)}\sqrt{\tr(P_i\rho_t^i)\tr(P_i\rho_t^0)}.
\end{align}

As the algorithm can discern the oracles $\mc{E}_i$ and $\mc{E}_0$ with probability at least $1-\delta$ we have the bound
\begin{align}
1-2\delta \leq
T( \rho^i_T ,\rho^0_T)
\leq \sqrt{1-F(\rho^i_T ,\rho^0_T)}.
\end{align}
We conclude that 
\begin{align}\label{eq:hellstroem_fidelity_main}
2\sqrt{\delta(1-\delta)}\geq 
\sqrt{F}(\rho^i_T ,\rho^0_T)\geq 1 - 
\sum_t 
 \frac{4\Delta_i^2}{\eta(1-\eta)}\sqrt{\tr(P_i\rho_t^i)\tr(P_i\rho_t^0)}.
\end{align}
Thus we get using $\sum_{i} \tr(P_i \rho_0^t)=1$ 
and $\tr(P_i \rho_t^i)\leq 1$
\begin{align}
\begin{split}
\frac{\eta(1-\eta)(1-2\sqrt{\delta(1-\delta)})}{4}\sum_{i\geq 2} \Delta_i^{-2}
\leq \sum_{t, i} \sqrt{\tr(P_i\rho_t^i)\tr(P_i\rho_t^0)}
\\
\qquad \qquad\leq \left( \sum_{t,i} \tr(P_i \rho_t^i)\right)^{\frac12}
\left( \sum_{t,i} \tr(P_i \rho_t^0)\right)^{\frac12}
\leq \sqrt{NT} \sqrt{T}.
\end{split}
\end{align}
Reorganizing this we obtain the claimed result.
\end{proof}

\subsection{Coupling Arguments and Density Matrix Decompositions}
\label{sec:main_proof}
Roughly, we have seen so far two ingredients to prove lower bounds. First, we considered in Section~\ref{sec:channel_oracles} the relation between 
purity and fidelity and used this to show that no $\sqrt{N}$
Grover type speedup is possible. Then, in Section~\ref{sec:non_adaptive}
we derived optimal bounds of the fidelity loss 
giving the right scaling in $\Delta_i^{-2}$ but which is not
sufficient to exclude the Grover search speedup, see Corollary~\ref{co:fid3}.
Now, we show how those two lines can be combined to give optimal bounds. 
Let us denote by $\rho_t^j$ the state of the algorithm after $t$ steps when the oracle $\mc{E}^j$ is used.
As in Section~\ref{sec:channel_oracles} we would like to 
use that $\rho_t^j$ decoheres with respect to $\rho_t^0$.
However, this is more difficult to formalize as $\rho_t^0$ now also is a highly mixed state, and
we did not succeed in finding a suitable quantity
that captures this. Thus, we decompose $\rho_t^0$
into a mixture of pure states (not its spectral decomposition, but we split each pure state in two whenever applying an oracle $\mc{F}_i^p$). We also define a corresponding decomposition of $\rho_t^j$ 
so that the reasoning of Section~\ref{sec:channel_oracles}
can be applied term by term.
Thus, our proof relies on a coupling argument, a technique that  is standard in the theory of stochastic processes but not so much in quantum information theory. For an overview of this technique, we refer to \cite{coupling}.

To implement this, we need a
strengthened version of Lemma~\ref{le:fidelity} which 
not only has the optimal scaling but in addition 
provides a decomposition $\rho = p_0\rho_0+p_1\rho_1$
and $\sigma=q_0\sigma_0 +q_1\sigma_1$
such that the lower bound also applies to 
$\sqrt{p_0q_0}\sqrt{F}(\rho_0,\sigma_0)+\sqrt{p_1q_1}\sqrt{F}(\rho_1,\sigma_1)$ (which lower bounds the fidelity of the
mixture). Let us state the key lemma.
\begin{lemma}\label{le:coupling_const}
Consider a density matrix $\rho$ and a pure state $\psi$. Consider a unitary and self-adjoint map 
$U$ and the channel $\mc{E}_U^p$ defined by 
$\mc{E}_U^p(\rho)= pU\rho U^\dagger+(1-p)\rho$.
Let $0<\eta<1/2$ and $p,q\in[\eta, 1-\eta]$.
Then there are density matrices $\rho_0=\rho$ and $\rho_1=U\rho U^\dagger$ and pure states $\psi_0$ and $\psi_1$ and a real number $q'$ such that
\begin{align}
    \mc{E}_U^p(\rho)=p\rho_1+(1-p)\rho_0,
    \quad 
    \mc{E}_U^q(\ket{\psi}\bra{\psi})=
q'\ket{\psi_1}\bra{\psi_1}+(1-q')\ket{\psi_0}\bra{\psi_0}
\end{align}
and the following bound holds.
Let $S= \sqrt{F}(\rho, \ket{\psi}\bra{\psi})=\sqrt{\bra{\psi}\rho\ket{\psi}}$
denote the initial fidelity.
Then we have
\begin{align}
\begin{split}
\sqrt{F}\Big(\mc{E}_U^p(\rho)&,\mc{E}_U^q(\ket{\psi}\bra{\psi})\Big)
\geq
\sqrt{(1-p)(1-q')}\sqrt{F}(\rho_0,\ket{\psi_0}\bra{\psi_0})
+ \sqrt{pq'}\sqrt{F}(\rho_1, \ket{\psi_1}\bra{\psi_1})
\\
&\geq S
- \frac{(p-q)^2 |(\psi,( U\rho U^\dagger - \rho)\psi)|}{2\eta S}
- \frac{(p-q)^2 |\mathrm{Re}(\psi, U\rho \psi)-(\psi,\rho\psi)|^2}{8\eta^2S^3}.
\end{split}
\end{align}
\end{lemma}
\begin{remark}
Let us give some explanation regarding this lemma.
\begin{enumerate}
\item Note that for $p=q$ we recover the result that quantum channels can only increase the fidelity for our specific channel. 
\item The problem that this lemma solves is that we need a bound on the loss in fidelity that has, firstly, the optimal quadratic rate in $(p-q)$, secondly, the bound needs  to have a form that allows to exploit
the specific structure of the oracles $O_i$ which act non-trivially only on a small subspace, and, thirdly, we later want to use that the density matrices
$\rho$ decohere with respect to the state $\psi$. The last point will become clearer in the proof of Theorem~\ref{th:main} in Appendix~\ref{app:main}, but note that
the expression $(\psi,( U\rho U^\dagger - \rho)\psi)$ appeared already in the proof of Theorem~\ref{th:faulty_grover} (see Equation \eqref{eq:fid_pur_relation}) which indicates that similar arguments can be applied. We remark that Lemma~\ref{le:fidelity}
above already satisfied the first two requirements, but the lack of the third requirement allowed us to only show the suboptimal bound in Corollary~\ref{co:fid3}. It is also quite straightforward to satisfy the second and the third requirement with the suboptimal rate $|p-q|$. But this also gives only a suboptimal bound. 
\item The second error term does not require all desiderata outlined above, but it is of higher order (note the extra square) which is sufficient to control it. 
\end{enumerate}
\end{remark}
The proof of this lemma can be found in Appendix~\ref{app:coupling}, there we also state a slight generalization that is used in the proof of Theorem~\ref{th:main} in Appendix~\ref{app:main}.

\section{Discussion}
In this work we investigated quantum algorithms for multi-armed bandit problems.
It was shown earlier in \cite{multi_armed_quantum} that  quantum algorithms for best arm identification with fixed confidence can have 
 a quadratic speedup compared to their classical counterparts. 
 This result is based on the assumption that the arms and the randomness of the rewards of the arms can be both queried in superposition.
These assumptions are reasonable in the setting of empirical risk minimization, 
 where we can evaluate loss values in superposition.
 However,
 there are many settings, e.g., motivated by quantum sensing, where it might not be possible to query the internal randomness of the bandits in superposition. Instead, every pull of the lever returns a single random reward.
We then show that in this setting no speedup compared to classical algorithms is possible. 

This highlights that classical randomness pose a major challenge for quantum algorithms.
In our case the randomness of the rewards even prevent Grover type speedups of the search problem that one would naively expect to arise from the search part of the multi-armed bandit problem.  Note that such a speedup is possible in the intermediate regime that we considered. When we can select the state of the internal randomness of the oracle (but not query  it in superposition) the statistical complexity of the problem remains the same, but
we can search through the arms faster, providing some speedup.

There are many open questions related to this work and we will briefly mention two. Firstly, 
classical randomness appears frequently in different settings. This has been studied a lot in the context of noise channels but not so much in other contexts, e.g., machine learning.

Secondly, our proofs proceed by directly controlling the fidelity between the quantum states when invoking different oracles. From a methodological side, it
would be interesting to see to what degree 
the well-known strategies to lower bound the query complexity
 like the polynomial  \cite{polynomial_lower}  or the adversarial method \cite{adversary} extend to non-unitary oracles.

\bibliography{./ml.bib}

\appendix
\section{Overview of the Appendix}
In this appendix we collect all the proofs of the results in the paper. We start with a brief review of distance measure of quantum states in Section~\ref{app:distance} where we list some well-known results on distance measures of quantum states for reference. Then, in Section~\ref{app:auxiliary},  we collect some simple
auxiliary results that will be used in the proofs later on. 

The following sections of this appendix then roughly
follow the outline given in Section~\ref{sec:overview} in the main text and add the missing proofs.
In particular, in Appendix~\ref{app:channel_oracles}
we give the complete proof of Theorem~\ref{th:faulty_grover},
in Appendix~\ref{app:proof_fid_lemma} we first prove 
a simpler version of Lemma~\ref{le:fidelity} stated in Lemma~\ref{le:fidelity_simple}, 
and then give the proofs of
Lemma~\ref{le:fidelity} and Corollary~\ref{co:fid1}. In Appendix~\ref{app:coupling} we prove Lemma~\ref{le:coupling_const}
and provide a small technical extension that is used in the proof of the main result. 
In Appendix~\ref{sec:classical} we give a complete quantum inspired proof of the lower bound for classical bandits. This proof clarifies some of the choices in the proof of Theorem~\ref{th:main} which can be found along with the more precise statement of the result in Theorem~\ref{th:optimal} in Appendix~\ref{app:main}.
Finally, the Appendices~\ref{app:lower_emp}, \ref{app:speedup} and \ref{app:reusable}
contain the proofs of Theorem~\ref{th:lower_emp}, Theorem~\ref{th:speedup} and
Theorem~\ref{th:reusable} respectively. The proof of Theorem~\ref{th:lower_emp} shares some ideas with,
e.g., the quantum inspired proof for the classical bandits, the other two proofs use different ideas than the remaining results of this paper and are independent of the other sections.
\section{A Brief Review of Distance Measures for Quantum States}\label{app:distance}
For the convenience of the reader we give a brief review of distance measures for quantum states.
Textbooks on quantum computation, e.g., \cite{NC00} discuss this thoroughly.
For a review on fidelities we refer to \cite{review_fidelity}. 
We consider the trace distance which is defined by
\begin{align}
T(\rho, \sigma) = \frac12 \lVert \rho-\sigma\rVert_{\tr}
\end{align}
where the norm indicates the trace norm defined by $\lVert A\rVert_{\tr}=\tr(\sqrt{A^\dagger A})$.
It has the property that for any POVM $\{E_i\}$ the outcome probabilities
\begin{align}
p_i=\tr(E_i\rho),\quad q_i=\tr(E_i\sigma)
\end{align} 
the total variation distance between the probability vectors $p_i$ and $q_i$ satisfy
\begin{align}
\frac12 \sum_i |p_i-q_i|\leq T(\rho,\sigma).
\end{align}
The Helmstrom measurement gives the optimal discrimination probability of two states
and has success probability
\begin{align}\label{eq:trace_success}
p_{\mathrm{success}}=\frac12 + \frac12 T(\rho, \sigma). 
\end{align}
For many applications the fidelity is a more useful distance measure to obtain optimal 
bounds. It is defined
by
\begin{align}
\sqrt{F}(\rho,\sigma) = \tr(\sqrt{\rho^{\frac12}\sigma\rho^{\frac12}})
= \lVert\rho^{\frac12}\sigma^{\frac12}\rVert_{\tr}.
\end{align}
Some authors instead call the square of this expression the fidelity, and to clarify our convention we added the square root. As suggested by the notation we set $F=\sqrt{F}^2$. 
We collect some properties of the fidelity that we will use frequently.
\begin{enumerate}
\item For a density matrix $\rho$ and a pure state $\psi$ the fidelity is given by
\begin{align}
\sqrt{F}(\ket{\psi}\bra{\psi},\rho)=\sqrt{\langle\psi,\rho\psi\rangle}.
\end{align}
\item For any density matrices $\rho$, $\sigma$ and a quantum channel $\mc{E}$ 
the following bound holds
\begin{align}
\sqrt{F}(\rho,\sigma)\leq \sqrt{F}(\mc{E}(\rho),\mc{E}(\sigma)).
\end{align}
\item If the quantum channel $\mc{E}$ acts by a unitary matrix, i.e., 
$\mc{E}(\rho)=U\rho U^\dagger$ then
\begin{align}
\sqrt{F}(\rho,\sigma)  = \sqrt{F}(\mc{E}(\rho),\mc{E}(\sigma)).
\end{align}
\item The fidelity is strongly concave 
\begin{align}\label{eq:fidelity_strong_concave}
\sqrt{F}\left(\sum_i p_i\rho_i,\sum_i q_i\sigma_i\right)\geq 
\sum_i \sqrt{p_iq_i} \sqrt{F}(\rho_i, \sigma_i).	
\end{align}
This directly implies concavity
\begin{align}
\sqrt{F}\left(\sum_i p_i\rho_i,\sum_i p_i\sigma_i\right)\geq 
\sum_i p_i \sqrt{F}(\rho_i, \sigma_i).	
\end{align}
\item Fidelity and trace distance are related by 
\begin{align}\label{eq:trace_fidelity}
1-\sqrt{F}(\rho,\sigma)\leq T(\rho,\sigma)\leq \sqrt{1-F(\rho,\sigma)}.
\end{align}
\end{enumerate}
Those properties can be proved using Uhlmann's Theorem which states that
\begin{align}
\sqrt{F}(\rho,\sigma)=\max_{\p, \psi} \langle \p, \psi\rangle
\end{align}
where the maximum is over all purifications $\psi$ and $\p$ of $\rho$ and $\sigma$,
respectively. We use this result to bound the fidelity change of certain quantum operations (see Lemma~\ref{le:fidelity} and \ref{le:coupling_const}).

\section{Auxiliary Lemmas}\label{app:auxiliary}
Here we include simple, mostly algebraic lemmas that are used in
the proof of Theorem~\ref{th:faulty_grover} and in the proofs in
 Appendix~\ref{app:proof_fid_lemma}. 
 The first lemma is a simple Cauchy-Schwarz estimate that in addition exploits invariant subspaces of an operator $O$. It is used in the proof of Theorem~\ref{th:faulty_grover}.
 Recall our convention that $\tr(\rho)^2=\tr \rho\rho$.
\begin{lemma}\label{le:projection}
Let $O$ be a unitary operator and $P$  a self-adjoint orthogonal projections such that $(\id-P)O=\id-P$, i.e., $O$ acts trivially on the orthogonal complement of the projection  $P$. Then, for any vector $\ket{\p}$ and density matrix $\sigma$
the bound 
\begin{align}
\left|\bra{\p}\left(O\sigma O^\dagger-\sigma\right)\ket{ \p}\right|
\leq
2\lVert P\p\rVert  \lVert \p\rVert \left(\tr\left(O\sigma O^\dagger - \sigma\right)^2\right)^{\frac12}
\end{align}
holds.
\end{lemma}
\begin{proof}
We define $Q=\id-P$. Then we have $QO=Q$.
By assumption, we can decompose
\begin{align}
\begin{split}\label{eq:projection_transform}
\sigma - O {\sigma}O^\dagger
&=
(P + Q)(\sigma- O\sigma O^\dagger)(P + Q)
\\
&=
(\sigma - O \sigma O^\dagger)P
+P (\sigma - O \sigma O^\dagger)Q
\end{split}
\end{align} 
where we used $QO\sigma O^\dagger Q= Q\sigma Q$.
Using Cauchy-Schwarz for the Hilbert-Schmidt scalar product we can bound for $M=M^\dagger$
\begin{align}\label{eq:scalar_product_by_HS}
\left|\bra{ \p_1} M \ket{\p_2}\rangle \right|
= \left|\tr( \ket{\p_2}\bra{\p_1} M)\right|
\leq \left( \tr (\ket{\p_2}\bra{\p_1}\ket{\p_1}\bra{\p_2})\tr M^2 \right)^{\tfrac12}
=\lVert \p_1\rVert \cdot\lVert \p_2\rVert \left(\tr M^2\right)^{\tfrac12}.
\end{align}
Using \eqref{eq:projection_transform} and \eqref{eq:scalar_product_by_HS} we can continue to estimate
\begin{align}
\begin{split}
\left|\bra{\p}\ket{\left(O\sigma O^\dagger-\sigma\right) \p}\right|
&\leq 
\left(\lVert P\p\rVert\,  \lVert \p\rVert
+ \lVert P\p\rVert \, \lVert Q\p\rVert\right)
 \left(\tr\left(O\sigma O^\dagger - \sigma\right)^2\right)^{\frac12}
\\
&\leq
 2 \lVert P\p\rVert\cdot \lVert \p \rVert\left(\tr\left(O\sigma O^\dagger - \sigma\right)^2\right)^{\frac12}.
\end{split}
\end{align}
This ends the proof.
\end{proof} 
 The next lemma states two simple algebraic bounds.
\begin{lemma}\label{le:trivial_bound}
For $p,q\in [c,1-c]$ the following bounds hold
\begin{align}\label{eq:bound_cos}
\sqrt{(1-p)(1-q)}+\sqrt{pq}&\geq 1-\frac{|p-q|^2}{4c(1-c)},
\\
\label{eq:bound_sin}
\left|\sqrt{(1-p)q}-\sqrt{(1-q)p}\right|&\leq 
\frac{|p-q|}
{2\sqrt{c(1-c)}}.
\end{align}
\end{lemma}
\begin{proof}
We first consider the second inequality.
We note that $|\sqrt{x}-\sqrt{y}|\leq |x-y|/(\sqrt{x}+\sqrt{y})$ and
thus
\begin{align}
\left|\sqrt{(1-p)q}-\sqrt{(1-q)p}\right|\leq \frac{|p-q|}{\sqrt{p(1-q)}+\sqrt{q(1-p)}}
\leq \frac{|p-q|}{2\sqrt[4]{p(1-q)q(1-p)}}
\leq \frac{|p-q|}{2\sqrt{c(1-c)}}
\end{align}
where we used the arithmetic geometric mean inequality in the middle step.
To prove the first bound, we note that 
\begin{align}
\left(\sqrt{(1-p)(1-q)}+\sqrt{pq}\right)^2
+
\left(\sqrt{(1-p)q}-\sqrt{(1-q)p}\right)^2
=1
\end{align}
This implies
\begin{align}
\begin{split}
\sqrt{(1-p)(1-q)}+\sqrt{pq}&=
\sqrt{1-\left(\sqrt{(1-p)q}-\sqrt{(1-q)p}\right)^2}
\\
&\geq 1-\left(\sqrt{(1-p)q}-\sqrt{(1-q)p}\right)^2
\geq 1-\frac{|p-q|^2}{4c(1-c)}.
\end{split}
\end{align}
\end{proof}

The following simple lemma provides a lower bound on the square root
that is used in the proof of Lemma~\ref{le:coupling_const} below.
\begin{lemma}\label{le:sqrt}
Let $s, t $ be real numbers such that $s+t\geq -1$.
Then the bound 
\begin{align}
\sqrt{1+s+t}\geq 1-|s|+\frac{t}{2}-\frac{t^2}{2}
\end{align}
holds.
\end{lemma}
\begin{proof}
First, we note that elementary manipulations show that for all $t\in \R$
the bound 
\begin{align}\label{eq:sq_root}
\sqrt{\max(1+t, 0)}\geq 1+\frac{t}{2}-\frac{t^2}{2}
\end{align}
holds. We now consider $s>0$. In this case, we can conclude
\begin{align}
\sqrt{1+s+t}\geq \sqrt{\max(1+t, 0)}\geq 1+\frac{t}{2}-\frac{t^2}{2}
\geq 1+\frac{t}{2}-\frac{t^2}{2} - |s|.
\end{align}
Note that $x-y=(\sqrt{x}-\sqrt{y})(\sqrt{x}+\sqrt{y})>(\sqrt{x}-\sqrt{y})$
if $x>1$ and $x>y$.
This implies for $s<0$ and $t\geq 0$ that
\begin{align}
\sqrt{1+t+s}\geq \sqrt{1+t} - |s| \geq  1+\frac{t}{2}-\frac{t^2}{2} - |s|.
\end{align}
Finally, we consider the case $t,s<0$ where we get from \eqref{eq:sq_root}
\begin{align}
\sqrt{1+t+s}
\geq 1 +\frac{t+s}{2} -\frac{(t+s)^2}{2}
=1+\frac{t}{2} -\frac{t^2}{2} +s\left(\frac{1-t-s}{2}\right)
\geq 1+\frac{t}{2} -\frac{t^2}{2} -|s|
\end{align}
using $-t-s\leq 1$ in the last step.
\end{proof}
We also state a simple fact on the relation of partial trace and operators.
\begin{lemma}\label{le:partial_operator}
Let $\rho$ be an operator on the system $Q\otimes R$. Let $O$ be a linear operator on $Q$.
Then
\begin{align}
\tr_R ((O\otimes \id)\rho)= O \tr_R(\rho).
\end{align}
\end{lemma}
\begin{proof}
By linearity it is sufficient to consider $\rho=S\otimes T$. But then
\begin{align}
\tr_R ((O\otimes \id)\rho)= \tr_R (OS\otimes T)
= \tr(T) OS=O\tr_R(\rho).
\end{align}
\end{proof}
Finally, we state an elementary result about the reward vectors $\pp^j$.
\begin{lemma}\label{le:reward}
Let $\pp^j\in [0,1]^N$ be reward vectors as defined at the beginning of Section~\ref{sec:overview}, i.e., there is $\pp=(p_0,\ldots, p_N)\in [0,1]^{N+1}$
with $p_0>p_1>p_2\geq \ldots\geq p_N$ and then $\pp^j_i=p_i$ for $1\leq i\neq j$ and $\pp^j_j=p_0$. Assume that 
$p_0-p_1=p_1-p_2$ and let $\Delta_i=p_0-p_i$.
Then for all $1\leq j\leq N$
\begin{align}
   \frac14H(\pp^0)\leq  H(\pp^j)\leq 2H(\pp^0)
\end{align}
\end{lemma}
\begin{proof}
We note  that  for $j\geq 1$
 \begin{align}
 \begin{split}\label{eq:Hp1}
 H(\pp^j)&=\sum_{i\neq j} (p_0-p_i)^{-2}\leq \sum_{i\geq 1} \Delta_i^{-2}
 = (p_0-p_1)^{-2} +\sum_{i > 1} (p_0-p_i)^{-2}
 \\
 &\leq 2\sum_{i > 1} (p_1-p_i)^2=2H(\pp^0)
 \end{split}
 \end{align}
 where we used $p_0-p_1=p_1-p_2$ in the last step.
 Similarly, we obtain
 \begin{align}
 \begin{split}\label{eq:Hp2}
 \sum_{i\geq 1} \Delta_i^{-2}&\geq H(\pp^j)=  \sum_{i\neq j} (p_0-p_i)^{-2} =
  \sum_{i\neq j} (p_1-p_i+\Delta_1)^{-2}
 \\
 &\geq \sum_{i>1} (p_1-p_i+\Delta_1)^{-2}
 \geq \frac14 \sum_{i>1} (p_1-p_i)^2=\frac14 H(\pp^0).
 \end{split}
 \end{align}
Here we used in the third step that $p_1-p_1+\Delta_1\leq p_1-p_i+\Delta_1$ for any $i\geq 1$ and $p_1-p_i+\Delta_1\leq 2(p_1-p_i)$ in the following inequality.
This ends the proof.
\end{proof}

\section{Proof of Theorem~\ref{th:faulty_grover}}\label{app:channel_oracles}
In this section we provide the missing parts of the proof 
of Theorem~\ref{th:faulty_grover}. 
To have a complete proof in one place, we repeat the 
parts that are already contained in the main part of the paper.

\begin{proof}[Proof of Theorem~\ref{th:faulty_grover}]
Denote by $\Phi_U$ the quantum channel acting by the unitary $U$, i.e., $\Phi_U(\rho)=U\rho U^\dagger$. We consider an algorithm
acting by $\Phi_{U_{T}}\circ \mc{F}\circ \Phi_{U_{T-1}}\circ\ldots \circ \Phi_{U_1}\circ \mc{F}\circ \Phi_{U_0}(\rho_0)$ on some initial state $\rho_0=\ket{\Omega}\bra{\Omega}$ for some pure state $\Omega$.
We define 
\begin{align}
\tilde\rho_t^i = \Phi_{U_t}( \rho_t^i), \quad \rho_t^i = \mc{F}_i( \tilde{\rho}_{t-1}^i),\quad \rho_0^i=\rho_0.
\end{align}
Note that for $i=0$ the state remains pure
during the entire algorithm, and we denote it by $\psi_t$ and $\tilde\psi_t$.
We now define 
\begin{align}
R_t^i &=\tr(\rho_t^i)^2,
\\
F_t^i &= F(\rho^i_t, \rho^0_t),
\end{align}
i.e., the purity of the state and the fidelity (defined by $F(\sigma, \rho)=\left( \tr\sqrt{\sqrt{\rho}\sigma\sqrt{\rho}}\right)^2$) of the state with respect to the state corresponding to the trivial oracle. 
As explained in the main part of this paper, we will now show that any decrease in the fidelity must be paid for with a decrease in purity of $\rho_t^j$.
Fidelity and purity are invariant under unitary maps and therefore $R_t^i = \tr(\tilde\rho_t^i)^2$ and $F_t^i = F(\tilde\rho^i_t, \tilde\rho^0_t)$.
We control, using that $O_i$ is unitary,
\begin{align}
\begin{split}\label{eq:change_purity}
R_{t-1}^i - R_{t}^i &=\tr(\tilde{\rho}_{t-1}^i)^2- \tr(\rho_t^i)^2
\\
&=
\tr(\tilde{\rho}_{t-1}^i)^2- \tr(pO_i\tilde{\rho}_{t-1}^iO_i^\dagger + (1-p)\tilde{\rho}_{t-1}^i)^2
\\
&=
\tr(\tilde{\rho}_{t-1}^i)^2- (p^2+(1-p)^2)\tr(\tilde{\rho}_{t-1}^i)^2 - 2p(1-p)  \tr(O_i\tilde{\rho}_{t-1}^iO_i^\dagger  \tilde{\rho}_{t-1})
\\
&=
2p(1-p)\left(\tr(\tilde{\rho}_{t-1}^i)^2 -  \tr(O_i\tilde{\rho}_{t-1}^iO_i^\dagger  \tilde{\rho}_{t-1})\right)
\\
&=
p(1-p) \tr(\tilde{\rho}_{t-1}^i-O_i\tilde{\rho}_{t-1}^iO_i^\dagger )^2
\end{split}\end{align}
Similarly we estimate the change in fidelity using that $\rho_0^t$ is a pure state and $\tilde{\psi}_{t-1}=\psi_t$ (because $\mc{E}_0=\id$)
\begin{align}
\begin{split}\label{eq:delta_fidelity}
F_{t-1}^i-F_t^i &= F(\tilde\rho^i_{t-1}, \tilde\rho^0_{t-1}) -F(\rho^i_t, \rho^0_t)
\\
&=
\langle \tilde{\psi}_{t-1}, \tilde{\rho}^i_{t-1}\tilde{\psi}_{t-1}\rangle - 
\langle{\psi}_{t}, {\rho}^i_{t}{\psi}_{t}\rangle
\\
&=
\langle {\psi}_{t}, (\tilde{\rho}^i_{t-1} - (1-p) \tilde{\rho}^i_{t-1} - p O_i  \tilde{\rho}^i_{t-1}O_i^\dagger){\psi}_{t}\rangle 
\\
&=
p\langle {\psi}_{t}, (\tilde{\rho}^i_{t-1} - O_i  \tilde{\rho}^i_{t-1}O_i^\dagger){\psi}_{t}\rangle .
\end{split}
\end{align}
We now relate the change of $R_t$ and $F_t$.
Suppose that there is an orthogonal projection $P$ and a unitary $O$ such that 
$(\id-P)O=\id-P$, i.e., $O$ acts trivially on the complement of the image of $P$.
Then Lemma~\ref{le:projection}
in Appendix~\ref{app:auxiliary} establishes the bound
\begin{align}\label{eq:le_projection_rewrite}
\left|\bra{\p}\left(O\sigma O^\dagger-\sigma\right) \ket{\p}\right|
\leq
2\lVert P\p\rVert \lVert \p\rVert  \left(\tr\left(O\sigma O^\dagger - \sigma\right)^2\right)^{\frac12}.
\end{align}

We define projections $P_i = \ket{i}\bra{i}\otimes \id$.
Note that, by definition of $O_i$ we have $(\id-P_{i})O_i = \id-P_i$.
Applying Lemma~\ref{le:projection}, i.e., 
\eqref{eq:le_projection_rewrite} (with $P=P_i$, $\sigma=\tilde{\rho}^i_{t-1}$, $\p=\psi_t$) we  can continue to estimate
\eqref{eq:delta_fidelity}  as follows
\begin{align}
\begin{split}\label{eq:fid_pur_relation2}
F_{t-1}^i-F_t^i 
&\leq 
2p \lVert P_i{\psi}_{t}\rVert\left( \tr(\tilde{\rho}^i_{t-1} - O_i  \tilde{\rho}^i_{t-1}O_i^\dagger)^2\right)^{\tfrac12}
\\
&\leq 
2 \lVert P_i{\psi}_{t}\rVert \sqrt{\frac{p}{1-p}} \sqrt{R_{t-1}^i - R_{t}^i }
\end{split}
\end{align}

Note that the initial values of $F$ and $R$ are $F_0^i=R_0^i=1$ and $R_t^i\geq 0$. Thus we can conclude
\begin{align}
\begin{split}
\sum_i (F_0^i - F_T^i)
&=
\sum_{i,t } (F_{t-1}^i-F_t^i)
\\
&\leq \sum_{i,t} 2 \lVert P_i{\psi}_{t}\rVert \sqrt{\frac{p}{1-p}} \sqrt{R_{t-1}^i - R_{t}^i }
\\
&\leq 2 \sqrt{\frac{p}{1-p}}
\left(\sum_{i,t} \lVert P_i{\psi}_{t}\rVert^2\right)^{\tfrac12}
\left(\sum_{i,t}R_{t-1}^i - R_{t}^i\right)^{\tfrac12}
\\
&\leq
2 \sqrt{\frac{p}{1-p}}
\left(\sum_{t} \lVert {\psi}_{t}\rVert^2\right)^{\tfrac12}
\left(\sum_{i}R_{0}^i - R_{T}^i\right)^{\tfrac12}
\\
&\leq
2 \sqrt{\frac{p}{1-p}}\sqrt{T}\sqrt{N}.
\end{split}
\end{align}
Finally we use our assumption that the algorithm is able to decide whether the oracle is trivial $\mc{E}=\mc{E}_0$ or not with probability $1-\delta$ for some $\delta<1/2$.  From here we can conclude that the output of the algorithm 
for oracles $\mc{F}^0$ and $\mc{F}^j$ must be sufficiently different. 
Formally success of the algorithm implies (see \eqref{eq:trace_success} and \eqref{eq:trace_fidelity} in
Appendix~\ref{app:distance}) that for each
$i$
\begin{align}
1-2\delta \leq T( \rho_T^i , \rho_T^0)
\leq \sqrt{1-F(\rho_T^i, \rho_T^0)}
\end{align}
where $T(\rho, \sigma) = \tfrac12 \lVert \rho-\sigma\rVert_{\tr}$ denotes the trace distance.
This implies the bound 
\begin{align}\label{eq:fidelity_delta}
F_T^i \leq 4\delta(1-\delta).
\end{align}
We conclude that 
\begin{align}
2 \sqrt{\frac{p}{1-p}}\sqrt{T}\sqrt{N}\geq N (1- 4\delta(1-\delta))
\Rightarrow T\geq \frac{N(1-p)(1-4\delta(1-\delta))^2}{p}.
\end{align}

\end{proof}

\section{Fidelity Loss of Oracle Calls}
\label{app:proof_fid_lemma}
In this section we prove Lemma~\ref{le:fidelity}, i.e.,
 the loss in fidelity 
when applying the oracles $\mc{F}_i^p$ and $\mc{F}_i^q$ to density matrices $\rho$ and $\sigma$ where we recall that $\mc{F}_i^p(\rho)=(1-p)\rho+pO_i\rho O_i^\dagger$.
We also prove the slight improvement stated in Corollary~\ref{co:fid2}.
First, however, we consider a simpler version of this result in
Lemma~\ref{le:fidelity_simple} below, which is 
required in the proof of Theorem~\ref{th:lower_emp}.
\subsection{Fidelity bound for invariant operators}
Here we discuss a lemma that is useful to bound the loss in fidelity $\sqrt{F}(\rho,\sigma)-\sqrt{F}(O\rho O^\dagger, \sigma)$ when it is known that $O\psi = \psi$ for many states $\psi$
(the eigenvalue 1 has large multiplicity).
Note that in terms of the oracles $\mc{F}_i^p$ this corresponds to the case $p=1$ and $q=0$.
\begin{lemma}\label{le:fidelity_simple}
Let $O$ be a unitary operator and $P$ a hermitian projection such that
\begin{align}
P(O-\id)=(O-\id),
\end{align}
i.e., $1-P$ projects on a subspace of the eigenspace of eigenvalue 1 of $O$
and $[P,O]=0$.
Let $\rho, \sigma$ be two density matrices. Then the bound
\begin{align}
\sqrt{F}(\rho,\sigma)-\sqrt{F}(\rho,O\sigma O^\dagger)
\leq 2\sqrt{\tr(P\rho)\tr(P\sigma)}
\end{align}
holds.
\end{lemma}
\begin{proof}
Call the system on which  $\rho, \sigma$ act $Q$.
Let $R$ be a copy of $Q$.
  Let $\p$ and $\psi$ be purifications of $\rho$ and $\sigma$ on the system $QR$ such that 
$\sqrt{F}(\rho,\sigma) = \langle \p, \psi\rangle $. 
Then 
$(O\otimes \id)\psi$ is a purification of $\sigma$ and 
we get
\begin{align}
\begin{split}
\sqrt{F}(\rho,O\sigma O^\dagger)
&\geq \langle \p, (O\otimes \id)\psi\rangle
= \langle \p, \psi\rangle
- \langle \p, ((O-\id)\otimes \id)\psi\rangle
\\
&=\sqrt{F}(\rho,\sigma)
-\langle \p, (P(O-\id)\otimes \id)\psi\rangle
\end{split}
\end{align}
We now bound using $[P,O]=0$ and Lemma~\ref{le:partial_operator}
\begin{align}
\begin{split}
\langle \p, (P(O-\id)\otimes \id)\psi\rangle
&\leq 
\langle (P\otimes \id)\p, (P\otimes \id)((O-\id)\otimes \id)\psi\rangle
\\
&\leq \lVert (P\otimes \id)\p\rVert \cdot 
\lVert((O-\id)\otimes \id) (P\otimes \id)
\psi\rVert
\\
&\leq 
2\left(\tr\tr_R((P\otimes \id)\ket{\p}\bra{\p}(P\otimes \id)^\dagger)
\tr\tr_R((P\otimes \id)\ket{\psi}\bra{\psi}(P\otimes \id)^\dagger)\right)^\frac12
\\
&\leq 2\sqrt{\tr(P\rho)\tr(P\sigma)}.
\end{split}
\end{align}
\end{proof}

\subsection{Proof of Lemma~\ref{le:fidelity}}

\begin{proof}
Call the system on which  $\rho, \sigma$ act $Q$ and we denote
$\bar{\rho}= \mc{F}_i^p(\rho)$ and $\bar{\sigma}=\mc{F}_i^{q}(\sigma)$.
Let $R$ be a copy of $Q$.
  Let $\p$ and $\psi$ be purifications of $\rho$ and $\sigma$ on the system $QR$ such that 
$\sqrt{F}(\rho,\sigma) = \langle \p, \psi\rangle $. 
We use the shorthand $\bar{O}_i=O_i\otimes \id_R$ in the following.
Let $S$ be a system consisting of a single qubit. We consider the following state on the system $QRS$
\begin{align}
\omega = \sqrt{1-p}\ket{\p, 0} + \sqrt{p} \ket{\bar{O}_i\p, 1}.
\end{align}
 It is easy to check that $\omega$ is a purification of $\bar{\rho}$
\begin{align}
\tr_{QR} \ket{\omega}\bra{\omega} = \tr_Q \left((1-p)\ket{\p}\bra{\p}
+ p \ket{\bar{O}_i\p}\bra{\bar{O}_i\p}\right)=
(1-p)\rho + p O_i \rho O_i^\dagger=\bar\rho.
\end{align}
To obtain a purification of $\bar\sigma$ we define for an angle $\alpha$ 
the state
\begin{align}
\zeta = \sqrt{1-q}\cos(\alpha)\ket{\psi, 0}+
\sqrt{q}\sin(\alpha) \ket{\bar{O}_i\psi, 0}
-\sqrt{1-q}\sin(\alpha)\ket{\psi, 1}
+\sqrt{q}\cos(\alpha)\ket{\bar{O}_i\psi, 1}.
\end{align}
It is easy to check that $\lVert \zeta\rVert =1$. 
We now check that this is a purification of $\bar{\sigma}$.
Note that the cross terms $\ket{\psi}\bra{\bar{O}_i\psi}$
and  $\ket{\bar{O}_i\psi}\bra{\psi}$ cancel, and thus
\begin{align}
\begin{split}
\tr_S \ket{\zeta}\bra{\zeta}
&=(1-q)\cos^2(\alpha) \ket{\psi}\bra{\psi}
+ q\sin^2(\omega)\ket{\bar{O}_i\psi}\bra{\bar{O}_i\psi}
\\
&\qquad+(1-q)\sin^2(\alpha)
\ket{\psi}\bra{\psi}
+ q\cos^2(\alpha)\ket{\bar{O}_i\psi}\bra{\bar{O}_i\psi}.
\\
&=(1-q)\ket{\psi}\bra{\psi}
+ q\ket{\bar{O}_i\psi}\bra{\bar{O}_i\psi}.
\end{split}
\end{align}
We calculate
\begin{align}
\begin{split}
\bra{\omega}\ket{\zeta}
&=\sqrt{(1-p)(1-q)}\cos(\alpha)\bra{\p}\ket{\psi}
+ \sqrt{(1-p)q}\sin(\alpha)\bra{\p}\ket{\bar{O}_i\psi}
\\
& \qquad -
\sqrt{p(1-q)} \sin(\alpha)\bra{\bar{O}_i\p}\ket{\psi}
+
\sqrt{pq} \cos(\alpha) \bra{\bar{O}_i\p}\ket{\bar{O}_i\psi}.
\end{split}
\end{align}
Using that $O_i$ is self-adjoint and unitary we obtain
\begin{align}
\bra{\eta}\ket{\zeta}
&=\left(\sqrt{(1-p)(1-q)}+\sqrt{pq}\right)\cos(\alpha)\bra{\p}\ket{\psi}
+ \left(\sqrt{(1-p)q}-\sqrt{p(1-q)} \right) \sin(\alpha)\bra{\p}\ket{\bar{O}_i\psi}.
\end{align}
Now we set
\begin{align}
\sin(\alpha) = \sqrt{(1-p)q} - \sqrt{p(1-q)}	\\
\cos(\alpha) = \sqrt{(1-p)(1-q)} + \sqrt{pq}
\end{align}
and obtain 
\begin{align}
\begin{split}
\bra{\omega}\ket{\zeta}
&=\cos^2(\alpha)\bra{\p}\ket{\psi}
+ \sin(\alpha)^2\bra{\bar{O}_i\p}\ket{\psi}
\\
&=\bra{\p}\ket{\psi}
+ \sin(\alpha)^2 \bra{\bar{O}_i\p-\p}\ket{\psi}.
\end{split}
\end{align}
We conclude that 
\begin{align}\label{eq:le_fid1}
\sqrt{F}(\bar{\rho},\bar{\sigma})
\geq |\bra{\eta}\ket{\zeta}|\geq
\sqrt{F}({\rho},{\sigma})
-|\sin^2(\alpha) \bra{\bar{O}_i\p-\p}\ket{\psi}|.
\end{align}
As above we write $\bar{P}_i=P_i\otimes \id$ and $P_{-i}=\id-P_i$ and we get
\begin{align}
\bar{O}\p-\p=(\bar{P}_{i}+\bar{P}_{-i})(\bar{O}_i\p-\p)
= \bar{P}_{i}(\bar{O}_i\p-\p)+ 
\bar{P}_{-i}(\bar{O}_i\p-\p)
=
\bar{P}_{i}(\bar{O}_i\p-\p).
\end{align}
Then we control using Lemma~\ref{le:partial_operator}
\begin{align}
\begin{split}\label{eq:le_fid2}
\left|\bra{\bar{O}_i\p-\p}\ket{\psi}\right|^2
&= \left|\bra{(\bar{O}_i\p-\p)}\ket{\bar{P}_{i}\psi}\right|^2
\leq 
\lVert \bar{P}_i( \bar{O}_i\p-\p)\rVert^2 \, \lVert \bar{P}_{i}\psi\rVert^2
\\
&=
\tr \bar{P}_i\left(\ket{\bar{O}_i\p-\p}\bra{ \bar{O}_i\p-\p}\bar{P}_i^\dagger\right)
\tr \left(\ket{\bar{P}_{i}\psi}\bra{ \bar{P}_{i}\psi}\right)
\\
&=
\tr\tr_R \left((\bar{O}_i-\id)\bar{P}_i\ket{\p}\bra{\p}\bar{P}_i^\dagger(
 \bar{O}_i^\dagger
-\id)\right)
\tr \tr_R\left(\bar{P}_{i}\ket{\psi}\bra{ \psi}\bar{P}_{i}^\dagger\right)
\\
&\leq 4 \tr\left(P_i\rho\right)
\tr \left({P}_{i}\sigma\right).
\end{split}
\end{align}
Using the bound \eqref{eq:bound_sin} from Lemma~\ref{le:trivial_bound} 
implies
\begin{align}\label{eq:le_fid3}
|\sin(\alpha)|= \sqrt{(1-p)q} - \sqrt{p(1-q)}
\leq \frac{|p-q|}{2\sqrt{\eta(1-\eta)}}.
\end{align}
From \eqref{eq:le_fid1}, \eqref{eq:le_fid2}, and \eqref{eq:le_fid3}
we conclude that 
\begin{align}
\sqrt{F}(\bar{\rho},\bar{\sigma})
\geq |\bra{\eta}\ket{\zeta}|\geq
\sqrt{F}({\rho},{\sigma})
-(p-q)^2\frac{\sqrt{\tr\left(P_i\rho\right)
\tr \left({P}_{i}\sigma\right)}}{2\eta(1-\eta)}.
\end{align}
\end{proof}
\subsection{Proof of Corollary~\ref{co:fid2}}\label{sec:co_fid2proof}
\begin{proof}
We write $\mc{E}_i = \mc{E}^{\pp^i}$.
We obtain using Corollary~\ref{co:fid1}  (which can be applied since 
the bit flip operation is self adjoint)
 for $m$ copies
\begin{align}
\begin{split}
\sqrt{F}(\mc{E}_i(\rho)^{\bigotimes  m},
\mc{E}_0(\rho)^{\bigotimes m})
&=
\sqrt{F}(\mc{E}_i(\rho),
\mc{E}_0(\rho))^m
\\
&\geq \left(1-\frac{(p_0 - p_i)^2}{c(1-c)}\tr(P_k\rho)\right)^m
\geq 1 - m \frac{4\Delta_k^2}{\eta(1-\eta)}\tr(P_k\rho)
\end{split}
\end{align}
where we used the Bernoulli inequality  in the last step.
From
\begin{align}
1-2\delta \leq
T( \mc{E}_i(\rho)^{\bigotimes  m} ,\mc{E}_0(\rho)^{\bigotimes m})
\leq \sqrt{1-F(\mc{E}_i(\rho)^{\bigotimes  m},
\mc{E}_0(\rho)^{\bigotimes m})}
\end{align}
we conclude that 
\begin{align}\label{eq:hellstroem_fidelity}
2\sqrt{\delta(1-\delta)}\geq 
\sqrt{F}(\mc{E}_i(\rho)^{\bigotimes  m},
\mc{E}_0(\rho)^{\bigotimes m})
\geq 1 - m \frac{4\Delta_i^2}{\eta(1-\eta)}\tr(P_i\rho).
\end{align}
Equivalently
\begin{align}
\tr(P_i\rho) \geq \frac{\eta(1-\eta)(1-2\sqrt{\delta(1-\delta)})}{4m \Delta_i^2}
\end{align}
Using $\tr(\sum_i P_i\rho)=1$ and summing over $i\geq 2$
we conclude
\begin{align}
1\geq \frac{\eta(1-\eta)(1-2\sqrt{\delta(1-\delta)})}{4m} \sum_{i\geq 2}^N \Delta_i^{-2}.
\end{align}
Using $\sum_{i\geq 2}^N \Delta_i^{-2}\geq \tfrac14 \sum_{i=2}^N (p_1-p_i)^{-2}$
ends the proof.
\end{proof}

\section{State Decompositions and Fidelity Bounds}\label{app:coupling}
In this section we prove Lemma~\ref{le:coupling_const}
and state a slight extension that will 
be used in the proof of Theorem~\ref{th:main}.
Let us restate the lemma including the definition
of $\psi_0$ and $\psi_1$ for the convenience of the reader.

\begin{lemma}[Lemma \ref{le:coupling_const} restated]\label{le:coupling_const2}
Consider a density matrix $\rho$ and a pure state $\psi$. Consider a unitary and self-adjoint map 
$U$ and the channel $\mc{E}_U^p$ defined by 
$\mc{E}_U^p(\rho)= pU\rho U^\dagger+(1-p)\rho$.
Let $0<\eta<1/2$ and $p,q\in[\eta, 1-\eta]$.
We define 
\begin{align}
\bar{\rho}_0 &= (1-p)\rho,\quad \rho_0=\rho, \qquad \bar{\rho}_1=pU\rho U^\dagger,
 \quad \rho_1= U\rho U^\dagger
\end{align}
and 
\begin{alignat}{2}
\bar\psi_0&= \sqrt{1-q}\cos(\alpha)\ket{\psi} + \sqrt{q}\sin(\alpha) \ket{U\psi} ,\quad 
&&\bar\psi_1 =
\sqrt{q}\cos(\alpha)\ket{U\psi}-\sqrt{1-q}\sin(\alpha)\ket{\psi}
\\
\psi_0 &= \bar{\psi}_0/\lVert \bar\psi_0\rVert,
\quad 
&&\psi_1 = \bar{\psi}_1/\lVert \bar\psi_1\rVert
\end{alignat}
where 
\begin{align}\label{eq:def_angle}
\cos(\alpha)= \sqrt{pq}+\sqrt{(1-p)(1-q)}, \quad \sin(\alpha)=\sqrt{(1-p)q}-\sqrt{(1-q)p}.
\end{align}
We define $q'=\lVert \bar\psi_1\rVert^2$.
Then
\begin{align}\label{eq:decomp}
\mc{E}_U^q(\ket{\psi}\bra{\psi})=
q'\ket{\psi_1}\bra{\psi_1}+(1-q')\ket{\psi_0}\bra{\psi_0}.
\end{align}
Let $S= \sqrt{F}(\rho, \ket{\psi}\bra{\psi})=\sqrt{\bra{\psi}\rho\ket{\psi}}$
denote the initial fidelity.
Then the following bound holds
\begin{align}
\begin{split}\label{eq:fidelity_lemma_bound}
\sqrt{F}\Big(\mc{E}_U^p(\rho),\mc{E}_U^q(\ket{\psi}\bra{\psi})\Big)
&\geq
\sqrt{(1-p)(1-q')}\sqrt{F}(\rho_0,\ket{\psi_0}\bra{\psi_0})
+ \sqrt{pq'}\sqrt{F}(\rho_1, \ket{\psi_1}\bra{\psi_1})
\\
&\geq S
- \frac{(p-q)^2 |(\psi,( U\rho U^\dagger - \rho)\psi)|}{2\eta S}
- \frac{(p-q)^2 |\mathrm{Re}(\psi, U\rho \psi)-(\psi,\rho\psi)|^2}{8\eta^2S^3}.
\end{split}
\end{align}
\end{lemma}
To clarify the origin of the expressions for $\sin(\alpha)$ and $\cos(\alpha)$ we remark that if we choose $\beta,\gamma$ such that $\sqrt{p}=\sin(\beta)$,
$\sqrt{q}=\sin(\gamma)$, then $\alpha=\gamma-\beta$. 
This also explains the simplifications in the formula below, in particular
\eqref{eq:add_thm0} and \eqref{eq:add_thm1} below are just the
trigonometric identities for angle sums. 
\begin{proof}
First, simple algebra shows $\lVert \bar{\psi}_0\rVert^2 = 1-q'$ and
\begin{align}
\mc{E}_U^p(\rho) &=\bar{\rho}_0+\bar{\rho}_1=(1-p)\rho_0 +p\rho_1\\
\mc{E}_U^q(\ket{\psi}\bra{\psi})
&= 
\ket{\bar{\psi}_0}\bra{\bar{\psi}_0}
+
\ket{\bar{\psi}_1}\bra{\bar{\psi}_1}
= 
(1-q')\ket{{\psi}_0}\bra{{\psi}_0}+
q'\ket{{\psi}_1}\bra{{\psi}_1},
\end{align}
in particular \eqref{eq:decomp} holds. 
Strong concavity of the fidelity implies the first bound
of \eqref{eq:fidelity_lemma_bound}.
We now address the second estimate. We can express, using that $U=U^\dagger$ is self adjoint
\begin{align}
\begin{split}\label{eq:lefid1}
&\sqrt{(1-p)(1-q')}\sqrt{F}(\rho_0, \psi_0)
=
\sqrt{(1-p)(1-q')\bra{\psi_0}\rho\ket{\psi_0}}
= 
\sqrt{(1-p)\bra{\bar\psi_0}\rho\ket{\bar\psi_0}}
\\
&=
\sqrt{(1-p)}\Big((1-q)\cos^2(\alpha)\bra{\psi}\rho\ket{\psi}
+ q\sin^2(\alpha)\bra{\psi}U\rho U^\dagger\ket{\psi}
\\
&\qquad \qquad
+ 2\sqrt{q(1-q)}\cos(\alpha)\sin(\alpha)\mathrm{Re}(\bra{\psi}\rho U
\ket{\psi})
   \Big)^\frac12
   \\
  &=
\sqrt{(1-p)}\Big(\left((1-q)\cos^2(\alpha)+q\sin^2(\alpha)+2\sqrt{q(1-q)}\cos(\alpha)\sin(\alpha)\right)S^2
\\
&\qquad \qquad+ q\sin^2(\alpha)\bra{\psi}(U\rho U^\dagger- \rho)\ket{\psi}
\\
&\qquad \qquad
+ 2\sqrt{q(1-q)}\cos(\alpha)\sin(\alpha)\mathrm{Re}(\bra{\psi}
(\rho U - \rho)
\ket{\psi})
   \Big)^\frac12
\end{split}
\end{align}
Now we calculate using the definition of $\alpha$ 
\begin{align}\label{eq:add_thm0}
\begin{split}
\sqrt{ (1-p)}&\left(\sqrt{1-q}\cos(\alpha)+\sqrt{q}\sin(\alpha)\right)
\\
&=\sqrt{(1-p)(1-q)}(\sqrt{(1-p)(1-q)}+\sqrt{pq})
+\sqrt{(1-p)q}(\sqrt{(1-p)q}-\sqrt{(1-q)p})
\\
&=(1-p)(1-q)+(1-p)q=1-p
.
\end{split}
\end{align}
Then we obtain
\begin{align}\label{eq:coupling1}
\begin{split}
&\sqrt{(1-p)(1-q')}\sqrt{F}(\rho_0, \psi_0)
\\
&=(1-p)S\Big(1+ \frac{q\sin^2(\alpha)\bra{\psi}(U\rho U^\dagger- \rho)\ket{\psi}
+ 2\sqrt{q(1-q)}\cos(\alpha)\sin(\alpha)\mathrm{Re}(\bra{\psi}
(\rho U - \rho)
\ket{\psi})}{(1-p)S^2}
   \Big)^\frac12.
\end{split}
\end{align}
Similarly the second term can be expressed as
\begin{align}
\begin{split}\label{eq:lefid2}
&\sqrt{pq'}\sqrt{F}(\rho_1,\psi_1)=\sqrt{p}\Big(\left(q\cos^2(\alpha)+(1-q)\sin^2(\alpha)-2\sqrt{q(1-q)}\cos(\alpha)\sin(\alpha)\right)S^2
\\
&\qquad \qquad+ (1-q)\sin^2(\alpha)\bra{\psi}(U\rho U^\dagger- \rho)\ket{\psi}
- 2\sqrt{q(1-q)}\cos(\alpha)\sin(\alpha)\mathrm{Re}(\bra{\psi}
(\rho U - \rho)
\ket{\psi})
   \Big)^\frac12.
\end{split}
\end{align}
We can calculate
\begin{align}\label{eq:add_thm1}
\begin{split}
\sqrt{p}&\left(\sqrt{q}\cos(\alpha)-\sqrt{1-q}\sin(\alpha)\right)
\\
&=\sqrt{p}\left(\sqrt{q}(\sqrt{(1-p)(1-q)}+\sqrt{pq})-\sqrt{1-q}(\sqrt{(1-p)q}-\sqrt{(1-q)p}\right)
\\
&= 
pq+(1-q)p=p.
\end{split}
\end{align}
we get
\begin{align}
\begin{split}\label{eq:coupling2}
&\sqrt{pq'}\sqrt{F}(\rho_1,\psi_1)=
\\
&=pS\Big(1+ \frac{(1-q)\sin^2(\alpha)\bra{\psi}(U\rho U^\dagger- \rho)\ket{\psi}
- 2\sqrt{q(1-q)}\cos(\alpha)\sin(\alpha)\mathrm{Re}(\bra{\psi}
(\rho U - \rho)
\ket{\psi})}{pS^2}
   \Big)^\frac12.
\end{split}
\end{align}
We continue to estimate the square root terms. Note that by first order Taylor expansion the mixed last terms of the two expressions \eqref{eq:coupling1} and \eqref{eq:coupling2} cancel and only the first term  and higher order corrections remains. 
To make this rigorous 
we use Lemma~\ref{le:sqrt} in Appendix~\ref{app:auxiliary} which states
that for $s+t\geq -1$ the bound
\begin{align}\label{eq:square_bound_again}
\sqrt{1+s+t}\geq 1+\frac{t}{2}-\frac{t^2}{2}-|s|
\end{align}
holds. We apply this to \eqref{eq:coupling1} and \eqref{eq:coupling2}
where $s$ corresponds to the term involving $\bra{\psi}(U\rho U^\dagger- \rho)\ket{\psi}$ and $t$ corresponds to the term 
involving $\mathrm{Re}(\bra{\psi}
(\rho U - \rho)
\ket{\psi})^2$. 
Then we get from \eqref{eq:coupling1}
\begin{align}
\begin{split}\label{eq:lefid3}
&\sqrt{(1-p)(1-q')}\sqrt{F}(\rho_0, \psi_0)
\\
&\geq (1-p)S\Big(1+ \frac{-q\sin^2(\alpha)|\bra{\psi}(U\rho U^\dagger- \rho)\ket{\psi}|
+ \sqrt{q(1-q)}\cos(\alpha)\sin(\alpha)\mathrm{Re}(\bra{\psi}
(\rho U - \rho)
\ket{\psi})}{(1-p)S^2}
\\
&\qquad \qquad- \frac{q(1-q)\cos^2(\alpha)\sin^2(\alpha)\mathrm{Re}(\bra{\psi}
(\rho U - \rho)
\ket{\psi})^2}{2(1-p)^2S^4}
   \Big).
\end{split}
\end{align}
From \eqref{eq:coupling2} we get similarly
\begin{align}
\begin{split}\label{eq:coupling3}
&\sqrt{pq'}\sqrt{F}(\rho_1,\psi_1)=
\\
&\geq pS\Big(1+ \frac{-(1-q)\sin^2(\alpha)\bra{\psi}(U\rho U^\dagger- \rho)\ket{\psi}
- \sqrt{q(1-q)}\cos(\alpha)\sin(\alpha)\mathrm{Re}(\bra{\psi}
(\rho U - \rho)
\ket{\psi})}{pS^2}
\\
&\qquad\qquad
-
\frac{{q(1-q)}\cos^2(\alpha)\sin^2(\alpha)\mathrm{Re}(\bra{\psi}
(\rho U - \rho)
\ket{\psi})}{2p^2S^4}
   \Big).
\end{split}
\end{align}
We notice that the linear terms in $\mathrm{Re}(\bra{\psi}
(\rho U - \rho)
\ket{\psi})$  cancel, and we get
\begin{align}
\begin{split}\label{eq:coupling4}
&\sqrt{(1-p)(1-q')}\sqrt{F}(\rho_0, \psi_0)+\sqrt{pq'}\sqrt{F}(\rho_1,\psi_1)
\\
&\geq S - \frac{\sin^2(\alpha)|\bra{\psi}(U\rho U^\dagger- \rho)\ket{\psi}|}{S}
-\left(\frac{1}{p}+\frac{1}{1-p}\right)\frac{q(1-q)\cos^2(\alpha)\sin^2(\alpha)\mathrm{Re}(\bra{\psi}
(\rho U - \rho).
\ket{\psi})^2}{2S^3}.
\end{split}
\end{align}
Finally we can estimate using the assumption $p,q\in[\eta,1-\eta]$ and
Lemma~\ref{le:trivial_bound} 
\begin{align}
\begin{split}\label{eq:fidfinal}
&\sqrt{(1-p)(1-q')}\sqrt{F}(\rho_0, \psi_0)+\sqrt{pq'}\sqrt{F}(\rho_1,\psi_1)
\\
&\quad\geq S - \frac{|p-q|^2\cdot|\bra{\psi}(U\rho U^\dagger- \rho)\ket{\psi}|}{4\eta(1-\eta)S}
-\left(\frac{1}{\eta}+\frac{1}{1-\eta}\right)\frac{|p-q|^2\mathrm{Re}(\bra{\psi}
(\rho U - \rho).
\ket{\psi})^2}{32\eta(1-\eta)S^3}.
\end{split}
\end{align}
This ends the proof.
\end{proof}

The proof of Theorem~\ref{th:main} requires a slight extension of the lemma above.
	Unfortunately, we cannot directly derive the result, but we need to slightly modify the proof of Lemma~\ref{le:fidelity}.
	\begin{corollary}\label{co:fidelity}
Assume the same setting as in Lemma~\ref{le:fidelity}, with the following 
changes. 
We consider another self adjoint traceless operator $\sigma$ and we define
\begin{align}
\bar{\rho}_0=(1-p)(\rho+p\sigma)
\quad
\rho_0=\rho+p\sigma
\quad \bar{\rho}_1=pU\rho U^\dagger + p(1-p)\sigma
\quad 
\rho_1 = U\rho U^\dagger+(1-p)\sigma.
\end{align}
We assume that $\rho_0$ and $\rho_1$ are density matrices, i.e., 
non-negative.
Then the following bound holds
\begin{align}
\begin{split}
&\sqrt{(1-p)(1-q')}\sqrt{F}(\rho_0,\ket{\psi_0}\bra{\psi_0})
+ \sqrt{pq'}\sqrt{F}(\rho_1, \ket{\psi_1}\bra{\psi_1}) -S\geq
\\
&
- \frac{(p-q)^2 |(\psi,( U\rho U^\dagger - \rho)\psi)|}{2\eta S}
- \frac{(p-q)^2 |\mathrm{Re}(\psi, U\rho \psi)-(\psi,\rho\psi)|^2}{8\eta^2S^3}
- \frac{p|\langle \bar\psi_0,\sigma\bar\psi_0\rangle|
+(1-p)|\langle \bar\psi_1,\sigma\bar\psi_1\rangle|
}{S}.
\end{split}
\end{align}
\end{corollary}
	\begin{proof}
	The proof proceeds exactly as the proof of Lemma~\ref{le:fidelity}
	with the following minor modifications. We have to insert an additional term
	$p\langle \bar\psi_0,\sigma\bar\psi_0\rangle$ in \eqref{eq:lefid1}
	and a term $-(1-p) \langle\bar\psi_1,\sigma\bar\psi_1\rangle$
	in \eqref{eq:lefid2}. We carry those terms and when 
	we estimate the square-root terms we add them to the $s$ part 
	in the bound \eqref{eq:square_bound_again}, thus we end up with an additional term $-p|\langle \bar\psi_0,\sigma\bar\psi_0\rangle|/S$
	in \eqref{eq:lefid3} and 
	$-(1-p)|\langle \bar\psi_1,\sigma\bar\psi_1\rangle|/S$
	in \eqref{eq:coupling3} and thus their sum appears in \eqref{eq:fidfinal}. 
	This ends the proof.
	\end{proof}

\section{Complexity bounds for classical bandits with a quantum perspective}
\label{sec:classical}
Before addressing the case of quantum bandits, we revisit the classical bandit problem and give a different proof for the required number of rounds in the fixed confidence setting.
This section serves two purposes.
It shows that the fidelity of probability distributions is a useful distance measure to analyze classical
bandit problems which offers the additional advantage that 
it readily generalizes to quantum states. 
Moreover, this section is a preparation that introduces some notation for the more involved proof in
the quantum setting in the next section. In fact, the proof for the quantum result shares some ideas, and it is essentially based on a combination of the proof given in this section with the optimal fidelity estimates discussed above.
Recall the setting introduced at the beginning of Section~\ref{sec:overview}, in particular the definition of $\pp^j$ (for reference, we recapitulate the setting at the beginning of the next section). 
 Then the following result implies Theorem~\ref{th:classical_bandit_result}.

\begin{theorem}\label{th:classical}
Let $\delta<1/2$. 
Assume that $\pp^j$ are as defined in Section~\ref{sec:overview} with $p_i\in [\eta,1-\eta]$ for
some $\eta>0$.
Any classical algorithm that identifies the best arm
when it is known that the reward vector is in 
$\{\pp^0,\ldots,\pp^n\}$ with probability at least $1-\delta$
requires at least 
\begin{align}
T\geq c H(\pp^1)=c\sum_{j=2}^n \Delta_j^{-2}
\end{align}
rounds where $c=c(\delta,\eta)>0$.
\end{theorem}
Since this result is well known, the proof serves merely pedagogical purposes to illustrate our approach to the quantum setting. 
Therefore, we do not give the most concise presentation, but instead highlight the main difference to the standard proofs of this result.
\begin{proof}
Suppose we are given an algorithm $A$. In each step $t$ the algorithm picks an arm
$a_t$ depending on all earlier outcomes and receives a (binary) reward $r_t\in \{0,1\}$.
We introduce the  variables
$x_t=(a_t, r_t)\in [N]\times \{0,1\}$ encoding the path of the algorithm.
We denote by $z_t=(x_1,\ldots,x_t)$ the entire history of the exploration.
Note that $a_t$ only depends on the outcomes of the previous rounds and therefore is a deterministic function of $z_{t-1}$, i.e., $a_t=a_t(z_{t-1})$. 
When the rewards follow the distribution $\pp^j$ this
induces a distribution on $x_t$ and $z_t$ and we denote the corresponding 
random variables by $Z_t^j$ and $X_t^j$ and the distribution by $\Pd^j$.
 The main idea of the proof is to bound the fidelity of 
 random variables 
$Z_t^j$ and $Z_t^0$ for each $t$ from below. On the other hand, 
we can upper bound the fidelity because the algorithm can 
identify the best arm for the reward distributions $\pp^j$ and $\pp^0$ and 
those arms are different
for $j>1$. Together, those two bounds will imply the claim. 
The proof will rely on the fidelity $\sqrt{F}$ of two discrete probability distributions
$p_x$ and $q_x$ which is defined by 
\begin{align}
\sqrt{F}(p,q)=\sum_x \sqrt{p_xq_x}.
\end{align}
We refer to Appendix~\ref{app:distance} for a brief summary of distance measures, here we only need the definition and the bound by the total variation distance (defined by the first equality)
\begin{align}
\mathrm{d}_{\mathrm{TV}}(p, q) =  \frac12 \sum_x |p_x-q_x|\leq 
\sqrt{1 - F(p,q)}.
\end{align}

We now discuss the simple upper bound on the fidelity after the final round $T$ coming from the assumption that the algorithm succeeds with high probability.
 Let $M_j$ be the disjoint sets of outcomes $z_T$ such that 
arm $j$ is selected by the algorithm. By assumption $\Pd^j(M_j)>1-\delta$
and $\Pd^0(M_1)>1-\delta$ and therefore $\Pd^0(M_j)<\delta$. This implies
for $j>1$ (see Appendix~\ref{app:distance} for a brief summary of distance measures)
\begin{align}
1-2\delta < \Pd^j(M_j)-\Pd^0(M_j)\leq \mathrm{d}_{\mathrm{TV}}( Z_T^j, Z_T^0)
\leq \sqrt{1-F(Z_T^j, Z_T^0)}.
\end{align}
We conclude that 
\begin{align}\label{eq:fidelity_classic_difference}
\sqrt{F}(Z_T^j, Z_T^0)\leq 2\sqrt{\delta(1-\delta)}.
\end{align}
We now bound the fidelity from below. The first bound will not be sufficient to
conclude, but it is nevertheless instructive to understand the difficulties of the quantum setting and the relation to earlier proofs. We will later refine
the following estimates. 
We bound
\begin{align}
\begin{split}\label{eq:fidelity_bound_classic}
\sqrt{F}(Z_t^j, Z_t^0)
&=\sum_{z_t} \sqrt{\Pd^j(z_t) \Pd^0(z_t)}
\\
&= \sum_{r_t = 0}^1 \sum_{z_{t-1}} \sqrt{\Pd^j(r_t\vert a_t(z_{t-1}))\Pd^j(z_{t-1})
\Pd^0(r_t\vert a_t(z_{t-1})\Pd^0(z_{t-1}))}
\\
&=
\sum_{z_{t-1}}   \sqrt{\Pd^j(z_{t-1})\Pd^0(z_{t-1})}
\sum_{r_t=0}^1\sqrt{\Pd^j(r_t\vert a_t(z_{t-1}))
\Pd^0(r_t\vert a_t(z_{t-1}))}.
\end{split}
\end{align}

We now consider two cases $a_t(z_{t-1})=j$ and $a_t(z_{t-1})\neq j$.
In the latter case 
\begin{align}
\Pd^j(r_t\vert a_t(z_{t-1}))=\Pd^0(r_t\vert a_t(z_{t-1}))=p_j
\end{align}
and thus for $a_t(z_{t-1})\neq j$
\begin{align}
\sum_{r_t=0}^1\sqrt{\Pd^j(r_t\vert a_t(z_{t-1}))\Pd^0(r_t\vert a_t(z_{t-1}))}
= \sum_{r_t=0}^1 \Pd^0(r_t\vert a_t(z_{t-1}))=1.
\end{align}
For $a_t(z_{t-1})= j$ we use the simple bound \eqref{eq:bound_cos}
from Lemma~\ref{le:trivial_bound} in Appendix~\ref{app:auxiliary}.
bounding for $p,q\in [c,1-c]$
\begin{align}
\sqrt{F}(\mathrm{Ber}(p),\mathrm{Ber}(q))\geq 1-\frac{|p-q|^2}{4c(1-c)}.
\end{align}
This implies
\begin{align}
\sum_{r_t=0}^1\sqrt{\Pd^j(r_t\vert a_t=j)\Pd^0(r_t\vert a_t=j)}
\geq 1 - \frac{|p_j-p_0|^2}{4\eta(1-\eta)}.
\end{align}
We can now use the last two displays to continue to estimate \eqref{eq:fidelity_bound_classic}
\begin{align}
\begin{split}\label{eq:fidelity_classic_key}
\sqrt{F}(Z_t^j, Z_t^0)
&\geq \sum_{z_{t-1}}   \sqrt{\Pd^j(z_{t-1})\Pd^0(z_{t-1})}
-\sum_{z_{t-1}} \sqrt{\Pd^j(z_{t-1})\Pd^0(z_{t-1})} \bs{1}_{a_t(z_{t-1})=j}
\frac{\Delta_j^2}{4\eta(1-\eta)}
\\
&=\sqrt{F}(Z_{t-1}^j, Z_{t-1}^0)
-\sum_{z_{t-1}} \sqrt{\Pd^j(z_{t-1})\Pd^0(z_{t-1})} \bs{1}_{a_t(z_{t-1})=j}
\frac{\Delta_j^2}{4\eta(1-\eta)}.
\end{split}
\end{align}
Using this iteratively we obtain (using $\sqrt{F}(Z_0^j, Z_0^0)=1$) the bound
\begin{align}\label{eq:fidelity_final_classic}
2\sqrt{\delta(1-\delta)}\geq \sqrt{F}(Z_0^j, Z_0^0)\geq
1-\sum_{t=1}^{T}\sum_{z_{t-1}} \sqrt{\Pd^j(z_{t-1})\Pd^0(z_{t-1})} \bs{1}_{a_t(z_{t-1})=j}
\frac{\Delta_j^2}{4\eta(1-\eta)}.
\end{align}
Now the standard way to proceed from here is to
show that with high probability $\Pd^j(z_{t-1})$ and $\Pd^0(z_{t-1})$ 
are similar using tail bounds for random variables (note that we already control the fidelity).
Suppose that up to small errors we could replace $\Pd^j$ by $\Pd^0$. Then we
could conclude from \eqref{eq:fidelity_final_classic} that 
\begin{align}
\frac{\Delta_j^{2}}{4\eta(1-\eta)}\sum_{t=1}^{T}\sum_{z_{t-1}}\Pd^0(z_{t-1})  \bs{1}_{a_t(z_{t-1})=j}
\geq 1 - 2\sqrt{\delta(1-\delta)}>0.
\end{align}
Dividing by $\Delta_j^2$ and summing over $j$ this would imply
that there is a constant $c>0$ such that
\begin{align}
\sum_j\Delta_j^{-2}\leq c\sum_{t=1}^{T}\sum_{z_{t-1}}\Pd^0(z_{t-1}) \sum_j \bs{1}_{a_t(z_{t-1})=j}=cT.
\end{align}
This approach cannot be extended to the quantum setting. The reason is that
the tail bounds rely on the fact that when we 
use a total of $\mc{O}(H)$ queries
then we 
 cannot query all arms more often than 
$\Delta_j^{-2}$. On the other hand, in the quantum setting
we query in superposition so that we cannot simply count the number of pulls on an arm. 
We now show how the tail bounds can be avoided in a way that can similarly be generalized to the quantum setting. 
Let us denote by $n_j(z_t)=|\{s: a_s(z_t)=j\}$ the number of times we queried the $j$-th arm.
We introduce the decay factor 
\begin{align}
\de^j(z_t)= \left(1-\frac{\Delta_j^2}{4\eta(1-\eta)}\right)^{n_j(z_t)}.
\end{align}
Then we will be interested in bounding from below the following (the lower bound on the fidelity of $\P^j$ and $\P^0$)
\begin{align}
    \sum_{z_t} \sqrt{\Pd^j(z_t) \Pd^0(z_t)}\de^j(z_t)
\end{align}
Introducing this decay factor artificially will allow us to derive stronger bounds. 
Note that 
\begin{align}\label{eq:recursion_classical_decay}
\de_j(z_t)=\de_j(z_{t-1})\left(1-\frac{\Delta_j^2}{4\eta(1-\eta)} \bs{1}_{a_t(z_{t-1})=j}\right).
\end{align}
Now we can bound using the same reasoning as in \eqref{eq:fidelity_bound_classic} and \eqref{eq:fidelity_classic_key} and the display above
\begin{align}
\begin{split}
&\sum_{z_t} \sqrt{\Pd^j(z_t) \Pd^0(z_t)}\de^j(z_t)
=
\sum_{z_t} \sqrt{\Pd^j(z_t) \Pd^0(z_t)}\de^{j}(z_{t-1})\left(1-\tfrac{\Delta_j^2}{4\eta(1-\eta)} \bs{1}_{a_t(z_{t-1})=j}\right)
\\
&\geq 
\sum_{z_{t-1}}   \sqrt{\Pd^j(z_{t-1})\Pd^0(z_{t-1})}\de^{j}(z_{t-1})
\left(1-\tfrac{\Delta_j^2}{4\eta(1-\eta)} \bs{1}_{a_t(z_{t-1})=j}\right)
\\
&\quad  -\sum_{z_{t-1}} \sqrt{\Pd^j(z_{t-1})\Pd^0(z_{t-1})}\de^j(z_{t-1})\left(1-\tfrac{\Delta_j^2}{4\eta(1-\eta)} \bs{1}_{a_t(z_{t-1})=j}\right) \bs{1}_{a_t(z_{t-1})=j}
\tfrac{\Delta_j^2}{4\eta(1-\eta)}
\\ &
\geq \sum_{z_{t-1}}   \sqrt{\Pd^j(z_{t-1})\Pd^0(z_{t-1})}\de^{j}(z_{t-1})
 -\sum_{z_{t-1}}   \sqrt{\Pd^j(z_{t-1})\Pd^0(z_{t-1})}\de^{j}(z_{t-1})\tfrac{\Delta_j^2}{2\eta(1-\eta)} \bs{1}_{a_t(z_{t-1})=j}
\\
&\geq \sum_{z_{t-1}}   \sqrt{\Pd^j(z_{t-1})\Pd^0(z_{t-1})}\de^{j}(z_{t-1})
  -\tfrac{\Delta_j^2}{2\eta(1-\eta)} \sum_{z_{t-1}} \sqrt{\Pd^j(z_{t-1})\Pd^0(z_{t-1})}\de^j(z_{t-1}) \bs{1}_{a_t(z_{t-1})=j}.
\end{split}
\end{align}
In particular,  we find that the decay factor is chosen such that up to a constant the same bound before holds (see \eqref{eq:fidelity_classic_key}),
except that we introduced the terms $d^j(z_t)$ in the loss terms which makes it easier to control them.

Using \eqref{eq:fidelity_classic_difference} and a telescopic series we conclude that
\begin{align}
\begin{split}
2&\sqrt{\delta(1-\delta)}\geq \sqrt{F}(Z_T^j,Z_T^0)
\geq \sum_{z_T} \sqrt{\Pd^j(z_T) \Pd^0(z_T)}\de^j(z_T)
\\
&\geq 1 - \tfrac{\Delta_j^2}{2\eta(1-\eta)}\sum_{t=1}^T\sum_{z_{t-1}} \sqrt{\Pd^j(z_{t-1})\Pd^0(z_{t-1})}\de^j(z_{t-1}) \bs{1}_{a_t(z_{t-1})=j}.
\end{split}
\end{align}
Equivalently, this can be rewritten as
\begin{align}
\begin{split}\label{eq:bound_classical_fid}
(1-2\sqrt{\delta(1-\delta)})\frac{2\eta(1-\eta)}{\Delta_j^2}
\leq 
\sum_{t=1}^T\sum_{z_{t-1}} \sqrt{\Pd^j(z_{t-1})\Pd^0(z_{t-1})}\de^j(z_{t-1}) \bs{1}_{a_t(z_{t-1})=j}.
\end{split}
\end{align}
It remains to bound the right-hand side of this inequality.
For $t<T$ we write  $z_{T|t}=(x_1,\ldots,x_t)$ where
$z_T=(x_1,\ldots, x_t,\ldots, x_T)$, i.e., $z_{T|t}$ denotes the restriction of the history $z_T$ to the first $t$ steps. Then we have by definition 
of $\Pd$
\begin{align}\label{eq:formal_nonsense}
\sum_{z_T} \Pd^j(z_T)f(z_{T|t})=
\sum_{z_t} \Pd^j(z_t) f(z_t).
\end{align}
Moreover, we note that for all $z_T$ we can bound
\begin{align}\label{eq:geom_sum}
\begin{split}
\sum_{t=1}^T\de^j(z_{T|(t-1)})^2 \bs{1}\left[a_t(z_{T|(t-1)})=j\right]
&=
\sum_{n=0}^{n_j(z_T)-1} \left(1-\frac{\Delta_j^2}{4\eta(1-\eta)}\right)^{2n}
\\
&\leq \sum_{n=0}^\infty  \left(1-\frac{\Delta_j^2}{4\eta(1-\eta)}\right)^{n}
= \frac{4\eta(1-\eta)}{\Delta_j^2}.
\end{split}
\end{align}
We can bound using the Cauchy-Schwarz estimate, \eqref{eq:formal_nonsense},
and
\begin{align}
\begin{split}
\sum_{t=1}^T\sum_{z_{t-1}}& \sqrt{\Pd^j(z_{t-1})\Pd^0(z_{t-1})}\de^j(z_{t-1}) \bs{1}_{a_t(z_{t-1})=j}^2
\\
&\leq 
\left(\sum_{t=1}^T\sum_{z_{t-1}} \Pd^0(z_{t-1}) \bs{1}_{a_t(z_{t-1})=j}  \right)^{\frac12}
\left(\sum_{t=1}^T\sum_{z_{t-1}} \Pd^j(z_{t-1})\de^j(z_{t-1})^2 \bs{1}_{a_t(z_{t-1})=j}\right)^{\frac12}
\\
&\leq 
\left(\sum_{t=1}^T \Pd^0(A_t=j)  \right)^{\frac12}
\left(\sum_{z_{T}} \Pd^j(z_{T}) \sum_{t=1}^T\de^j(z_{T|(t-1)})^2 \bs{1}_{a_t(z_{T|(t-1)})=j}\right)^{\frac12}
\\
&\leq 
\left(\sum_{t=1}^T \Pd^0(A_t=j)  \right)^{\frac12}
\left(\sum_{z_{T}} \Pd^j(z_{T}) \frac{4\eta(1-\eta)}{\Delta_j^2}\right)^{\frac12}
\\
&\leq 
\left(\sum_{t=1}^T \Pd^0(A_t=j)  \right)^{\frac12}\frac{2\sqrt{\eta(1-\eta)}}{\Delta_j}.
\end{split}
\end{align}
Plugging this in \eqref{eq:bound_classical_fid}, dividing by 
$2\sqrt{\eta(1-\eta)}$ and summing over $j>1$ gives
\begin{align}
\begin{split}
(1-2\sqrt{\delta(1-\delta)})\sqrt{\eta(1-\eta)}\sum_{j=2}^N\Delta_j^{-2}
&\leq  \sum_{j=2}^N \left[\Delta_j^{-1}\left(\sum_{t=1}^T \Pd^0(A_t=j)  \right)^{\frac12}\right]
\\
&\leq \left(\sum_{j=2}^N \Delta_j^{-2}\right)^{\frac12}
\left(\sum_{j=2}^n\sum_{t=1}^T \Pd^0(A_t=j)  \right)^{\frac12}
\\
&=\left(\sum_{j=2}^n \Delta_j^{-2}\right)^{\frac12}
\left(\sum_{t=1}^T \sum_{j=1}^n \Pd^0(A_t=j)  \right)^{\frac12}
\\
&= \sqrt{H(\pp^1)}\sqrt{ T}.
\end{split}
\end{align}
Squaring this relation ends the proof.
\end{proof}

\section{Proof of Theorem~\ref{th:main}}\label{app:main}
In this section, we finally provide a proof of our main result.
\blue{
Let us for reference first summarize the setting that we introduced at the beginning of Section~\ref{sec:overview}.
We consider  vectors $\pp=(p_0,\ldots,p_N)\in [\eta, 1-\eta]^{N+1}$
with $p_i$ decreasing and then we let $\pp^i\in [\eta,1-\eta]^N$ be the reward vectors
with $\pp^i_j=\pp_j$ if  $i\neq j$ and $\pp^i_i=\pp_0$.
Moreover, $\pp^0_i=\pp_i$ for $1\leq i\leq N$ and note that $\pp^i$ and $\pp^0$ only differ in entry $i$ and recall that $\Delta_i=p_0-p_i$. 
We introduced the oracles $\mc{F}^p_i(\rho)=(1-p)\rho + pO_i\rho O_i^\dagger$ (where $O_i$ acts on $\ket{i}$ only) and then defined
\begin{align}
    \mc{E}_i = \mc{F}_1^{\pp^i_1}\circ\ldots \circ\mc{F}_N^{\pp^i_N}.
\end{align}
We then consider any algorithm start with an initial state $\rho_0$
and whose output is given by a POVM measurement of the state
\begin{align}
    \mc{T}_i\rho_0=(\mc{E}_i\otimes \id)\circ \mc{E}_{U_T}
\circ \ldots \circ (\mc{E}_i\otimes \id)\circ \mc{E}_{U_1}\rho_0
\end{align}
where $U_i$ are unitary maps. We denote the intermediate states of the algorithm by 
\begin{align}
    \tilde{\rho}^i_{t} = \mc{E}_{U_t}\rho^i_{t} 
    \quad 
    \rho^i_{t+1} = (\mc{E}_i\otimes \id) \tilde{\rho}^i_t,
\end{align}
i.e., $\rho^i_t$ denotes the state after $t$ oracle invocations and $
\tilde{\rho}^i_t$ the state before the $(t+1)$-th oracle call.

Let us now state a formal version of the main result, Theorem~\ref{th:main}.
}
\begin{theorem}\label{th:optimal}
Let $\delta<1/2$. 
Assume that $\pp^j$ are reward vectors as introduced  at the beginning of this section, where $p_i\in [\eta,1-\eta]$ for
some $\eta>0$.
Any quantum algorithm that identifies the best arm
when it is known that the reward vector is in 
$\{\pp^0,\ldots,\pp^n\}$ with probability at least $1-\delta$
requires at least 
\begin{align}
T\geq c H(\pp^1)=c\sum_{i=2}^N \Delta_i^{-2}
\end{align}
calls to the oracle $\mc{E}_i$ where $c=c(\delta,\eta)$ where an 
explicit expression under the condition $\Delta_N^2/\eta <1/2$ is given by
\begin{align}
c(\delta,\eta)=\left(\frac{\eta(1-2\sqrt{\delta(1-\delta))}}{20}\right)^2.
\end{align}
\end{theorem} 
\begin{remark}
We only do the proof under the condition that $\Delta_N^2/\eta\leq 1/2$ 
(the behaviour for small $\Delta_i$ is the main interest anyway). If this does not hold, some definitions in the proof of Proposition~\ref{prop:key} below need to be slightly adjusted (starting with \eqref{eq:defperturb1} and \eqref{eq:defperturb2}) but the final result will be the same except that the constant has a poorer dependence on $\eta$.
\end{remark}
\blue{

The key ingredient in the proof is a lower bound on the fidelity between the states obtained when applying the different oracles. Let us state this as a separate proposition

\begin{proposition}\label{prop:key}
    Let $\delta<1/2$. 
Assume that $\pp^j$ are reward vectors as introduced  at the beginning of this section, where $p_i\in [\eta,1-\eta]$ for
some $\eta>0$.
Then the following lower bound on the fidelity holds for $1\leq i\leq N$
\begin{align}\label{eq:fid_bound_final}
    \fid{\rho_T^i,\rho_T^0}\geq 1-\frac{20}{\eta}\Delta_i 
    \left(\sum_{t=1}^T \tr P_i \tilde{\rho}^0_t\right)^{\frac12}
\end{align}
for $\Delta_i^2/\eta<1/2$ and otherwise the bound holds for some constant $c(\eta)$ instead of $20/\eta$.
\end{proposition}
Once this proposition is proved, the proof of the main result is straightforward.
\begin{proof}[Proof of Theorem~\ref{th:optimal}]
As in the proof of Theorem~\ref{th:faulty_grover}
we have that success of the algorithm implies that 
\begin{align}
   \fid{\rho_T^i,\rho_T^0}\geq 2\sqrt{\delta(1-\delta)}
\end{align}
holds for $i\geq 2$ ( the first arm has the highest mean reward  for $\pp^1$ and $\pp^0$ and thus no bound can be derived for $i=1$).
Combining this with \eqref{eq:fid_bound_final} we obtain
\begin{align}
\begin{split}
\frac{\eta(1-2\sqrt{\delta(1-\delta))}}{20}
\sum_{i\geq 2}
\Delta_i^{-2} &\leq \sum_{i\geq 2}
\Delta_i^{-1} \left(\sum_{t=1}^T \tr P_i \rho_t^0\right)^{\frac12}
\\
&\leq 
\left(\sum_{i\geq 2}
\Delta_i^{-2} \right)^{\frac12}
\left(
\sum_{i\geq 2}
\sum_{t=1}^T \tr P_i \rho_t^0\right)^{\frac12}
\leq \left(\sum_{i\geq 2}
\Delta_i^{-2} \right)^{\frac12}\sqrt{T}.
\end{split}
\end{align}
This  ends the proof.
\end{proof}
It remains to prove Proposition~\ref{prop:key}.
}
\begin{proof}[Proof of Proposition~\ref{prop:key}]
\blue{
The proof is a bit technical and lengthy, and we therefore split it into 
several steps and provide intermediate results. 
\paragraph{Overview and notation.}
Let us first introduce some additional notation and give a high-level overview of the proof.
Note that we reviewed the general setting at the beginning of the section.

Our general  strategy is to first decompose the corresponding density matrices $\rho^i$ and $\rho^0$ into a sum of density matrices and then apply strong concavity to lower bound $\fid{\rho^i_T,\rho^0_T}$. To control the loss of fidelity that we incur during a single oracle call we will rely on Corollary~\ref{co:fidelity}. Aggregating the loss terms this  gives rise to a lower bound on the fidelity which, however, involves several error terms that are not straightforward to control. Bounding those error terms will rely on strategies that are similar to the one used in  the proof of Theorem~\ref{th:classical}.

Let us now sketch how we decompose the density matrices before we give the actual definitions.
To do this we introduce some notation. To denote the rewards in step $t$ we consider $x_t\in \{0,1\}^N$ and
  we collect those rewards in the vector $z_t=(x_1,\ldots, x_t)$.
  As before we consider the measure $\Pd^i$ on sequences $z_T$ which has the property that $\Pd^i((x_t)_j=1)= \pp^i_j$ and
  $\Pd^i((x_t)_j=0)=1-\pp_j^i$ and those variables are independent.

}

   
We generally split all states in two after each invocation of an oracle
$\mc{F}_j^{p_j}$ depending on the realization of the randomness.
For $j\neq i$ this will be the natural separation in reward 0 and 1 respectively, but for $j=i$ more complex decompositions need to be used to obtain optimal bounds. Indeed, we have seen in Lemma~\ref{le:coupling_const} that
the optimal loss in fidelity when applying $\mc{E}_i$ and $\mc{E}_0$ is of the order
$\Delta_i^2$ and we essentially
use the decomposition constructed there.

For $\rho_t^i$ we consider a decomposition given by 
\begin{align}\label{eq:decomp_i}
\rho_t^i = \sum_{z_t} \Pd^i(z_t) \rho(z_t).
\end{align}
For $\rho_t^0$ we consider a decomposition into pure states depending on
$i$, i.e., we construct a distribution $\Qd^i$ on sequences $z_T$ and states $\psi(z_t)$ such that
\begin{align}\label{eq:decomp_0}
\rho_t^0 = \sum_{z_t}\Qd^i(z_t) \ket{\psi(z_t)}\bra{\psi(z_t)}.
\end{align}
We emphasize again that while $\rho_t^0$ does not depend on $i$ the decomposition does depend on $i$ because it is used to bound the distance to $\rho^i_t$.
To simplify the notation, we drop the $i$ dependence of the decomposition into $\rho(z_t)$ 
and $\psi(z_t)$.
Note that the decomposition for $\rho^0$ involves the complexity, while
the decomposition of $\rho^i$ will be relatively straightforward.
Let us now define those decompositions formally.

\paragraph{Decomposition of $\rho^i_t$.}
We decompose $\rho^i_t$ roughly as 
\begin{align}
\rho_t(z_t) \approx\ket{\p(z_t)}\bra{\p(z_t)},
\qquad 
\p(z_t=(x_1,\ldots, x_t))= (O_{x_t}\otimes \id)U_t
\ldots (O_{x_1}\otimes \id)U_1,
\end{align}
i.e., we just decompose it according to the realizations of the rewards.
However, we in addition need to ensure that the density matrices
$\rho_t(z_t)$ decohere with respect to
$\psi(z_t)$.
For the definition we introduce the notation $\hat{x}_t\in \{0,1\}^{N}$ for
the vector $x_t$ with the $i$-th entry set to 0, i.e.,
$(\hat{x}_t)_j=(x_t)_j$ for $j\neq i$ and $(x_t)_i=0$.
With this notation we get the following decomposition.
\blue{
\begin{lemma}\label{le:decomp_rho_i}
Assume $\Delta_i^2/\eta <1/2$.
We define (recall that $\rho_0$ denotes the initial state of the algorithm)
\begin{align}
\rho(z_0)&=\rho_0,
\\
\tilde{\rho}(z_t) &= \mc{E}_{U_{t+1}}(\rho(z_t)),
\\
 \label{eq:defperturb1}
\tilde{\rho}_0(z_t) &=  \left(1 - \frac{p_0 \Delta_i^2}{\eta}\right) \; \tilde{\rho}(z_t)
+ \frac{p_0\Delta_i^2}{\eta}\; O_i\tilde{\rho}(z_t)O_i^\dagger,
\\ 
\label{eq:defperturb2}
\tilde{\rho}_1(z_t) &=\left(1 - \frac{(1-p_0) \Delta_i^2}{\eta}\right)\; O_i\tilde{\rho}(z_t)O_i^\dagger
+  \frac{(1-p_0)}{\eta}\Delta_i^2 \; \tilde{\rho}(z_t),
\\
\rho(z_{t+1}) & = \mc{E}_{O_{\hat{x}_{t+1}}}(\tilde{\rho}_{(x_{t+1})_i}(z_{t})).
\end{align}
Then \eqref{eq:decomp_i} holds, i.e.,
\begin{align}
\tilde{\rho}_t^i =
\sum_{z_t} \Pd(z_t) \tilde{\rho}(z_t).
\end{align}
\end{lemma}
We emphasize that here the notation $\rho_0$ and $\rho_1$ corresponds (roughly) to reward 0 or 1 on arm $i$.
Note that the  reason to slightly perturb $\tilde{\rho}_0(z_t)$ and
$\tilde{\rho}_1$ from their natural definitions $\tilde{\rho}(z_t)$
and $O_i\tilde{\rho}(z_t)O_i^\dagger$ is that our definition ensures up to a constant the same loss in fidelity of order $\Delta_i^2$ but in addition we induce decoherence of $\tilde{\rho}$ so that we can argue as in the proof of 
Theorem~\ref{th:faulty_grover}. 
The factor $\eta$ leads to slightly simpler expressions and tighter bounds, but is not strictly necessary. Indeed, if $\Delta_i^2/\eta>1$ the definition needs to be adapted by removing the $\eta$ which results in a weaker $\eta$-dependence.
}
\begin{proof}
We argue by induction. We have 
\begin{align}
\tilde{\rho}_t^i=\mc{E}_{U_{t+1}}(\rho_t^i)
=\mc{E}_{U_{t+1}}\left(\sum_{z_t} \Pd^i(z_t)\rho(z_t)\right)
= \sum_{z_t} \Pd^i(z_t)\mc{E}_{U_{t+1}}(\rho(z_t))
= \sum_{z_t} \Pd^i(z_t)\tilde{\rho}(z_t).
\end{align}
Next we note that
\begin{align}
\begin{split}
p_0 \tilde{\rho}_1(z_t)
+ (1-p_0) \tilde{\rho}_0(z_t)
&=  p_0\; O_i\tilde{\rho}(z_t)O_i^\dagger
+  (1-p_0) \;\tilde{\rho}(z_t)
= \mc{F}_i^{p_0}(\tilde{\rho}(z_t)).
\end{split}
\end{align}
This implies together with the definition of $\Pd^i$ that
\begin{align}
\begin{split}
\sum_{x_{t+1}}
\Pd^i(x_{t+1}) \rho((z_t,x_{t+1}))
&=
\sum_{x_{t+1}}
\Pd^i(x_{t+1}) \mc{E}_{O_{\hat{x}_{t+1}}}(\tilde{\rho}_{(x_{t+1})_i}(z_t))
=
\sum_{x_{t+1}}
\Pd^i(x_{t+1})
\mc{E}_{O_{\hat{x}_{t+1}}}\circ\mc{F}^{p_0}_i(\tilde{\rho}(z_t))
\\
&=
\sum_{x_{t+1}}
\Pd^i(x_{t+1})
\mc{E}_{O_{{x}_{t+1}}}(\tilde{\rho}(z_t))
= \mc{E}_i (\tilde{\rho}(z_t)).
\end{split}
\end{align}
Here we used that $\Pd^i((x_{t+1})_i=1)=\pp^i_i= p_0$.
Using the induction hypothesis, we conclude that 
\begin{align}
\sum_{z_{t+1}}
\Pd^i(z_{t+1}) \rho(z_{t+1})
=
\sum_{z_{t}}\Pd^i(z_{t}) \sum_{x_{t+1}}
\Pd(x_{t+1})
\rho((z_{t},x_{t+1}))
=\sum_{z_{t}}\Pd^i(z_{t}) \mc{E}_i (\tilde{\rho}(z_t))
=\mc{E}_i(\tilde{\rho}_t^i)=\rho_{t+1}^i.
\end{align}
\end{proof}

\paragraph{Decomposition of $\rho^0_t$.}
To define the decomposition of $\rho^0_t$ we first define several quantum states.
Let
\begin{align}
\tilde{\psi}(z_t)
= U_t \big(\psi(z_{t})\big).
\end{align}
Define (essentially as in Lemma~\ref{le:coupling_const})
\begin{align}
\label{eq:barpsi1}
\bar\psi_1(z_t) &= \sqrt{1-p_i}\cos(\alpha)\tilde{\psi}(z_t)
+\sqrt{p_i}\sin(\alpha)O_i\tilde{\psi}(z_t)
\\
\tilde{\psi}_1(z_t)&=\bar\psi_1(z_t) /\lVert \bar\psi_1(z_t)\rVert
\\
\label{eq:barpsi0}
\bar\psi_0(z_t)&= -\sqrt{1-p_i}\sin(\alpha)\tilde{\psi}(z_t)
+\sqrt{p_i}\cos(\alpha)O_i\tilde{\psi}(z_t)
\\
\tilde{\psi}_0(z_t)&=\bar\psi_0(z_t) /\lVert \bar\psi_0(z_t)\rVert
\end{align}
where $\alpha\in [0,\pi)$ is defined through
$\cos(\alpha)= \sqrt{p_ip_0}+\sqrt{(1-p_i)(1-p_0)}$ (i.e., as in\eqref{eq:def_angle} with $p$ and $q$ replaced by $p_i$ and $p_0$).
Finally, let 
\begin{align}
\psi(z_{t+1}) =O_{\hat{x}_{t+1}} \tilde{\psi}_{(x_{t+1})_i} (z_t).
\end{align}
Next we consider 
 a probability distribution $\Qd^i$ on sequences $z_t=(x_1, \ldots, x_t)$  as before. It will factorize according to 
\begin{align}
\Qd^i(z_{t+1})
=\Qd^i(z_{t})\Qd^i(x_{t+1}|z_{t})
= \Qd^i(z_{t}) \Qd^i((x_{t+1})_i|z_{t}) \prod_{j\neq i} p_j^{(x_{t+1})_j}
(1-p_j)^{1 - (x_{t+1})_j}.
\end{align} 
In other words 
$\Qd^i$ is the unique distribution on $z_t$ such that the variables $(x_t)_j$ for $i\neq j$ are independent
of everything
and distributed according to $\mathrm{Ber}(\pp^0_j)=\mathrm{Ber}(p_j)$
and, moreover, we require
\begin{align}\label{eq:def_of_Q}
\begin{split}
\Qd^i\Big((x_{t+1})_i=1\;|\; z_t\Big)&=\lVert \bar\psi_1(z_t) \rVert^2
\\
\Qd^i\Big((x_{t+1})_i=0\;|\;z_t\Big)&=\lVert \bar\psi_0(z_t) \rVert^2
= 1- \lVert \bar\psi_1(z_t) \rVert^2.
\end{split}
\end{align}
Clearly this defines uniquely a probability distribution. 
\blue{
With this definition we can state the following lemma.
\begin{lemma}\label{le:decomp_rho_0}
The following identity holds
\begin{align}\label{eq:trace0}
\rho_t^0 &=\sum_{z_t} \Qd^i(z_t) \ket{\psi(z_t)}\bra{\psi(z_t)},
\\ \label{eq:trac0tilde}
\tilde\rho_t^0 &=\sum_{z_t} \Qd^i(z_t) \ket{\tilde{\psi}(z_t)}\bra{\tilde\psi(z_t)}.
\end{align}
\end{lemma}
}
\begin{proof}
To show this, we argue by induction. The first step is simple
\begin{align}
\begin{split}
\sum_{z_t} \Qd^i(z_t) \ket{\tilde{\psi}(z_t)}\bra{\tilde\psi(z_t)}
&= 
\sum_{z_t} \Qd^i(z_t) \ket{U_{t+1}{\psi}(z_t)}\bra{U\psi(z_t)}
\\
&=
U_{t+1}\left(\sum_{z_t} \Qd^i(z_t) \ket{{\psi}(z_t)}\bra{\psi(z_t)}\right)U^\dagger_{t+1}
=
\mc{E}_{U_{t+1}}( \rho_t^0)=\tilde{\rho}^0_t.
\end{split}
\end{align}
By \eqref{eq:decomp} in the proof of  Lemma~\ref{le:coupling_const}
the following relation holds
\begin{align}
\mc{F}_i^{p_i}(\tilde{\psi}(z_t))
= \Qd^i\Big((x_{t+1})_i=1\;|\; z_t\Big) \,\ket{\tilde{\psi}_1(z_t)}
\bra{\tilde{\psi}_1(z_t)}+
\Qd^i\Big((x_{t+1})_0=1\;|\; z_t\Big) \,\ket{\tilde{\psi}_0(z_t)}
\bra{\tilde{\psi}_0(z_t)}.
\end{align}
Using this relation we conclude that
\begin{align}
\begin{split}
&\sum_{z_{t+1}}
\Qd^i(z_{t+1}) \ket{{\psi}_0(z_{t+1})}
\bra{{\psi}_0(z_{t+1})}
=
\sum_{z_{t+1}}
\Qd^i(z_{t+1}) \ket{O_{\hat{x}_{t+1}} \tilde{\psi}_{(x_{t+1})_i} (z_t)}
\bra{O_{\hat{x}_{t+1}} \tilde{\psi}_{(x_{t+1})_i} (z_t)}
\\
&=
\mc{F}_1^{p_1}
\circ\ldots
\circ 
\mc{F}_{i-1}^{p_{i-1}}
\circ
\mc{F}_{i+1}^{p_{i+1}}
\circ\ldots\circ
\mc{F}_N^{p_N}
\left(
\sum_{z_t}\Qd^i(z_t)\sum_{s=0}^1
\Qd^i((x_{t+1})_i=s|z_t) \ket{ \tilde{\psi}_{s} (z_t)}
\bra{\tilde{\psi}_{s} (z_t)}\right)
\\
&=
\mc{F}_1^{p_1}
\circ\ldots
\circ 
\mc{F}_{i-1}^{p_{i-1}}
\circ
\mc{F}_{i+1}^{p_{i+1}}
\circ\ldots\circ
\mc{F}_N^{p_N}\circ \mc{F}_i^{p_i}\left(\sum_{z_t} \Qd^i(z_t)
\ket{ \tilde{\psi} (z_t)}
\bra{\tilde{\psi} (z_t)}\right)
\\
&= \mc{E}_0(\tilde{\rho}_t^0)=\rho_{t+1}^0.
\end{split}
\end{align}
\end{proof}

\paragraph{Bounding the fidelity loss in a single step.}
We now start to estimate the fidelity of $\rho_t^k$ and $\rho_t^0$.
The first step is to bound the fidelity when making a single oracle call
on a single term in the decomposition. The loss only occurs when passing
from $\tilde{\rho}(z_t)$
and $\tilde{\psi}(z_t)$ to $\tilde{\rho}_{0/1}(z_t)$
and $\tilde{\psi}_{0/1}(z_t)$.
\blue{
We can show the following bound.

\begin{lemma}\label{le:fidelity_split2}
    Let $z_t\in \{0,1\}^{Nt}$ and set
    \begin{align}
        S=\sqrt{F}(\tilde{\rho}(z_t),\tilde{\psi}(z_t))
    \end{align}
    and assume that $S\geq \tfrac12$.
    Then the following bound holds  with $R(z_t)=\tr(\tilde{\rho}(z_t)^2)$
    \begin{align}
\begin{split}\label{eq:fidloss_one_step_lemma}
&\sqrt{   \Pd^i((x_{t+1})_i=0)
\Qd^i((x_{t+1})_i=0|z_t)}
\sqrt{F}(\tilde\rho_0(z_t),\tilde{\psi}_0(z_t))
\\
&\qquad +
\sqrt{ \Pd^i((x_{t+1})_i=1)
\Qd^i((x_{t+1})_i=1|z_t)}
\sqrt{F}(\tilde\rho_1(z_t),\tilde{\psi}_1(z_t))
\\
&\qquad\qquad\geq S - \frac{5\Delta_i}{\sqrt{2}\eta}
\lVert P_i \tilde{\psi}(z_t)\rVert \; \sqrt{R(z_{t}) -\max_{x_{t+1}} R((z_{t},x_{t+1})}
- \frac{4\Delta_i^2}{\eta^2}\lVert P_i \tilde{\psi}(z_t)\rVert^2.
\end{split}
\end{align}
\end{lemma}
}
\begin{proof}
We will bound this loss using Corollary~\ref{co:fidelity} above.
The additional flexibility of the $\sigma$ term in this corollary allows us to apply this to our setting where $\tilde{\rho}_{0/1}$
are defined as in \eqref{eq:defperturb1} and \eqref{eq:defperturb2}. Specifically, 
we apply this Corollary with $\rho=\tilde\rho(z_t)$, $\psi=\tilde\psi(z_t)$, 
$p=p_0$, $q=p_i$, $U=O_i$ and $\sigma=\Delta_i^2(O_i\tilde{\rho}(z_t)O_i^\dagger
- \tilde{\rho}(z_t))$.
Then we note that the definition
of $\tilde{\psi}_{0/1}$ agrees with the definition of $\psi_{0/1}$
in Lemma~\ref{le:fidelity} and $\tilde{\rho}_{0/1}(z_t)$ agrees with
$\rho_{0/1}$ and $q'=\Qd^i((x_{t+1})_i=1|z_t)$, $p=\Pd^i(x_i=1)$. 
We conclude from Corollary~\ref{co:fidelity} that
\begin{align}
\begin{split}\label{eq:init}
&\sqrt{ (1-p_0)   
\Qd^i((x_{t+1})_i=0|z_t)}
\sqrt{F}(\tilde\rho_0(z_t),\tilde{\psi}_0(z_t))
+
\sqrt{ p_0 
\Qd^i((x_{t+1})_i=1|z_t)}
\sqrt{F}(\tilde\rho_1(z_t),\tilde{\psi}_1(z_t))
\\
&\geq S - \Delta_i^2
\bigg(\frac{|(\tilde{\psi}(z_t), (O_i\tilde{\rho}(z_t)O_i^\dagger - \tilde{\rho}(z_t))\tilde{\psi}(z_t))|}{\eta}
+\frac{\Big|\mathrm{Re}\Big((\tilde{\psi}(z_t), (O_i\tilde{\rho}(z_t) - \tilde{\rho}(z_t))\tilde{\psi}(z_t))\Big)\Big|^2}{\eta^2}
\\
&\qquad
+ \frac{2(1-p)}{\eta}\left|(\bar{\psi}_1(z_t), \left(O_i\tilde{\rho}(z_t)O_i^\dagger - \tilde{\rho}(z_t)\right)\bar{\psi}_1(z_t))\right|
+ \frac{2p}{\eta}\left|(\bar{\psi}_0(z_t), \left(O_i\tilde{\rho}(z_t)O_i^\dagger - \tilde{\rho}(z_t)\right)\bar{\psi}_0(z_t))\right|
\bigg).
\end{split}
\end{align}
We control the right-hand side of this expression by exploiting the specific structure of the oracle $O_i$. As in the proof of Theorem~\ref{th:faulty_grover} we
use that $P_i = \ket{i}\bra{i}\otimes \id$ satisfies
$(1-P_i)O_i=(1-P_i)$ and apply 
 Lemma~\ref{le:projection}. We get 
 \begin{align}\label{eq:cont1}
 |(\tilde{\psi}(z_t), (O_i\tilde{\rho}(z_t)O_i^\dagger - \tilde{\rho}(z_t))\tilde{\psi}(z_t))|\leq 2\lVert P_i\tilde{\psi}(z_t)\rVert
 \left(\tr (O_i\tilde{\rho}(z_t)O_i^\dagger - \tilde{\rho}(z_t))^2\right)^{\frac12}.
 \end{align}
For the second term we use that
$O_i-\id=(\id-P_i + P_i)(O_i-\id)
=P_i(O_i-\id)$ which implies after an application of Cauchy-Schwarz
\begin{align}\label{eq:cont2}
\left|\mathrm{Re}\Big((\tilde{\psi}(z_t), (O_i\tilde{\rho}(z_t) - \tilde{\rho}(z_t))\tilde{\psi}(z_t))\Big)\right|^2
\leq \lVert P_i \tilde{\psi}(z_t)\rVert^2
\cdot \lVert (O_i-\id)\tilde{\rho}(z_t)\tilde{\psi}(z_t)\rVert^2
\leq 4 \lVert P_i \tilde{\psi}(z_t)\rVert^2.
\end{align}
For the third term we use again Lemma~\ref{le:projection}, the definition 
\eqref{eq:barpsi1}, and $[O_i, P_i]=0$  and bound 
\begin{align}
\begin{split}\label{eq:cont3}
&\left|(\bar{\psi}_1(z_t), \left(O_i\tilde{\rho}(z_t)O_i^\dagger - \tilde{\rho}(z_t)\right)\bar{\psi}_1(z_t))\right|
\leq 
\lVert P_i \bar{\psi}_1(z_t)\rVert \cdot
\lVert \bar{\psi}_1(z_t)\rVert  \cdot
\left(\tr (O_i\tilde{\rho}(z_t)O_i^\dagger - \tilde{\rho}(z_t))^2\right)^{\frac12}
\\
&\leq 
\left(\sqrt{1-p_i}\cos(\alpha) \lVert P_i \tilde{\psi}(z_t)\rVert
+\sqrt{p_i} \sin(\alpha) \lVert P_iO_i \tilde{\psi}(z_t)\rVert
\right)
\lVert \bar{\psi}_1(z_t)\rVert  \cdot
\left(\tr (O_i\tilde{\rho}(z_t)O_i^\dagger - \tilde{\rho}(z_t))^2\right)^{\frac12}
\\
&\leq 
2\lVert P_i \tilde{\psi}(z_t)\rVert \; \left(\tr (O_i\tilde{\rho}(z_t)O_i^\dagger - \tilde{\rho}(z_t))^2\right)^{\frac12}.
\end{split}
\end{align}
The same reasoning implies
\begin{align}\label{eq:cont4}
\left|(\bar{\psi}_0(z_t), \left(O_i\tilde{\rho}(z_t)O_i^\dagger - \tilde{\rho}(z_t)\right)\bar{\psi}_0(z_t))\right|
\leq 
2\lVert P_i \tilde{\psi}(z_t)\rVert \; \left(\tr (O_i\tilde{\rho}(z_t)O_i^\dagger - \tilde{\rho}(z_t))^2\right)^{\frac12}.
\end{align}
Plugging \eqref{eq:cont1}, \eqref{eq:cont2}, \eqref{eq:cont3}, and \eqref{eq:cont4} in \eqref{eq:init} we obtain
\begin{align}
\begin{split}\label{eq:fidloss_one_step}
&\sqrt{ (1-p_0)   
\Qd^i((x_{t+1})_i=0|z_t)}
\sqrt{F}(\tilde\rho_0(z_t),\tilde{\psi}_0(z_t))
+
\sqrt{ p_0 
\Qd^i((x_{t+1})_i=1|z_t)}
\sqrt{F}(\tilde\rho_1(z_t),\tilde{\psi}_1(z_t))
\\
&\qquad\geq S - \frac{5\Delta_i^2}{\eta}
\lVert P_i \tilde{\psi}(z_t)\rVert \; \left(\tr (O_i\tilde{\rho}(z_t)O_i^\dagger - \tilde{\rho}(z_t))^2\right)^{\frac12}
- \frac{4\Delta_i^2}{\eta^2}\lVert P_i \tilde{\psi}(z_t)\rVert^2.
\end{split}
\end{align}
It remains to bound the trace term,
which will be very similar to the  proof of Theorem~\ref{th:faulty_grover}. We define (again not indicating the $i$ dependence)
\begin{align}
R(z_t)=\tr(\tilde{\rho}(z_t)^2).
\end{align}
Invariance of the purity under unitary operations implies that
\begin{align}
R(z_{t+1}) = \tr(\tilde{\rho}_{(x_{t+1})_i}(z_t)^2).
\end{align}
Calculations as in \eqref{eq:change_purity} give us for $(x_{t+1})_i=0$
\begin{align}
R(z_{t}) - R(z_{t+1})
= \left(1-\frac{p_0\Delta_i^2}{\eta}\right)\frac{p_0\Delta_i^2}{\eta} \cdot \tr (\tilde\rho(z_t)
- O_i\tilde\rho(z_t)O_i^\dagger)^2.
\end{align}
A similar identity for $(x_t)_i=1$ 
together with the assumption $\Delta_i^2/\eta \leq 1/2$ and $\min(p_0,1-p_0)\geq \eta$ 
imply
\begin{align}\label{eq:trace_purity-1}
\tr (\tilde\rho(z_t)
- O_i\tilde\rho(z_t)O_i^\dagger)^2
\leq \frac{1}{2 \Delta_i^2}(R(z_{t}) - R(z_{t+1})).
\end{align}
Note that the left-hand side only depends on $z_t$ so we conclude
\begin{align}\label{eq:trace_purity}
\tr (\tilde\rho(z_t)
- O_i\tilde\rho(z_t))O_i^\dagger)^2
\leq \frac{1}{2 \Delta_i^2}(R(z_{t}) -\max_{x_{t+1}} R((z_{t},x_{t+1})).
\end{align}
Plugging this bound in \eqref{eq:fidloss_one_step} ends the proof of the lemma.
\end{proof}
The previous Lemma is   the key relation that allows us to control the fidelity loss between
$\rho_t^0$ and $\rho_t^i$.
\blue{
We will use the following corollary.
\begin{corollary}\label{co:fid_loss_one_step}
     Let $z_t\in \{0,1\}^{Nt}$ and let as before
    \begin{align}
        S=\sqrt{F}(\tilde{\rho}(z_t),\tilde{\psi}(z_t))
    \end{align}
    and assume that $S\geq \tfrac12$.
    Then the following bound holds  (recall $R(z_t)=\tr(\tilde{\rho}(z_t)^2)$)
\begin{align}
\begin{split}\label{eq:fid_loss_summed_co}
&\sum_{x_{t+1}} \sqrt{ \Pd^i(x_{t+1})\Qd^i(x_{t+1}|z_t)}
\sqrt{F}\Big(\tilde\rho((z_t,x_{t+1})), \tilde{\psi}((z_t,x_{t+1}))\Big)
\\
&\geq 
\sqrt{F}(\tilde{\rho}(z_t),\tilde{\psi}(z_t))- \frac{5\Delta_i}{\sqrt{2}\eta}
\lVert P_i \tilde{\psi}(z_t)\rVert \; \sqrt{R(z_{t}) -\max_{x_{t+1}} R((z_{t},x_{t+1})}
- \frac{4\Delta_i^2}{\eta^2}\lVert P_i \tilde{\psi}(z_t)\rVert^2.
\end{split}
\end{align}
\end{corollary}

}
\begin{proof}
We first remark that invariance of the fidelity under unitary maps implies
\begin{align}\label{eq:fid_inv1}
\sqrt{F}\big(\rho((z_t,x_{t+1})), \ket{\psi((z_t,x_{t+1}))}\big)
&=
\sqrt{F}\big(\tilde{\rho}_{(x_{t+1})_i}(z_t)
, \ket{\psi_{(x_{t+1})_i}(z_t)}\big),
\\
\label{eq:fid_inv2}
\sqrt{F}\big(\rho(z_{t+1}), \ket{\psi(z_{t+1})}\big)
&=
\sqrt{F}\big(\tilde{\rho}(z_{t+1}), \ket{\tilde{\psi}(z_{t+1})}\big).
\end{align}

We now sum the bound in Lemma~\ref{le:fidelity_split2} over all possible values of $x_{t+1}$ to 
move from step  $(t+1)$ back to step $t$. 
We denote $\hat{x}_{t+1}$ the vector $x_{t+1}$ with entry $i$ removed.
Note that under $\Pd^i$ and $\Qd^i$ this vector is independent of $z_t, (x_{t+1})_i$ and $\Pd^i(\hat{x}_{t+1})=\Qd^i(\hat{x}_{t+1})$.
We also introduce the notation $\hat{x}_{t+1}^c$ for the vector that has entry $i$ equal to $c$.
Then we get using \eqref{eq:fid_inv1} and \eqref{eq:fid_inv2}
 for any $z_t$ and $c\in \{0,1\}$
\begin{align}
\begin{split}
\sum_{\hat{x}_{t+1}}
\sqrt{\Pd^i(\hat{x}_{t+1})\Qd^i(\hat{x}_{t+1})}
\sqrt{F}\Big(\tilde\rho((z_t,\hat{x}_{t+1}^c)), \tilde{\psi}((z_t,\hat{x}_{t+1}^c))\Big)
&=
\sum_{\hat{x}_{t+1}}
\sqrt{\Pd^i(\hat{x}_{t+1})\Pd^i(\hat{x}_{t+1})}
\sqrt{F}\Big(\tilde{\rho}_c(z_t), \tilde{\psi}_c(z_t)\Big)
\\
&=
\sqrt{F}\Big(\tilde{\rho}_c(z_t), \tilde{\psi}_c(z_t)\Big).
\end{split}
\end{align}
\blue{
Summing this over $c=0,1$ and using \eqref{eq:fidloss_one_step}
we get for all $z_t$ such that $\sqrt{F}(\tilde{\rho}(z_t),\tilde{\psi(z_t)})\geq 1/2$ the bound
\begin{align}
\begin{split}\label{eq:fid_loss_summed}
&\sum_{x_{t+1}} \sqrt{ \Pd^i(x_{t+1})\Qd^i(x_{t+1}|z_t)}
\sqrt{F}\Big(\tilde\rho((z_t,x_{t+1})), \tilde{\psi}((z_t,x_{t+1}))\Big)
\\
&
=\sum_{c=0}^1 \sqrt{ \Pd^i((x_{t+1})_i=c)
\Qd^i((x_{t+1})_i=c|z_t)}
\sum_{\hat{x}_{t+1}}
\sqrt{\Pd^i(\hat{x}_{t+1})\Qd^i(\hat{x}_{t+1})}
\sqrt{F}\Big(\tilde\rho((z_t,\hat{x}_{t+1}^c)), \tilde{\psi}((z_t,\hat{x}_{t+1}^c))\Big)
\\
&
=\sum_{c=0}^1  \sqrt{ \Pd^i((x_{t+1})_i=c)
\Qd^i((x_{t+1})_i=c|z_t)} \sqrt{F}\Big(\tilde{\rho}_c(z_t), \tilde{\psi}_c(z_t)\Big)
\\
&\geq 
\sqrt{F}(\tilde{\rho}(z_t),\tilde{\psi}(z_t))- \frac{5\Delta_i}{\sqrt{2}\eta}
\lVert P_i \tilde{\psi}(z_t)\rVert \; \sqrt{R(z_{t}) -\max_{x_{t+1}} R((z_{t},x_{t+1})}
- \frac{4\Delta_i^2}{\eta^2}\lVert P_i \tilde{\psi}(z_t)\rVert^2.
\end{split}
\end{align}
}
\end{proof}

\paragraph{Deriving an error decomposition for the fidelity.}
Equipped with this corollary, we now move on to control the total loss in fidelity.

The general strategy is now to use joint concavity of the fidelity 
and then inductive application of the estimate \eqref{eq:fidloss_one_step} above to lower bound the fidelity.
There are two technical difficulties: 
When directly applying the corollary above, we do not get the optimal bound
because it is difficult to control the difference between $\P^i$ and $\Qd^i$. This is the same problem as in the proof of Theorem~\ref{th:classical}, and we can address this with a similar idea.
The second difficulty is that the change in fidelity in Lemma~\ref{le:coupling_const} involves the inverse
of the initial fidelity, and so we derived \eqref{eq:fid_loss_summed}
only for fidelity $S\geq 1/2$.  
The high-level argument that this is sufficient is that if we know that the weighted mean of the fidelities is large, then there cannot be too many small terms in the mixture, so the condition will mostly hold.

Technically, we address both 
difficulties by introducing additional sequences $\de(z_t)$, $s(z_t)$, and $h(z_t)$ that we smuggle into the sum.
We define $\de(z_0)=1$ and then recursively for $z_t=(z_{t-1},x_t)$ 
\begin{align}\label{eq:rel_r}
\de(z_t)=\de(z_{t-1})\left(1-\frac{\Delta_i^2\lVert P_i\tilde\psi(z_{t-1})\rVert^2}{\eta^2}\right)
= \de(z_{t-1})- \de(z_{t-1})\frac{\Delta_i^2\lVert P_i\tilde\psi(z_{t-1})\rVert^2}{\eta^2}.
\end{align}
This is the quantum analogue of the classical definition 
in \eqref{eq:recursion_classical_decay} and has the same motivation.
Note that the $\eta$ factor was introduced for convenience, it could be dropped at the price of a slightly worse $\eta$ dependence.

  Next, we define 
  \begin{align}
      s(z_t)=1 \quad \text{iff for all $t'\leq t$ the bound
      $\sqrt{F}(\tilde\rho(z_{t'}), \tilde\psi(z_{t'}))\geq 1/2$ holds.}
  \end{align}
  In other words $s(z_t)=0$ for $z_t=(x_1,\ldots,x_t)$ if there is 
  $t'\leq t$ such that $z_{t'}=(x_1,\ldots, x_{t'})$
  satisfies $\sqrt{F}(\tilde\rho(z_t), \tilde\psi(z_t))<1/2$.
  Moreover, we define $h(z_t)=1$ if $s(z_t)=0$ but
  $s(z_{t'})=1$ for all $z_{t'}$ as above. Put differently for $z_t=(z_{t-1},x_t)$ the relation 
  \begin{align}
  h(z_t)=s(z_{t-1})-s(z_t)
  \end{align}
  holds, 
  i.e., $h(z_t)$ keeps track when the fidelity  $\sqrt{F}(\rho(z_t), \psi(z_t))$ falls below $1/2$ for the first time.

We can now bound (loosely speaking) the change in fidelity 
$  \sqrt{F}(\rho^k_{T-1},\rho_{T-1}^0) - \sqrt{F}(\rho^k_T,\rho_T^0)$
(actually, we only bound the difference on the lower bounds).
This will be achieved by combining the estimates above and the introduction of
various error terms.
We first note that
Lemma~\ref{le:decomp_rho_0} and Lemma~\ref{le:decomp_rho_i}
together with strong concavity of the fidelity (see \eqref{eq:fidelity_strong_concave}) and the upper bounds 
 $s\leq $ and $\de\leq 1$ imply
\begin{align}
\begin{split}
\sqrt{F}(\tilde{\rho}^k_T,\tilde\rho_T^0)
&\geq \sum_{z_T} 
\sqrt{\Pd^i(z_T)\Qd^i(z_T)}
\sqrt{F}(\tilde\rho(z_T), \tilde{\psi}(z_T))
\\
&\geq \sum_{z_T} 
\sqrt{\Pd^i(z_T)\Qd^i(z_T)}
\sqrt{F}(\tilde\rho(z_T), \tilde{\psi}(z_T))\de(z_T)s(z_T).
\end{split}
\end{align}
Now we first extract the error terms that we get from the change in the sequences $\de$ and $s$. 
We again use the notation $z_{T|t}=(x_1,\ldots, x_t)$ for $t\leq T$ and
$z_T=(x_1,\ldots, x_T)$ and then we obtain
\begin{align}
\begin{split}
\sqrt{F}(\tilde{\rho}^k_T,\tilde\rho_T^0)
&\geq \sum_{z_T} 
\sqrt{\Pd^i(z_T)\Qd^i(z_T)}
\sqrt{F}(\tilde\rho(z_T), \tilde{\psi}(z_T))\de(z_T)s(z_T)
\\
&=
 \sum_{z_T} 
\sqrt{\Pd^i(z_T)\Qd^i(z_T)}
\sqrt{F}(\tilde\rho(z_T), \tilde\psi(z_T))\de(z_T)( s(z_{T|T-1})-h(z_T)  )
\\
&=
 \sum_{z_T} 
\sqrt{\Pd^i(z_T)\Qd^i(z_T)}
\sqrt{F}(\tilde\rho(z_T), \tilde\psi(z_T))\de(z_T)s(z_{T|T-1})  - E^1_T.
\\
&=-E^1_T+
 \sum_{z_T} 
\sqrt{\Pd^i(z_T)\Qd^i(z_T)}
\sqrt{F}(\tilde\rho(z_T), \tilde\psi(z_T))\de(z_{T|T-1})s(z_{T|T-1})  
\\
&\qquad - \frac{\Delta_i^2}{\eta^2} \sum_{z_T} 
\sqrt{\Pd^i(z_T)\Qd^i(z_T)}
\sqrt{F}(\tilde\rho(z_T), \tilde\psi(z_T))\lVert P_i\tilde\psi(z_{T|T-1})\rVert^2 \de(z_{T|T-1})s(z_{T|T-1})
\\
&= -E^1_T-E^2_T +
\sum_{z_T} 
\sqrt{\Pd^i(z_T)\Qd^i(z_T)}
\sqrt{F}(\tilde\rho(z_T), \tilde\psi(z_T))\de(z_{T|T-1})s(z_{T|T-1}) .
\end{split}
\end{align}
 Where the error terms $E_T^1$ and $E_T^2$ are defined by those equations, i.e.,
 \blue{
 \begin{align}
     E^1_T& =
      \sum_{z_T} 
\sqrt{\Pd^i(z_T)\Qd^i(z_T)}
\sqrt{F}(\tilde\rho(z_T), \tilde\psi(z_T))\de(z_T)h(z_T),
\\
E_2^T &= \frac{\Delta_i^2}{\eta^2} \sum_{z_T} 
\sqrt{\Pd^i(z_T)\Qd^i(z_T)}
\sqrt{F}(\tilde\rho(z_T), \tilde\psi(z_T))\lVert P_i\tilde\psi(z_{T|T-1})\rVert^2 \de(z_{T|T-1})s(z_{T|T-1}).
 \end{align}
 }
We continue to estimate the remaining term using \eqref{eq:fid_loss_summed_co} in Corollary~\ref{co:fid_loss_one_step}.
Here we use that either the factor $s(z_{T|T-1})=0$  vanishes 
and the inequality below is trivially true (both sides are zero)  or the bound
$\sqrt{F}(\tilde{\rho}(z_{T|T-1}),\tilde{\psi}(z_{T|T-1}))\geq 1/2$ holds
so that Corollary~\ref{co:fid_loss_one_step} can be applied
\begin{align}
\begin{split}
 &\sum_{z_T} 
\sqrt{\Pd^i(z_T)\Qd^i(z_T)}
\sqrt{F}(\tilde\rho(z_T), \tilde{\psi}_T(z_T))\de(z_{T|T-1})s(z_{T|T-1})
\\
&\geq 
\sum_{z_{T-1}} 
\sqrt{\Pd^i(z_{T-1})\Qd^i(z_{T-1})}\de(z_{T-1})s(z_{T-1})\Bigg(
\sqrt{F}(\tilde{\rho}(z_{T-1}), \tilde{\psi}(z_{T-1}))
\\
&\qquad - \frac{5\Delta_i}{\sqrt{2}\eta}  
\lVert P_i\tilde{\psi}(z_{T-1})\rVert \left( R(z_{T-1})-\max_{x_t} R((z_{T-1}, x_t))\right)^{\frac12}
-  \frac{4\Delta_i^2}{\eta^2}\lVert P_i \tilde{\psi}(z_{T-1})\rVert^2\Bigg)
 \\
 &\geq 
 \sum_{z_{T-1}} 
\sqrt{\Pd^i(z_{T-1})\Qd^i(z_{T-1})}
\sqrt{F}(\tilde{\rho}(z_{T-1}), \tilde{\psi}(z_{T-1}))\de(z_{T-1})s(z_{T-1})
- E_{T}^3 -E_{T}^4
\end{split}
\end{align}
where we again define the error terms $E_{T}^3$ and $E_{T}^4$ implicitly through these equations, i.e.,
\blue{
\begin{align}
    E^3_T&=\frac{5\Delta_i}{\sqrt{2}\eta} \sum_{z_{T-1}} 
\sqrt{\Pd^i(z_{T-1})\Qd^i(z_{T-1})}\de(z_{T-1})s(z_{T-1})
\lVert P_i\tilde{\psi}(z_{T-1})\rVert \left( R(z_{T-1})-\max_{x_t} R((z_{T-1}, x_t))\right)^{\frac12},
\\
E^4_T &= \frac{4\Delta_i^2}{\eta^2} \sum_{z_{T-1}} 
\sqrt{\Pd^i(z_{T-1})\Qd^i(z_{T-1})}\de(z_{T-1})s(z_{T-1})
\lVert P_i \tilde{\psi}(z_{T-1})\rVert^2\Bigg).
\end{align}
}
 Applying this inductively together with $\sqrt{F}(\tilde{\rho}_0^i,\tilde{\rho}_0^0)=1$ we get
 \begin{align}\label{eq:error_ineq}
\sqrt{F}(\tilde\rho^i_T,\tilde\rho_T^0)
\geq  \sum_{z_T} 
\sqrt{\Pd^i(z_T)\Qd^i(z_T)}
\sqrt{F}(\tilde\rho(z_T), \tilde{\psi}(z_T))\de(z_T)s(z_T)
\geq 1 - \sum_{t=1}^T
E_t^1+E_t^2 + E_t^3+E_t^4
 \end{align}

\paragraph{Bounding the error terms $E^2_t$ and $E^4_t$.}
\blue{
\begin{lemma}\label{le:bound_e2_e4}
The error terms $E^2_t$ and $E^4_t$ satisfy the bound
\begin{align}
\begin{split}\label{eq:boundE2_E4}
\sum_{t=1}^T E^2_t+E^4_t 
&\leq \frac{5\Delta_i}{\eta}
\left(\sum_{t=1}^T \tr(P_i \tilde{\rho}^0_t)
\right)^{\frac12}.
\end{split}
\end{align} 
\end{lemma}
}
\begin{proof}
We start to bound $E^2_t$. We find
 using $\sqrt{F}\leq 1$, $s(z_t)\leq 1$
 \begin{align}
 \begin{split}
 E^2_t&\leq 
 \frac{\Delta_i^2}{\eta^2} \sum_{z_{t-1}} 
\sqrt{\Pd^i(z_{t-1})\Qd^i(z_{t-1})}\lVert P_i\tilde\psi(z_{t-1})\rVert^2 \de(z_{t-1})\sum_{x_t}
\sqrt{\Pd^i(x_t|z_{t-1})\Qd^i(x_t|z_{t-1})}
\\
&\leq 
 \frac{\Delta_i^2}{\eta^2} \sum_{z_{t-1}} 
\sqrt{\Pd^i(z_{t-1})\Qd^i(z_{t-1})}\lVert P_i\tilde\psi(z_{t-1})\rVert^2 \de(z_{t-1})
\end{split}
 \end{align}
Combining this with the definition of $E^4_t$ 
and using again $s(z_{t-1})\leq1$ we obtain
 using Cauchy Schwarz 
\begin{align}
\begin{split}\label{eq:boundE2a}
\sum_{t=1}^T E^2_t+E^4_t &\leq \frac{5 \Delta_i^2}{\eta^2}
\sum_{t=1}^T \sum_{z_{t-1}} \sqrt{\Pd^i(z_{t-1})\Qd^i(z_{t-1})} \lVert P_i\tilde\psi(z_{t-1})\rVert^2 \de(z_{t-1}) 
\\
&\leq \frac{5\Delta_i}{\eta}
\left(\sum_{t=1}^T \sum_{z_{t-1}}\Qd^i(z_{t-1})\lVert P_i\tilde\psi(z_{t-1})\rVert^2 
\right)^{\frac12}
\left(\sum_{t=1}^T \sum_{z_{t-1}}\Pd^i(z_{t-1})
\frac{\Delta_i^2}{\eta^2}\lVert P_i\tilde\psi(z_{t-1})\rVert^2 
\de(z_{t-1})
\right)^{\frac12}.
\end{split}
\end{align} 
For the second factor we apply the definition \eqref{eq:rel_r}
combined with $\sum \P^i(x_t)=1$ to rewrite the expression, and we find
\begin{align}
\begin{split}\label{eq:boundE2b}
\sum_{t=1}^T E^2_t+E^4_t 
&\leq 
 \frac{5\Delta_i}{\eta}
\left(\sum_{t=1}^T \sum_{z_{t-1}}\Qd^i(z_{t-1}) \tr(P_i \ket{\tilde{\psi}(z_{t-1})}
\bra{\tilde{\psi}(z_{t-1})}P_i)
\right)^{\frac12}
\\
&
\hspace{10em}
\left(\sum_{t=1}^T \sum_{z_{t-1}}\Pd^i(z_{t-1})
\left(\de(z_{t-1}) -\sum_{x_t} \Pd^i(x_t)\de((z_{t-1},x_t))\right)
\right)^{\frac12}
\\
&\leq 
 \frac{5\Delta_i}{\eta}
\left(\sum_{t=1}^T \tr(P_i \tilde{\rho}^0_t)
\right)^{\frac12}
\left(\sum_{t=1}^T \sum_{z_{t-1}}\Pd^i(z_{t-1})\de(z_{t-1})
-\sum_{t=1}^T \sum_{z_{t}}\Pd^i(z_{t})\de(z_{t})
\right)^{\frac12}
\\
&\leq \frac{5\Delta_i}{\eta}
\left(\sum_{t=1}^T \tr(P_i \tilde{\rho}^0_t)
\right)^{\frac12}\left( \de(z_0)-\sum_{z_T} \Pd^i(z_{T})\de(z_{T})\right)^{\frac12}
\\
&\leq \frac{5\Delta_i}{\eta}
\left(\sum_{t=1}^T \tr(P_i \tilde{\rho}^0_t)
\right)^{\frac12}
\end{split}
\end{align} 
 where we used Lemma~\ref{le:decomp_rho_0} for the first factor and \eqref{eq:trac0tilde} and the telescopic sum together with $\de(z_0)=1$ and $0\leq \de(z_T)\leq1$ in the last steps for the second factor.
\end{proof}

 \paragraph{Bounding the error term $E^t_3$.}
 \blue{
 \begin{lemma}\label{le:e3}
     The error term $E^3_t$ is bounded by
\begin{align}
\begin{split}\label{eq:boundE3_le}
\sum_{t=1}^{T} E^3_t
&\leq \frac{5\Delta_i}{\sqrt{2}\eta}
\left(\sum_{t=0}^{T-1} \tr(P_i \tilde{\rho}^0_t)
\right)^{\frac12}.
\end{split}
\end{align} 
 \end{lemma}
 }
 We bound, using again Cauchy Schwarz and $s(z_t)d(z_t)\leq 1$,
\begin{align}
\begin{split}\label{eq:boundE3}
\sum_{t=1}^{T} E^3_t
&= \frac{5\Delta_i}{\sqrt{2}\eta}
\sum_{t=0}^{T-1}\sum_{z_{t}} 
\sqrt{\Pd^i(z_{t})\Qd^i(z_{t})}\de(z_{t})s(z_{t})  
\lVert P_i\tilde{\psi}(z_{t})\rVert \left( R(z_{t})-\max_{x_{t+1}} R((z_{t}, x_{t+1}))\right)^{\frac12}
\\
&\leq 
\frac{5\Delta_i}{\sqrt{2}\eta}
\left(\sum_{t=0}^{T-1}\sum_{z_{t}} 
\Qd^i(z_{t}) 
\lVert P_i\tilde{\psi}(z_{t})\rVert^2\right)^{\frac12}
\left(\sum_{t=0}^{T-1}\sum_{z_{t}} \Pd^i(z_t) \left( R(z_{t})-\max_{x_{t+1}} R((z_{t}, x_{t+1}))\right)\right)^{\frac12}
\\
&\leq
\frac{5\Delta_i}{\sqrt{2}\eta}
\left(\sum_{t=0}^{T-1} \tr(P_i \tilde{\rho}^0_t)
\right)^{\frac12}
\left(\sum_{t=0}^{T-1}\sum_{z_{t}} \Pd^i(z_t) R(z_{t})-
\sum_{t=0}^{T-1}\sum_{z_{t}, x_{t+1}} \Pd^i(z_t)\Pd^i(x_{t_1}) R((z_{t},x_{t_1})\right)^{\frac12}
\\
&\leq
\frac{5\Delta_i}{\sqrt{2}\eta}
\left(\sum_{t=0}^{T-1} \tr(P_i \tilde{\rho}^0_t)
\right)^{\frac12}
\left(\sum_{t=0}^{T-1}\sum_{z_{t}} \Pd^i(z_t) R(z_{t})-
\sum_{t=1}^{T}\sum_{z_{t}} \Pd^i(z_t) R(z_{t})\right)^{\frac12}
\\
&\leq \frac{5\Delta_i}{\sqrt{2}\eta}
\left(\sum_{t=0}^{T-1} \tr(P_i \tilde{\rho}^0_t)
\right)^{\frac12}.
\end{split}
\end{align} 
 Here we applied Lemma~\ref{le:decomp_rho_0} similar to the previous lemma, and we once more used the telescopic sum and the fact that $0\leq R(z_t)\leq 1$.

\paragraph{Bounding the error term $E^1_t$.}
It remains to bound the last remaining error term $E^1_t$ where we can prove the following bound.
\blue{
\begin{lemma}\label{le:e1}
    The error term $E^1_t$ can be bounded as follows
        \begin{align}
\sum_{t=1}^TE^1_t\leq  \sum_{t=1}^T
E_t^2 + E_t^3+E_t^4.
\end{align}
\end{lemma}
}
\begin{proof}
The key ingredient to bound $E^1_t$ is the observation
that by definition of $s(z_t)$ and $h(z_t)$ we have $\sqrt{F}(\tilde{\rho}(z_t), \tilde{\psi}(z_t))<1/2$ if $h(z_t)=1$ 
which implies 
\begin{align}
\sqrt{F}(\tilde{\rho}(z_t), \tilde{\psi}(z_t))<1 - \sqrt{F}(\tilde{\rho}(z_t), \tilde{\psi}(z_t)).
\end{align}
From here we conclude using the definition of $E_t^1$
(and $h(z_t)=0$ if $h(z_t)\neq 1$)
\begin{align}
\begin{split}
\sum_{t+1}^T E^1_t
&=\sum_{t=1}^T\sum_{z_t} 
\sqrt{\Pd^i(z_t)\Qd^i(z_t)}
\sqrt{F}(\tilde\rho(z_t), \tilde\psi(z_t))\de(z_t)h(z_t)
\\
&
\leq\sum_{t=1}^T\sum_{z_t} 
\sqrt{\Pd^i(z_t)\Qd^i(z_t)}
\left(1-\sqrt{F}(\tilde\rho(z_t), \tilde\psi(z_t))\right)\de(z_t)h(z_t)
\\
&\leq 
\sum_{t=1}^T\sum_{z_t} 
\sqrt{\Pd^i(z_t)\Qd^i(z_t)}h(z_t)
-\sum_{t=1}^T E_t^1.
\end{split}
\end{align}
We now find using \eqref{eq:error_ineq} and the last display that
\begin{align}
\begin{split}\label{eq:error_reversed}
\sum_{z_T} 
\sqrt{\Pd^i(z_T)\Qd^i(z_T)}s(z_T)
&\geq 
\sum_{z_T} 
\sqrt{\Pd^i(z_T)\Qd^i(z_T)}
\sqrt{F}(\tilde\rho(z_T), \tilde{\psi}(z_T))\de(z_T)s(z_T)
\\
&\geq 1 - \sum_{t=1}^T
\left(E_t^1+E_t^2 + E_t^3+E_t^4\right)
\\
&\geq 1 - \sum_{t=1}^T
\left(E_t^2 + E_t^3+E_t^4\right)+\sum_{t=1}^TE^1_t-
\sum_{t=1}^T\sum_{z_t} 
\sqrt{\Pd^i(z_t)\Qd^i(z_t)}h(z_t).
\end{split}
\end{align}
Now it is easy to see that 
\begin{align}
\sum_{t=1}^T\sum_{z_t} 
{\Pd^i(z_t)}h(z_t)
+
\sum_{z_T} 
{\Pd^i(z_T)}s(z_T)= 1
\end{align}
because this is the probability that the fidelity drops below $1/2$ under
the measure $\Pd^i$ plus the probability that it is at least $1/2$.
Clearly, a similar relation holds for $\Qd^i$.
Now we apply $\sqrt{\Pd^i(z_T)\Qd^i(z_T)}\leq (\Pd^i(z_T)+\Qd^i(z_T))/2$
and find
\begin{align}
\begin{split}
\sum_{z_T}& 
\sqrt{\Pd^i(z_T)\Qd^i(z_T)}s(z_T)+\sum_{t=1}^T\sum_{z_t} 
\sqrt{\Pd^i(z_t)\Qd^i(z_t)}h(z_t)
\\
&\leq 
\frac12\left(\sum_{z_T} 
\Pd^i(z_T)s(z_T)+\sum_{t=1}^T\sum_{z_t} 
\Pd^i(z_t)h(z_t)+\sum_{z_T} 
\Qd^i(z_T)s(z_T)+\sum_{t=1}^T\sum_{z_t} 
\Qd^i(z_t)h(z_t)\right)=
1.
\end{split}
\end{align}

Using this relation in \eqref{eq:error_reversed} we find
\begin{align}
1\geq 1 - \sum_{t=1}^T\left(
E_t^2 + E_t^3+E_t^4\right)+\sum_{t=1}^TE^1_t
\end{align}
from which 
\begin{align}\label{eq:E1_final}
\sum_{t=1}^TE^1_t\leq  \sum_{t=1}^T
\left(E_t^2 + E_t^3+E_t^4\right)
\end{align}
follows.
\end{proof}
\paragraph{Conclusion.}
To finish the proof of Proposition~\ref{prop:key} we only have to plug all the derived bounds in \eqref{eq:error_ineq}.
Indeed, applying Lemma~\ref{le:bound_e2_e4}, Lemma~\ref{le:e3},
and Lemma~\ref{le:e1} we get
\begin{align}
\begin{split}
\sqrt{F}(\rho^i_T,\rho_T^0)
&\geq 1 - \sum_{t=1}^T \left(E_t^1
+E_t^2+E_t^3+E_t^4\right)
\geq 1 - 2 \sum_{t=1}^T \left(E_t^2+E_t^3+E_t^4\right)
\\
&\geq 1 - 20\eta^{-1}\Delta_i\left(\sum_{t=1}^T \tr P_i \tilde{\rho}_t^0\right)^{\frac12}.
\end{split}
\end{align}
This ends the proof.
\end{proof}

\section{Proof of Theorem~\ref{th:lower_emp}}\label{app:lower_emp}
Here we prove the lower bound in Theorem~\ref{th:lower_emp}.
Let us first give a precise statement of the result.
We consider the same probability vectors $\pp_i$ as introduced at the beginning of 
Section~\ref{sec:overview}.
First, we note that the result does not hold for fixed oracles $O_i$ with reward
vector $\pp_i$ because the algorithm could exploit the specific structure of the oracles. 
There could be, e.g., a state $\omega_0$ such that
$O(\pp_i)\ket{j}\ket{\omega_0}\ket{0}=\ket{j}\ket{\omega_0}\ket{\delta_{ij}}$.
Then the problem reduces to the unstructured search problem when $\omega_0$ is known.
Thus, the result only holds when we assume that $O_i$ is a random oracle
with the fixed reward vector $\pp_i$, emulating the situation where we have no additional information about the oracles. 
We consider for a reward vector $r\in \{0,1\}^{|\mc{H}_A|\cdot |\mc{H}_P|}$ the oracle
$O^r$ acting as in \eqref{eq:def_oracle_general}, i.e., 
\begin{align}
O^r\ket{i}\ket{\omega}\ket{c}=\ket{i}\ket{\omega}\ket{c+r_i(\omega)}.
\end{align}
We consider a random reward distribution $r(i)$ where $r(i)$ is the uniform distribution over
all reward vectors with mean reward vector $\pp_i$, i.e., $|\mc{H}_P|^{-1}\sum_\omega r(i)_j(\omega)=(\pp_i)_j$.
 Then a more precise version of Theorem~\ref{th:lower_emp} reads as follows.
\begin{theorem}
Let $\delta<1/2$. 
Assume that $\pp^j$ are as before with $p_i\in [\eta,1-\eta]$ for
some $\eta>0$.
Any algorithm that identifies the best arm with probability at least $1-\delta$ 
given an oracle $O^{r(i)}$
where $r(i)$ is distributed as above 
requires at least 
\begin{align}
T\geq\frac{1}{2} \eta (1-2\sqrt{\delta(1-\delta)}) \left(\sum_{i=2}^n \Delta_i^{-2}\right)^{\frac12}
 \geq c(\delta,\eta)\sqrt{H(\pp^1)}
\end{align}
calls to the oracle.
This is still true even when 
 it is known that the  vector of mean rewards is
$\{\pp^0,\ldots,\pp^n\}$.
\end{theorem}
\begin{proof}

The proof is close to the proof of Theorem~\ref{th:faulty_grover}.
We assume we are given any algorithm  acting by
$ (\mc{E}_O \otimes \id)\circ \mc{E}_{U_T}
\circ \ldots \circ (\mc{E}_O\otimes \id)\circ \mc{E}_{U_1}$
where $U_t$ are arbitrary unitary maps where $O$ denotes the given oracle.
Assume that the initial state is a fixed density matrix $\rho$.
We denote the  state using the oracle $O^{r(i)}$ before the 
$t+1$-th invocation of the oracle by $\rho_t^{r(i)}$, i.e., $\rho_0^{r(i)}=\mc{E}_{U_1}(\rho)$.
We introduce the notation $\mathbb{E}_i$ when we average over the reward distribution
$r(i)$ and we write 
\begin{align}\label{eq:mean_density}
\rho_t^{i}=\mathbb{E}_i \rho_t^{r(i)}.
\end{align}
By assumption we can identify $i$ given $\rho_T^i$ with probability at least $1-\delta$.
This implies, as in \eqref{eq:fidelity_classic_difference}, \eqref{eq:fidelity_delta} for $i>1$
\begin{align}\label{eq:upper_bound_fidelity}
2\sqrt{\delta(1-\delta)}\geq 
\sqrt{F}(\rho_T^{i},\rho_T^{0}).
\end{align}
One main ingredient of the proof is to define a suitable coupling 
of the random variables $r(0)$ and $r(i)$. Note that its reward vectors are
$\pp_0$ and $\pp_i$ such that $(\pp_i)_j=p_j=(\pp_0)_j$ for $j\neq i$
and $(\pp_i)_i=p_0$, $(\pp_0)_i=p_i$. We now consider a coupling where 
$r(0)_j=r(i)_j$ for $i\neq j$ and $r(i)_i\geq r(0)_i$ and 
the distribution of $r(i)_i\in \{0,1\}^{|\mc{H}_P|}$ is uniform over all 
rewards under this constraint for a fixed $r(0)$.
Note that this entails that the distribution of $r(0)$ for a fixed $r(i)$
is also uniform over all rewards with the right mean reward satisfying $r(0)\leq r(i)$.
 It is straightforward to see that such a 
coupling exists (a possible explicit construction is to draw i.i.d.\ random numbers $u_i(\omega)$ and set $r(0)_i(\omega)=1$ iff $u_i(\omega)$ is in the $p_i$-th quantile
of the numbers $u_i(\cdot)$ and similarly for $r(i)$). We denote this coupling by $p_i(r(i), r(0))$. 
Observe that for this coupling satisfies, for all $\omega$,
\begin{align}\label{eq:cond1}
\mathbb{P}(r(i)_j(\omega)=1| r(0)_j(\omega)=0)=\mathbb{P}(r(i)_j=1|r(0)_j=0)
= \begin{cases}
0\quad &\text{for $j\neq i$}
\\
\frac{p_0-p_i}{1-p_i}=\frac{\Delta_i}{1-p_i}\quad &\text{for $j=i$}.
\end{cases}
\end{align}
Similarly, we get for all $\omega$
\begin{align}\label{eq:cond2}
\mathbb{P}(r(0)_j(\omega)=0| r(i)_j(\omega)=1)=\mathbb{P}(r(0)_j=0|r(i)_j=1)
= \begin{cases}
0\quad &\text{for $j\neq i$}
\\
\frac{p_0-p_i}{p_0}=\frac{\Delta_i}{p_0}\quad &\text{for $j=i$}.
\end{cases}
\end{align}
We can bound using the concavity of the fidelity
\begin{align}
\sqrt{F}(\rho_t^{i},\rho_t^{0})
\geq 
\sum_{r(i),r(0)} p_i(r(i),r(0)) \sqrt{F}( \rho_t^{r(i)}, \rho_t^{r(0)}).
\end{align}
Now we lower bound the fidelity terms on the right-hand side based on our construction of the coupling.
By construction we have $r(i)\geq r(0)$ which implies that 
$r(i)-r(0)\in \{0,1\}^{|\mc{H}_A|\cdot |\mc{H}_P|}$. Note that
\begin{align}\label{eq:rel_oracles_Or}
\begin{split}
O^{r(i)}O^{r(0)}\ket{j, \omega, c}
&= \ket{j, \omega, c+(r(i))_j(\omega)+(r(0))_j(\omega)}
= \ket{j, \omega, c+(r(i))_j(\omega)-(r(0))_j(\omega)}
\\
&=
O^{r(i)-r(0)}\ket{j, \omega, c}.
\end{split}
\end{align}
Let us introduce the self adjoint projection $P^{r(i),r(0)}$ given by
\begin{align}
P^{r(i),r(0)}\ket{j, \omega, c}=\bs{1}_{(r(i))_j(\omega)\neq (r(0))_j(\omega)} \ket{j, \omega, c}.
\end{align}
Note that then 
\begin{align}\label{eq:proj_inv}
(\id - P^{r(i),r(0)})O^{r(i)-r(0)}\ket{j, \omega, c} 
=(\id - P^{r(i),r(0)})\ket{j, \omega, c} .
\end{align}
We get, using the invariance of the fidelity under unitary transformations,
$(O^{r(i)})^2=\id$,
and \eqref{eq:rel_oracles_Or}
\begin{align}
\begin{split}
\sqrt{F}( \rho_{t+1}^{r(i)},\rho_{t+1}^{r(0)})
&=
\sqrt{F}(\mc{E}_{O^{r(i)}}(  \rho_t^{r(i)}), \mc{E}_{O^{r(0)}}(\rho_t^{r(0)}))
=
\sqrt{F}(  \rho_t^{r(i)}, \mc{E}_{O^{r(i)}}\circ \mc{E}_{O^{r(0)}}(\rho_t^{r(0)}))
\\
&=
\sqrt{F}( \rho_t^{r(i)}, \mc{E}_{O^{r(i)-r(0)}}(\rho_t^{r(0)})).
\end{split}
\end{align}
Next we apply Lemma~\ref{le:fidelity_simple} to $O^{r(i)-r(0)}$ and
$P^{r(i),r(0)}$. The assumptions of the lemma are satisfied because 
\eqref{eq:proj_inv} and all $P^{r(i),r(0)}$ and $O^r$ commute.
Lemma~\ref{le:fidelity_simple} together with the last display imply
\begin{align}
\sqrt{F}( \rho_{t+1}^{r(i)},\rho_{t+1}^{r(0)})
\geq \sqrt{F}  (\rho_t^{r(i)},\rho_t^{r(0)})
-2\sqrt{\tr(P^{r(i),r(0)}\rho_t^{r(i)}) \tr(P^{r(i),r(0)}\rho_t^{r(0)})}
\end{align}
We obtain
\begin{align}
\sqrt{F}( \rho_{T}^{i},\rho_T^{0})
\geq 1 - 2\sum_{r(i),r(0)}\sum_t p_i(r(i),r(0)) \sqrt{\tr(P^{r(i),r(0)}\rho_t^{r(i)}) \tr(P^{r(i),r(0)}\rho_t^{r(0)})}.
\end{align}
Using \eqref{eq:upper_bound_fidelity} followed by Cauchy-Schwarz we get
\begin{align}
\begin{split}\label{eq:upper2}
&\frac12-\sqrt{\delta(1-\delta)}\leq \sum_{r(i),r(0)}\sum_t p_i(r(i),r(0)) \sqrt{\tr(P^{r(i),r(0)}\rho_t^{r(i)}) \tr(P^{r(i),r(0)}\rho_t^{r(0)})}
\\
&\leq
\left(\sum_{t,r(i),r(0)} p_i(r(i),r(0)) \tr(P^{r(i),r(0)}\rho_t^{r(i)})\right)^{\frac12} \left(\sum_{t,r(i),r(0)} p_i(r(i),r(0)) \tr(P^{r(i),r(0)}\rho_t^{r(0)})\right)^{\frac12}.
\end{split}
\end{align}
We observe that by construction of the coupling and \eqref{eq:cond1} we have 
\begin{align}
\begin{split}
\sum_{r(i)} p_i(r(i), r(0)) P^{r(i),r(0)}\ket{j, \omega, c}
&=
\sum_{r(i)}
p_i(r(i), r(0)) 
\bs{1}_{(r(i))_j(\omega)\neq (r(0))_j(\omega)} \ket{j, \omega, c}
\\
&= p(r(0)) \mathbb{P}((r(i))_j(\omega)=1| r(0))\ket{j, \omega, c}
\\
&=
p(r(0)) \bs{1}_{i=j}\bs{1}_{r(0)_i(\omega)=0} \frac{\Delta_i}{1-p_i}\ket{j, \omega, c}.
\end{split}
\end{align}
And similarly, using \eqref{eq:cond2}
\begin{align}
\begin{split}
\sum_{r(0)} p_i(r(i), r(0)) P^{r(i),r(0)}\ket{j, \omega, c}
&=
p(r(i)) \bs{1}_{i=j}\bs{1}_{r(i)_i(\omega)=1} \frac{\Delta_i}{p_0}\ket{j, \omega, c}.
\end{split}
\end{align}
As before, we define $P_i$ to be the projection on arm $i$, i.e., 
$P_i\ket{j, \omega,c}=\delta_{ij}\ket{j, \omega,c}$.
We conclude that 
\begin{align}
\sum_{r(i)} p_i(r(i),r(0)) \tr(P^{r(i),r(0)}\rho_t^{r(0)})
&\leq  p(r(0))\frac{\Delta_i}{\eta} \tr( P_i\rho_t^{r(0)}),
\\
\sum_{r(0)} p_i(r(i),r(0)) \tr(P^{r(i),r(0)}\rho_t^{r(i)})
&\leq  p(r(i))\frac{\Delta_i}{\eta} \tr( P_i\rho_t^{r(i)}).
\end{align}
Combining this with \eqref{eq:upper2} and \eqref{eq:mean_density} we obtain
\begin{align}
\begin{split}
\frac12-\sqrt{\delta(1-\delta)}&\leq 
\left(\frac{\Delta_i}{\eta}\sum_{t,r(i)} p(r(i))  \tr( P_i\rho_t^{r(i)})\right)^{\frac12} \left(\frac{\Delta_i}{\eta}\sum_{t,r(0)} p(r(0)) \tr( P_i\rho_t^{r(0)}))\right)^{\frac12}
\\
&= \frac{\Delta_i}{\eta}
\left(\sum_{t}  \tr( P_i\rho_t^{i})\right)^{\frac12} \left(\sum_{t} \tr( P_i\rho_t^{0})\right)^{\frac12}
\leq \frac{\Delta_i\sqrt{T}}{\eta}
\left(\sum_{t} \tr( P_i\rho_t^{0})\right)^{\frac12}.
\end{split}
\end{align}
We square this relation divide by $\Delta_i^2$ and sum over $i>1$ and get
\begin{align}
\begin{split}
\left(\frac12-\sqrt{\delta(1-\delta)}\right)^2\sum_{i>1}\Delta_i^{-2}&\leq \frac{T}{\eta^2}
\sum_i\sum_{t} \tr( P_i\rho_t^{0})=\frac{T}{\eta^2}\sum_t \tr(\rho_t^0)=\frac{T^2}{\eta^2}.
\end{split}
\end{align}
This implies 
\begin{align}
T\geq \frac{1}{2}\eta (1-2\sqrt{\delta(1-\delta)})\sqrt{\sum_{i>1} \Delta_i^{-2}}.
\end{align}
\end{proof}

\section{Proof of Theorem~\ref{th:speedup}}\label{app:speedup}
\begin{proof}
We assume that we work on a Hilbert space $\mc{H}=\mc{H}_I\otimes \mc{H}_S \otimes \mc{H}_A$
where $\mc{H}_I$ is the input space, $\mc{H}_S$ is a single qubit state space and $\mc{H}_A$ consists of $T$ ancilla qubits where $T=\mc{O}(\sqrt{N}/p)$ will be defined below in \eqref{eq:defT}. 
We assume that the initial state of $\mc{H}_I$ is $\ket{s}=\sqrt{N}^{-1} \sum \ket{i}$ and all remaining qubits are in state $\ket{0}$.
The algorithm consists of applying in turn the three operators $\mc{F}$, a controlled Grover-diffusion
\begin{align}
U_\omega = \id\otimes \ket{0}\bra{0}\otimes \id + (2\ket{s}\bra{s}-\id)\otimes \ket{1}\bra{1}\otimes \id, 
\end{align}
and a swap operation $S_t$ that swaps the state qubit with the $t$-th ancilla qubit.
We can rewrite this concisely as
\begin{align}
(\mc{E}_{S_T}\circ \mc{E}_{U_\omega} \circ \mc{F}) \circ \ldots \circ (\mc{E}_{S_1}\circ \mc{E}_{U_\omega} \circ \mc{F})(\rho).
\end{align}
We denote by $\tilde{O}_i$ the standard oracle on $\mc{H}_I$ and denote by 
 $G= (2\ket{s}\bra{s}-\id)\tilde{O}_i$ the map from Grover's algorithm acting on  $\mc{H}_I$.
 We claim that the final state of the algorithm is 
 \begin{align}\label{eq:state_rho_T}
 \rho_T = \sum_{x\in \{0,1\}^T} p^{|x|_1}(1-p)^{T-|x|_1} \ket{G^{|x|_1}s, 0, x, 0^{A-T}}\bra{G^{|x|_1}s, 0, x, 0^{A-T}}.
 \end{align} 
 The proof of this is by induction, indeed, we note
 \begin{align}
 \begin{split}
 &(\mc{E}_{S_T}\circ \mc{E}_{U_\omega} \circ \mc{E}) (  \ket{G^{|x|_1}s,0,x,0^{A-T}}\bra{G^{|x|_1}s,0,x,0^{A-T}}) 
 \\
 &=
 (\mc{E}_{S_T}\circ \mc{E}_{U_\omega} ) \left( p \ket{\tilde{O}G^{|x|_1}s, 1, x, 0^{A-T}}\bra{\tilde{O}G^{|x|_1}s, 1, x, 0^{A-T}}\right.
 \\
 &\hspace{3cm}\left.+(1-p) \ket{G^{|x|_1}s, 0, x, 0^{A-T}}\bra{G^{|x|_1}s, 0, x, 0^{A-T}}\right)
 \\
 &=
 \mc{E}_{S_T} \left( p \ket{GG^{|x|_1}s, 1, x, 0^{A-T}}\bra{GG^{|x|_1}s, 1, x, 0^{A-T}}\right.
 \\
 &\hspace{3cm}
 \left.+(1-p) \ket{G^{|x|_1}s, 0, x, 0^{A-T}}\bra{G^{|x|_1}s, 0, x, 0^{A-T}}\right)
 \\
  &=
  p \ket{GG^{|x|_1}s, 0, x,1, 0^{A-T-1}}\bra{GG^{|x|_1}s, 0, x, 1,0^{A-T-1}}
 \\
 &\hspace{3cm}
 +(1-p) \ket{G^{|x|_1}s, 0, x, 0^{A-T}}\bra{G^{|x|_1}s, 0, x, 0^{A-T}}
 \end{split}
 \end{align}
 which implies \eqref{eq:state_rho_T}.
 The reduced density matrix on $\mc{H}_I$ is 
 \begin{align}
 \rho^I_T = 
 \sum_{k=0}^T p^k(1-p)^{T-k} \ket{G^k s}\bra{G^k s}. 
 \end{align}
The classical analysis of the Grover algorithm proves
\begin{align}
G^k \ket{s} = \cos\left(\frac{2k+1}{2}\theta\right)\ket{s'} + \sin\left(\frac{2k+1}{2}\theta\right)\ket{i} 
\end{align}
where $s'=\sqrt{N-1}^{-\tfrac12} \sum_{j\neq i}\ket{i}$ and 
\begin{align}
\theta=2\arccos(\sqrt{(N-1)/N})=2\arcsin(\sqrt{N}^{-1}).
\end{align}
We remark that $\theta/2 \approx \sqrt{N}^{-1}$ as $N\to \infty$. 
Note that for 
\begin{align}
k\in I= (\pi/(4\theta)-1/2, 3\pi/(4\theta)-1/2)
\end{align}
we have $\sin\left(\frac{2k+1}{2}\theta\right)^2 \geq 1/2$.
Let 
\begin{align}\label{eq:defT}
T=\lfloor  \pi / (2\theta p) \rfloor.
\end{align}
 It remains to be shown that with probability at least $1/2$ a $\mathrm{Bin}(T,p)$ distributed variable is
contained in  the interval $I$.

Using that the variance of  $X\sim \mathrm{Bin}(T, p)$ is $Tp(1-p)$ we can bound
\begin{align}
\mathbb{P}(|X-pT|>pT/8)\leq \frac{64\mathbb{E}(|X-pT|^2}{p^2T^2}
\leq \frac{64}{pT}\leq \frac{256\theta}{\pi }\leq \frac12
\end{align}
for $N$ sufficiently large.
Moreover, $|X-pT|<pT/8$ implies 
\begin{align}
X\in \left(\frac{3pT}{8}, \frac{5pT}{8}\right)\subset
\left(\frac{3\pi}{8\theta}-1 , \frac{5\pi}{8\theta}\right)
\subset 
\left( \frac{\pi}{4\theta}-1/2, \frac{3\pi}{4\theta}-1/2\right)
\end{align}
 for $N$ sufficiently large. This ends the proof.

\end{proof}

\section{Analysis of reusable oracles}\label{app:reusable}
Here we sketch a proof of Theorem~\ref{th:reusable}.
The proof essentially relies on the algorithm given in \cite{multi_armed_quantum}. 
The only building block that needs to be changed is the gapped amplitude estimation 
(Corollary~2 in \cite{multi_armed_quantum}).
Let us explain this algorithm along with the replacement based on oracles as
in \eqref{eq:oracle_col}.
Gapped amplitude estimation assumes we have access to
an oracle $O_p$ and its adjoint acting via
\begin{align}
O_p\ket{0}=\sqrt{p}\ket{1}+\sqrt{1-p}\ket{0}=\ket{\mathrm{coin}_p}.
\end{align}
Then the following result holds.
\begin{lemma}[Corollary~2 in \cite{multi_armed_quantum}]
For $\eps>0$, $l\in [0,1]$ and $\delta>0$ there is a unitary procedure with $\bigO(\varepsilon^{-1}\ln(\delta))$ queries to $O_p$ that prepares the state
\begin{align}
\ket{\mathrm{coin}_p}(\alpha_0\ket{0}\ket{\psi_1}+\alpha_1\ket{1}\ket{\psi_2}
\end{align}
with $\alpha_0,\alpha_1\in [0,1]$ and  $\alpha_1\leq \delta$ if $p\leq l-2\varepsilon$
and $\alpha_0\leq \delta$ if $p\geq l-\varepsilon$.
\end{lemma}
We replace this by classical estimation of the mean. 
Using tail bounds of random variables, we obtain the following simple variant of Hoeffding's inequality.
\begin{lemma}\label{le:hoeffding}
Let $S_k$ be the sum of $k$ independent random variables with distribution $\mathrm{Ber}(p)$.
For  $\eps>0$ the bounds
\begin{align}
\P(S_k>k(p+\eps))\leq e^{-2\eps^2k}, \quad  \P(S_k<k(p-\eps))\leq e^{-2\eps^2k}
\end{align}
hold.
\end{lemma}
\begin{proof}
This is just Hoeffding's inequality.
\end{proof}
We then have the following corollary.
\begin{corollary}\label{co:algo}
Given access to oracles as in \eqref{eq:oracle_col} we can construct for any $\eps,\delta>0$ and $l\in [0,1]$ an algorithm $\mc{A}$ that maps with probability 
at least $1-\delta$ for all $1\leq i\leq n$
\begin{align}
\mc{A}\ket{i}\ket{0}=\ket{i}\ket{c}
\end{align}
where $c=1$ if $p_i<l-2\eps$  and $c=0$ if $p_i>l-\varepsilon$
and $\mc{A}$ requires $\varepsilon^{-2}$ oracle calls.
\end{corollary}
\begin{proof}
Set $k=\lceil 2\eps^{-2}\ln(n/\delta)\rceil$. Then we consider the sequence
\begin{align}
\ket{i}\ket{0}\ket{0}
\to \ket{i}\ket{\sum_{j=0}
^k X_i^t} \ket{0}\to 
\ket{i}\ket{\sum_{j=0}
^k X_i^t} \ket{\bs{1}_{\sum_{j=0}
^k X_i^t< k(l-3\eps/2)}}
\to \ket{i}\ket{0} \ket{\bs{1}_{\sum_{j=0}
^k X_i^t< k(l-3\eps/2)}}.
\end{align}
Where we in the first step sum the rewards using the oracles 
$O_{X^t}$, then set the flag register conditional on the sum and then uncompute the work register. This requires $2k$ oracle calls.
Using Lemma~\ref{le:hoeffding} we conclude that for $p_i>l-\eps$
\begin{align}
\P\left(\sum_{j=0}
^k X_i^t< l-3\eps/2\right)\leq e^{-2k(\eps/2)^2}\leq e^{-\ln(n/\delta)}=\frac{\delta}{n}
\end{align}
and a similar statement holds for $p_i<l-2\eps$.
The union bound implies the statement.
\end{proof}
Now we consider Theorem~\ref{th:reusable}. Giving a full proof would require a large amount of notation that is not worth the effort. Thus, we give a very brief sketch of the argument and leave the details to the reader. Their proof relies on variable time algorithms
\cite{var_time} which we also do not introduce here. For readers not familiar with them, the following proof shall merely serve as a heuristic.
\begin{proof}[Sketch of the proof of Theorem~\ref{th:reusable}]
We apply the algorithm constructed in \cite{multi_armed_quantum} but replace the gapped amplitude estimation by the algorithm constructed in Corollary~\ref{co:algo}.
The main strategy used in their proof is to construct a variable time algorithm  based on the gapped amplitude estimate that flags all arms whose reward is smaller than a given threshold. This allows to construct algorithms to count the number of flagged arms 
and rotate on the subspace of flagged arms using variable time amplitude amplification
\cite{var_time}. Those sub-routines can be used to first estimate 
$p_1$ and $p_2$ and then identify the best arm.

Their variable time algorithm  based on the gapped amplitude estimate
has query complexity on arm $i$ is at most $\Delta_i^{-1}\log(1/a)$ 
with $a$ polynomial in $\Delta_1 n$
and the $l^2$ averaged run-time
thus amounts to 
\begin{align}
t_{\mathrm{av}}^2
\leq C \frac{1}{n}\sum_{i=2}^n \Delta_i^{-2}\ln^2(1/a).
\end{align}
When relying on the algorithm from Corollary~\ref{co:algo}
we obtain the query complexity $\Delta_i^{-2}\log(n/\delta')$ for arm $i$
where $\delta'$ denotes the bound on the failure probability for the algorithm and the $l^2$ averaged run-time then amounts to 
\begin{align}
t_{\mathrm{av}}^2
\leq C \frac{1}{n}\sum_{i=2}^n \Delta_i^{-4}\ln^2(n\/\delta').
\end{align}
Then the variable time amplitude amplification gives an algorithm with success probability more than $1/2$ and query complexity
$t_{\mathrm{av}}/\sqrt{p_{\mathrm{succ}}}\ln(t_{\mathrm{max}})$ where $p_{\mathrm{succ}}$ denotes the success
probability and $t_{\mathrm{max}}$ the maximal complexity of the initial variable algorithm (there is another term $t_{\mathrm{max}}\ln(t_{\mathrm{max}})$ which is smaller in our case).

Since the initial success probability is $n^{-1}$ we obtain the query complexity bound
\begin{align}
T\leq {t_{\mathrm{av}}}{\sqrt{n}}\ln(t_m)
\leq C \left(\sum_{i=2}^n \Delta_i^{-4}\right)^{\frac12}\ln(n/\delta')\ln(t_{\mathrm{max}}).
\end{align}
We need to apply
a total of $\mc{O}(\ln(\Delta_1^{-1})$ such amplified variable time algorithms 
(essentially we perform binary search 
to find $l_L<l_R$ such that $p_2<l_L-\Delta_2/4$, $l_R+\Delta_2/4<p_1$, and $l_L+\Delta_2/4<l_R$).
To ensure that all constructed oracles as in Corollary~\ref{co:algo}
succeed with probability more than $1-\delta$ it is thus sufficient to pick $\delta'=c\delta/\ln(\Delta_1^{-1})$.
This together with  $t_{\mathrm{max}}=\bigOt(\Delta_1^{-2})$  ends the proof sketch.
\end{proof}

\end{document}